\documentclass[12pt,letterpaper,hidelinks]{article}
\newcommand{\pr}{{\rm{pr}}}
\newcommand{\purple}{\color{black}}
\newcommand{\E}{{{E}}}
\usepackage{graphicx}
\usepackage{epstopdf}
\usepackage{amsmath,amssymb,amsfonts,thumbpdf}
\usepackage{amsthm}
\usepackage{enumerate}
\usepackage{bm,bbm}
\usepackage{multirow}
\usepackage[authoryear]{natbib}
\usepackage{hyperref}
\usepackage{color}
\usepackage[plain,noend,ruled]{algorithm2e}
\usepackage{float,subcaption}
\usepackage{url, xcolor}
\usepackage[margin=1in]{geometry}

\graphicspath{{figures/}}

\newcommand{\red}{\color{black}}
\newcommand{\blue}{\color{black}}
\newcommand{\bluee}{\color{black}}

\newtheorem{lemma}{Lemma}{\bf}{\it}
\newtheorem{theorem}{Theorem}{\bf}{\it}
{\bf}{\it}

\theoremstyle{definition}
{\bf}{\it}
\newtheorem{assumption}{Assumption}{\bf}{\it}
{\bf}{\it}
\newtheorem{remark}{Remark}{\bf}{\it}

\newcommand{\blind}{1}

\begin{document}

\def\spacingset#1{\renewcommand{\baselinestretch}%
{#1}\small\normalsize} \spacingset{1}

\if1\blind
{
   \title{\bf Optimal Distributed Subsampling for Maximum Quasi-Likelihood Estimators with Massive Data}
  \author{ Jun Yu, HaiYing Wang, Mingyao Ai and Huiming Zhang \footnote{
  Jun Yu is an Assistant Professor, School of Mathematics and Statistics, Beijing Institute of Technology, Beijing 100081, China (yujunbeta@bit.edu.cn);
  HaiYing Wang is an Assistant Professor, Department of Statistics, University of Connecticut, Storrs, Mansfield, CT 06269, USA (haiying.wang@uconn.edu);
  Mingyao Ai is a Professor,  LMAM, School of Mathematical Sciences and Center for Statistical Science, Peking University, Beijing 100871, China (myai@math.pku.edu.cn);
  and Huiming Zhang is a PhD student, School of Mathematical Sciences and Center for Statistical Science, Peking University, Beijing 100871, China (zhanghuiming@pku.edu.cn).
   Correspondence should be addressed to Mingyao Ai.
  Yu's work was partially supported by Beijing Institute of Technology research fund program for young scholars. Wang's work was partially supported by NSF grant 1812013. Ai's work was partially supported by NSFC grants 11671019 and LMEQF.
  The authors are grateful to the co-Editor, the Associate Editor, and three referees for their valuable comments and suggestions.
      }}
  \date{} 
  \maketitle
} \fi

\if0\blind
{
  \bigskip
  \bigskip
  \bigskip
  \begin{center}
    {\LARGE\bf Optimal Distributed Subsampling for Maximum Quasi-Likelihood Estimators with Massive Data}
\end{center}
  \medskip
} \fi
\begin{abstract}
Nonuniform subsampling methods are effective to reduce computational burden and maintain estimation efficiency for massive data. Existing methods mostly focus on subsampling with replacement due to its high computational efficiency. If the data volume is so large that nonuniform subsampling probabilities cannot be calculated all at once, then subsampling with replacement is infeasible to implement. This paper solves this problem using Poisson subsampling. We first derive optimal Poisson subsampling probabilities in the context of  quasi-likelihood estimation under the A- and L-optimality criteria. For a practically implementable algorithm with approximated optimal subsampling probabilities, we establish the consistency and asymptotic normality of the resultant estimators. To deal with the situation that the full data are stored in different blocks or at multiple locations, we develop a distributed subsampling framework, in which statistics are computed simultaneously on smaller partitions of the full data. Asymptotic properties of the resultant aggregated estimator are investigated. We illustrate and evaluate the proposed strategies through numerical experiments on simulated and real data sets.
\end{abstract}
\noindent%

\newpage
\spacingset{1.45} %
\section{Introduction} \label{sec:intro}
Nowadays, the sizes of collected data are ever increasing, and the
incredible sizes of big data bring new challenges for data
analysis. Although many traditional statistical {\blue methods} are still valid
with big data, it is often computationally infeasible to perform statistical analysis
due to relatively limited computing power. In this
scenario, the bottleneck for big data analysis is the limited
computing resources, and extracting useful information
from massive data sets is a primary goal.

In general, there are two computational barriers for big data analysis: the first is that the data set is too large to be held in a computer's memory; and the second is that the computation takes too long to obtain the results.
Faced with these two challenges, current research on statistical inference for big data sets can be categorized into two basic approaches.
 One approach utilizes parallel computing platforms by dividing the whole data set into subsets to compute; the results from subsets are then  combined to obtain a final estimator, see \cite{Lin2011Aggregated,Duchi2012Dual,Li2013Statistical,Kleiner2015A,schifano2016online,Jordan2018communication} and the references therein.
The other approach uses subsampling to reduce the computational burden by carrying out intended calculations on a subsample drawn from the full data, see \cite{Drineas2011Faster,Dhillon2013New,Ma2015A,Matias2018Speeding}, among others.

A key tactic of subsampling {\bluee methods} is to specify nonuniform sampling probabilities to include more informative data points with higher probabilities.
Typical examples are the leverage score-based subsampling \citep[see][]{Drineas2011Faster,Mahoney2012Randomized,Ma2015A} and optimal subsampling method under the A-optimality criterion \citep{Wang2017Optimal}.  \cite{Wang2017Information} proposed the information based optimal subdata selection for linear models which selects the subsample deterministically without random sampling.

It is worth mentioning that most of the current subsampling strategies  focus on linear regression models and logistic regression models.
However, many more complicated models are required in mining massive data because a linear regression model or a logistic regression model may not be sufficient to fit a complicated large data set. For example, the paper citation data set  (\url{https://www.aminer.cn/citation}) contains text information for over four million research papers. Although we can extract numerical features from these texts, a linear regression or a logistic regression is clearly not adequate to model the number of citations for these papers. As another example, the airline data set (\url{http://stat-computing.org/dataexpo/2009/the-data.html}) has more than one hundred million observations, and a primary goal is to model the airline delays, which are right skewed and always positive. A log or power transform may help to alleviate the skewness, but a Gamma regression may give better interpretability. More details about these two data sets will be provided in Section~\ref{sec:experiments}.
In order to support more statistical models, this paper focuses on the quasi-likelihood estimator  which only requires assumptions on the moments of the response variable and the form of the distribution is not specified.

Subsampling with replacement according to unequal probabilities requires  {\bluee accessing} subsampling probabilities for the full data all at once. This takes a large memory to implement and may reduce the computational efficiency. To overcome this challenge, we propose an algorithm based on Poisson sampling \citep{Sarndal1992Model}. Compared with subsampling with replacement, Poisson subsampling also has a high estimation efficiency with nonuniform subsampling probabilities. %
In order to utilize parallel computing facilities, %
a distributed version of the algorithm is also developed which enables us to select subsamples in parallel or in different locations simultaneously.
To the best of our knowledge, theoretical and methodological discussions with statistical guarantees on optimal subsampling from massive data are limited for statistical models beyond linear regression models and logistic regression models.
This paper not only develops optimal subsampling method for quasi-likelihood estimators but also solves storage constraints imposed by large scale data sets.

The rest of the paper is organized as follows.
In Section \ref{sec:model}, we introduce the model setup, present the general {\blue Poisson} subsampling algorithm, and derive theoretical results for the resultant estimator.
 Section \ref{sec:opt-sample} presents optimal subsampling strategies based on the A- and L-optimality criteria for quasi-likelihood estimators.
 Some practical issues to approximate and implement the optimal subsampling procedures are also considered with theoretical justifications.
Section \ref{sec:dist-sample} designs a distributed version of the Poisson subsampling algorithm and presents asymptotic properties of the resultant estimators.
Section \ref{sec:experiments} provides numerical results on simulated and real data sets. All proofs are deferred in the {supplementary material}.

\section{Preliminaries}\label{sec:model}
In this section, we first provide a brief overview of quasi-likelihood estimation and then present the general Poisson subsampling algorithm.

\subsection{Models and Assumptions}\label{sec:glm}
 We adopt the notations for quasi-likelihood estimator discussed in \cite{chen1999Strong}.
 Let $\{({\bm x}_i, y_i)\} _{i = 1}^N$ be a sequence of independent and identically distributed (i.i.d) random variables with each covariate ${\bm x}_i\in \mathbb{R}^{d}$ and response $y_i \in \mathbb{R}$.
The conditional expectation of the response $y_i$ given ${\bm x}_i$ is
 \begin{equation}\label{eq:model1}
   E(y_i| {\bm x}_i)=\psi({\bm \beta_t}^{T}{\bm x}_i),\quad i = 1,2, \ldots ,N,
 \end{equation}
 for some true regression parameter vector $\bm\beta_t\in\mathbb{R}^d$, where $\psi(\cdot)$ is a twice continuously differentiable function such that $\dot{\psi}(t):=d\psi(t)/dt>0$ for all $t$.
 The quasi-likelihood estimator $\hat{\bm\beta}_{\rm QLE}$ is the solution to the following estimation equation:
 \begin{equation}\label{eq:likelihood}
   Q(\bm\beta):=\sum_{i=1}^N \{y_i-\psi({\bm \beta}^{T}\bm x_i)\}\bm x_i=\bm 0.
 \end{equation}
 The inference procedure based on \eqref{eq:likelihood} is very general, and a typical example is the maximum likelihood estimation for generalized linear models \citep{Mccullagh1989Generalized}.
 More details  can be found in \cite{fahrmeir2001multivariate,Chen2011Quasi} and the references therein.

\subsection{General  Poisson Subsampling Algorithm}\label{sec:subsamplingalg}

{\blue Let $p_i$ be the probability to sample the $i$-th data point  for $i = 1,...,N$, and let $S$ be a set of subsample observations and the corresponding sampling probabilities.}
A general Poisson subsampling algorithm is presented in Algorithm~\ref{alg:1}.
\newpage

\begin{algorithm}[H]
\SetAlgoLined
 {\bf Initialization} $S=\varnothing$\;
 \For{$i=1,\ldots,N$}{
  Generate a Bernoulli variable $\delta_i\sim\text{Bernoulli}(p_i)$\;
  \If {$\delta_i=1$}{ Update $S=S\cup\{(\bm {x}_{i},y_{i},p_{i})\}$\;}{}
 }
 {\bf Estimation:}\
  Solve the following weighted estimation equation to obtain  $\tilde{\bm\beta}$ based on the subsample $S$,
    \begin{equation}\label{eq:reweigt}
      Q^*(\bm\beta)=\sum_{S}\frac{1}{{p}_{i}}\{y_{i}-\psi({\bm \beta}^{T}\bm x_{i})\}\bm x_{i}={\bm 0}.
    \end{equation}
 \caption{General Poisson subsampling algorithm}\label{alg:1}
\end{algorithm}

An advantage of Poisson subsampling is that the decision of inclusion for each data point $({\bm x}_i, y_i)$ is made on the basis of $p_i$ only. We do not need to use all $p_i$ for $i=1,...N$ together.
In Algorithm \ref{alg:1}, $p_i$ can be used one-by-one or block-by-block to generate $\delta_i$ {while scanning through the full data.  %
Therefore, there is no memory constraint problem for massive data.}

The subsample size, say $r^*$, in Algorithm \ref{alg:1} is random such that $E(r^*)=\sum_{i=1}^Np_i$.
We use $r=\sum_{i=1}^Np_i$ to denote the expected subsample size, and  further assume $r<N$ throughout this paper, which is {\blue natural} in the big data setting.

To establish our asymptotic results, we need the following assumptions.
\begin{assumption}\label{assumptionA}
The regression parameter lies in the $l_1$ ball
        $\Lambda=\left\lbrace  \bm\beta \in \mathbb{R}^{p}:\: \|\bm\beta\|_1 \le B \right\rbrace, $ with $\bm\beta_t$ and $\hat{\bm\beta}_{\rm QLE}$ being the inner points of $\Lambda$,
        where $B$ is a constant.
\end{assumption}
\begin{assumption}\label{assumptionB}
Suppose that
\begin{align*}
  &{\rm(i)}~E(\left\| {{\bm x_1}} \right\|^9)<\infty, & \qquad&
{\rm(ii)}~E(|y_1|^6)<\infty,\\
  &{\rm(iii)}~E\Big\{\sup_{\bm\beta\in\Lambda}{\psi}^6({\bm\beta}^T\bm x_1)\Big\}<\infty, & \qquad &
  {\rm(iv)}~E\Big\{\sup_{\bm\beta\in\Lambda}{\dot\psi}^6({\bm\beta}^T\bm x_1)\Big\}<\infty.
\end{align*}
\end{assumption}
\begin{assumption}\label{assumptionC}
Let $\Sigma_\psi({\bm\beta}) = N^{-1}\sum_{i=1}^N\dot{\psi}({\bm\beta}^T\bm x_{i})\bm x_{i}\bm x_{i}^T$, and further assume that  it satisfies
$\lim_{N\rightarrow\infty}\inf_{\bm\beta\in\Lambda}\lambda_{\textrm{min}}\{\Sigma_\psi({\bm\beta})\}>0$ with probability approaching one, where $\lambda_{\textrm{min}}(A)$ means the smallest eigenvalue of matrix $A$.
\end{assumption}
\begin{assumption}\label{assumptionD}
Assume that both $\psi(\bm\beta^T\bm x_i)$ and $\dot\psi(\bm\beta^T\bm x_i)\bm x_i\bm x_i^T$ are $m(\bm x_i)$-Lipschitz continuous. To be precise, for all $\bm\beta_1,\bm\beta_2\in \Lambda$, there exist $m_1(\bm x_i)$ and $m_2(\bm x_i)$ such that
     $\|\psi(\bm\beta_1^T\bm x_i)-\psi(\bm\beta_2^T\bm x_i)\|\le m_1(\bm x_i)\|\bm\beta_1-\bm\beta_2\|$ and $\|\dot\psi(\bm\beta_1^T\bm x_i)\bm x_i\bm x_i^T-\dot\psi(\bm\beta_2^T\bm x_i)\bm x_i\bm x_i^T\|_s\le m_2(\bm x_i)\|\bm\beta_1-\bm\beta_2\|$, where $\|A\|_s$ denotes the spectral norm of matrix $A$. %
     Further assume that both $E\{m_1^3(\bm x_i)\}$ and $ E\{m_2(\bm x_i)\}$ are finite.
\end{assumption}
\begin{assumption}\label{assumptionE}
Assume that $ \max_{i=1,\ldots,N}\ (Np_i)^{-1}={\blue O_P(r^{-1})}$.
\end{assumption}

Assumption \ref{assumptionA} is required to guarantee consistency and it is commonly used in the literature such as \cite{NEWEY1994Large}.
Assumption \ref{assumptionB} imposes some moment assumptions, and similar conditions are also assumed in \cite{chen1999Strong}.
Conditions (iii) and (iv) in Assumption \ref{assumptionB} are satisfied by many examples of generalized linear models  such as linear regressions, logistic regressions, and binomial regressions when the covariate distributions are sub-Gaussian.
Assumption \ref{assumptionC} is mainly to ensure that the quasi-likelihood estimator is unique, since this condition indicates that the quasi log-likelihood function is convex \citep[cf.][]{Tzavelas1998note,Rao2007Linear,Chen2011Quasi}.
Assumption \ref{assumptionD} adds restrictions on smoothness. Similar assumptions are common in statistics \citep[see][Chapter 5 as an example]{Vaart2000Asymptotic}.
Assumption \ref{assumptionE} restricts the weights in the estimation equation~\eqref{eq:reweigt}. It is mainly to protect the estimation equation from being dominated by data points with extremely small subsampling probabilities. {\blue This assumption is quite common in classic sampling
techniques %
\citep[see][as examples]{berger2016empirical,Breidt2000Local}. \red In this paper, we allow the subsampling probability $p_i$ to dependent on the observed data, so we use the $O_P$ notation in Assumption~\ref{assumptionE}. }

To facilitate the presentation, denote the full data by $\mathcal{F}_N=\{{\bm x}_i, y_i\} _{i = 1}^N$.
The following theorems establish consistency to the full data QLE and asymptotically normality of $\tilde{\bm\beta}$ from Algorithm~\ref{alg:1}.

\begin{theorem}\label{thm:as-general-alg}
  If Assumptions 1 -- 5 hold, %
  then as $N\rightarrow\infty$ and $r\rightarrow\infty$, $\tilde{\bm\beta}$ is consistent to $\hat{\bm\beta}_{\rm QLE}$ in conditional probability,
  given $\mathcal{F}_N$ in probability. Moreover, the rate of convergence is $r^{-1/2}$.
  That is, with probability approaching one, for any $\epsilon>0$, there exists a finite $\Delta_\epsilon$ and $r_\epsilon$ such that
  \begin{equation}\label{eq:18}
    P(\|\tilde{\bm\beta}-\hat{\bm\beta}_{\rm QLE}\|\ge
    r^{-1/2}\Delta_\epsilon|\mathcal{F}_N)<\epsilon
  \end{equation}
  for all $r>r_\epsilon$.
\end{theorem}

\begin{theorem}\label{thm:CLT}
  If Assumptions 1 -- 5 hold, %
   then as $N\rightarrow\infty$ and $r\rightarrow\infty$, conditional on $\mathcal{F}_N$ in probability,
  \begin{equation}\label{normal}
    V^{-1/2}(\tilde{\bm\beta}-\hat{\bm\beta}_{\rm QLE})
    \longrightarrow N(0,I)
  \end{equation}
  in distribution, where
  \begin{equation}\label{varorder}
    V=\Sigma_\psi(\hat{\bm\beta}_{\rm QLE})^{-1}V_c\Sigma_\psi(\hat{\bm\beta}_{\rm QLE})^{-1}
  \end{equation}
  and
  \begin{equation}
    V_c=\frac{1}{N^2}
    \sum_{i=1}^N\frac{\{y_i-{\psi}(\hat{\bm\beta}_{\rm QLE}^T\bm x_i)\}^2\bm x_i\bm x_i^T}{p_i}-\frac{1}{N^2}\sum_{i=1}^N{\{y_i-\psi(\hat{\bm\beta}_{\rm QLE}^T\bm x_i)\}^2\bm x_i\bm x_i^T}.
    \label{eq:vc}
  \end{equation}
\end{theorem}

\begin{remark}
  When $r/N\rightarrow 0$, the second term on the right-hand-side of
  \eqref{eq:vc} can be ignored. In this case, the result is the same as that for sampling with replacement in logistic regression \citep[see][]{Wang2017Optimal}.
  However, when $r/N\rightarrow c\in(0,1]$, Poisson subsampling will
  lead to a smaller variance.
\end{remark}

\section{Optimal Poisson Subsampling}\label{sec:opt-sample}
In this section, we derive optimal subsampling probabilities to better approximate $\hat{\bm\beta}_{\rm QLE}$.

\subsection{Optimal Subsampling Strategies}\label{sec:opt-prob}
The result in Theorem~\ref{thm:CLT} can be used to find optimal subsampling probabilities that minimize the asymptotic mean squared error (MSE) of $\tilde{\bm\beta}$ in approximating $\hat{\bm\beta}_{\rm QLE}$.
This is equivalent to minimizing ${\rm tr}(V)$, which corresponds to the A-optimality in the language of optimal design \citep[see][]{pukelsheim2006optimal}.

\begin{theorem}\label{thm:3}
   For {\blue ease} of presentation, define
  \begin{equation}\label{eq:pi-amse-w}
    \hbar_i^{\mathrm{MV}}=
    |y_i-{\psi}(\hat{\bm\beta}_{\rm QLE}^T\bm x_i)|\|\Sigma_\psi(\hat{\bm\beta}_{\rm QLE})^{-1}\bm x_i\|,\;
    i=1 ,\ldots, N,
  \end{equation}
  and let $\hbar_{(1)}^{\mathrm{MV}}\le\hbar_{(2)}^{\mathrm{MV}} \le\ldots\le \hbar_{(N)}^{\mathrm{MV}}$ denote the order statistics of $\{\hbar_{i}^{\mathrm{MV}}\}_{i=1}^N $. For convenience, denote $\hbar_{(N+1)}^{\mathrm{MV}}=+\infty$, and assume that $\hbar_{(N-r)}^{\mathrm{MV}}>0$.
  The asymptotic MSE of $\tilde{\bm\beta}$, $\mathrm{tr}(V)$, attains its minimum, if $p_i$'s in Algorithm \ref{alg:1} are chosen to be
  \begin{equation}\label{eq:1}
    p_i^{\rm MV}=r\frac{\hbar_i^{\rm MV}\wedge M}{\sum_{j=1}^N(\hbar_j^{\rm MV}\wedge M)},
  \end{equation}
   where $a\wedge b=\min(a,b)$,
  \begin{equation}\label{eq:2}
    M=\frac{1}{r-k}\sum_{i=1}^{N-k}\hbar_{(i)}^{\mathrm{MV}},
  \end{equation}
  and
  \begin{equation}
    k=\min\left\{s\ \Bigg{|} 0\le s \le r, \quad
  (r-s)\hbar_{(N-s)}^{\mathrm{MV}}< \sum_{i=1}^{N-s}\hbar_{(i)}^{\mathrm{MV}}\right\},
  \end{equation}
that is, $k$ satisfies
\begin{equation}\label{eq:3}
(r-k+1)\hbar_{(N-k+1)}^{\mathrm{MV}}
  \ge
  \sum_{i=1}^{N-k+1}\hbar_{(i)}^{\mathrm{MV}}
  \quad\text{and}\quad
  (r-k)\hbar_{(N-k)}^{\mathrm{MV}}
  <\sum_{i=1}^{N-k}\hbar_{(i)}^{\mathrm{MV}}.
\end{equation}
\end{theorem}

\begin{remark} In \eqref{eq:1}, if $r\hbar_{(N)}^{\rm MV}/(\sum_{j=1}^N\hbar_j^{\rm MV})<1$, then $\hbar_{(N)}^{\rm MV}< M=r^{-1}\sum_{j=1}^N\hbar_j^{\rm MV}$ and
  the optimal subsampling probabilities reduce to $p_i^{\rm MV}=r\hbar_i^{\rm MV}/(\sum_{j=1}^N\hbar_j^{\rm MV})$. In this case, all $p_i^{\rm MV}$'s are smaller than one and the inclusion of any data point in the subsample is random. If $r\hbar_i^{\rm MV}/(\sum_{j=1}^N\hbar_j^{\rm MV})\ge1$ for some $i$, then some $p_i^{\rm MV}$'s will be equal to one. For this scenario, $k$ is the number of $p_i^{\rm MV}$'s that are one and $M$ is the threshold that satisfies
   \begin{equation}\label{eq:M}
 \max_{i=1,\ldots,N}\frac{r(\hbar_i^{\rm MV}\wedge M)}{\sum_{j=1}^N(\hbar_j^{\rm MV}\wedge M)}=1.
\end{equation}
From \eqref{eq:2} and \eqref{eq:3}, we see that
\begin{equation}
  \hbar_{(N-k)}^{\rm MV} < M \le \hbar_{(N-k+1)}^{\rm MV}.
\end{equation}
\end{remark}

\begin{remark}
  In order to determine the value of $k$, we need to find and sort at most $r$ largest values of $\hbar_i^{\rm MV}$'s. Thus, the required time to find the value of $k$ is $O(N+r\log r)$ using partition based partial selection algorithm.
  The simulation results reveal that when {\blue$r/N\to c>0$}, it also works well if we select $M$ as some quantile of $\{\hbar_i\}_{i=1}^N$.
\end{remark}

As observed in (\ref{eq:pi-amse-w}), the optimal subsampling probability $\bm{p}^{\mathrm{MV}}=\{p_{i}^{\mathrm{MV}}\}_{i=1}^{N}$ depends on data through both the covariates and the responses directly.
 For the covariates,  the terms $\|\Sigma_\psi(\hat{\bm\beta}_{\rm QLE})^{-1}\bm x_i\|$ describe the structure information of the covariates and they are similar to statistical leverage scores.
  The direct effect of the responses on the optimal subsampling probabilities is through $|y_i-{\psi}(\hat{\bm\beta}_{\rm QLE}^T\bm x_i)|$.
  Intuitively, including data points with lager values of $|y_i-{\psi}(\hat{\bm\beta}_{\rm QLE}^T\bm x_i)|$ will improve the robustness of the subsample estimator.

The optimal subsampling strategy derived in the previous section requires the calculation of $\|\Sigma_\psi(\hat{\bm\beta}_{\rm QLE})^{-1}\bm x_i\|$ for $i=1,2,...,N$, which takes $O(Nd^2)$ time even if $\Sigma_\psi(\hat{\bm\beta}_{\rm QLE})$ is available.
 To further reduce the computation time, \cite{Wang2017Optimal}  proposed to minimize  ${\rm tr}(V_c)$.
 This criterion essentially is the linear optimality (L-optimality) criterion in optimal experimental design \citep[see][]{pukelsheim2006optimal}, which is to improve the quality of the estimator for some linear combinations of unknown parameters.

 The following theorem gives the optimal subsampling probabilities that minimize ${\rm tr}(V_c)$.

\begin{theorem}\label{thm:5}
   Let
  \begin{equation}\label{eq:pi-mvc-w}
    \hbar_i^{\mathrm{MVc}}=
    |y_i-{\psi}(\hat{\bm\beta}_{\rm QLE}^T\bm x_i)|\|\bm x_i\|,\;
    i=1,\ldots,N,
  \end{equation}
   and let $\hbar_{(1)}^{\mathrm{MVc}}\le\hbar_{(2)}^{\mathrm{MVc}} \le\ldots\le \hbar_{(N)}^{\mathrm{MVc}}$ denote the order statistics of $\{\hbar_{i}^{\mathrm{MVc}}\}_{i=1}^N$. For convenience, denote $\hbar_{(N+1)}^{\mathrm{MVc}}=+\infty$ and assume that $\hbar_{(N-r)}^{\mathrm{MVc}}>0$.
  The trace of $V_c$ defined in \eqref{eq:vc} attains its minimum if $p_i$'s in Algorithm \ref{alg:1} are selected as
  \begin{equation}
    p_i^{\rm MVc}=r\frac{\hbar_i^{\rm MVc}\wedge M}{\sum_{j=1}^N\hbar_j^{\rm MVc}\wedge M},
  \end{equation}
   where
  \begin{equation}
    M=(r-k)^{-1}\sum_{i=1}^{N-k}\hbar_{(i)}^{\mathrm{MVc}},
  \end{equation}
and
  \begin{equation}
    k=\min\Bigg\{s\ \Bigg|\ 0\le s \le r, \quad
  (r-s)\hbar_{(N-s)}^{\mathrm{MVc}}< \sum_{i=1}^{N-s}\hbar_{(i)}^{\mathrm{MVc}}\Bigg\},
  \end{equation}
that is,  $k$ satisfies
\[(r-k+1)\hbar_{(N-k+1)}^{\mathrm{MVv}}\ge\sum_{i=1}^{N-k+1}\hbar_{(i)}^{\mathrm{MV}}\quad\text{and}\quad
  (r-k)\hbar_{(N-k)}^{\mathrm{MVc}}<\sum_{i=1}^{N-k}\hbar_{(i)}^{\mathrm{MVc}}.\]
\end{theorem}

The structural results for $\bm{p}^{\mathrm{MVc}}=\{p_{i}^{\mathrm{MVc}}\}_{i=1}^{N}$ and $\bm{p}^{\mathrm{MV}}$ are similar. The difference is in the covariate effect: $\bm{p}^{\mathrm{MV}}$ uses $\|\Sigma_\psi(\hat{\bm\beta}_{\rm QLE})^{-1}\bm x_i\|$ while $\bm{p}^{\mathrm{MVc}}$ uses $\|\bm x_i\|$. {\blue The computational benefits is obvious, only $O(Nd)$ time is required to compute $\bm{p}^{\mathrm{MVc}}$ while $O(Nd^2)$ is needed for $\bm{p}^{\mathrm{MV}}$.  }

\subsection{Practical Implementation}\label{sec:prac}
For ease of presentation, we use a unified notation $p_i^{\rm os}$ to denote the optimal subsampling probabilities $p_i^{\rm MV}$ or $p_i^{\rm MVc}$ derived in Theorems \ref{thm:3} or \ref{thm:5}, respectively.
To be precise,
\begin{equation}\label{eq:p-os}
  p_i^{\rm os}=r\frac{\hbar_i^{\rm os}\wedge M}{\sum_{j=1}^N(\hbar_j^{\rm os}\wedge M)}
    =r\frac{\hbar_i^{\rm os}\wedge M}{N\Psi},\quad i=1,\ldots,N,
\end{equation}
where $M=(r-k)^{-1}\sum_{i=1}^{N-k}\hbar_{(i)}^{\rm os}$, $\Psi=N^{-1}\sum_{j=1}^N(\hbar_j^{\rm os}\wedge M)$, and $\hbar_i^{os}$ is either $\hbar_i^{\rm MV}$ or $\hbar_i^{\rm MVc}$.

To practically implement the optimal subsampling probabilities, we need to  replace  the unknown $\hat{\bm\beta}_{\rm QLE}$ by a pilot estimator, {\blue say $\tilde{\bm\beta}_0$}, which can be obtained by taking a uniform subsample.
Some other sampling distributions can also be used to obtain the pilot estimator as long as they satisfy Assumption \ref{assumptionE} and are computationally feasible to implement.
Furthermore, in order to take advantage of Poisson subsampling and determine the inclusion of each data point separately, we use the pilot sample to approximate $M$ and $\Psi$.

In the setting of subsampling for computational efficiency, it is typical that $r\ll N$ and the number of cases that $\hbar_{i}^{\rm os}>M$ is small. Thus, taking $M=\infty$ will not significantly affect the optimal subsampling probabilities. In facts, if ${r\hbar_{(N)}^{\rm os}}/{(\sum_{j=1}^N\hbar_j^{\rm os})}\le1$, then taking $M=\infty$ does not affect the optimal subsampling probabilities at all. Simulation results in Section~\ref{sec:experiments} show that taking $M=\infty$ does not reduce the estimation efficiency as long as $r/N$ is small.

{\blue Let $\tilde{S}_{r_0}$ be the set of the pilot subsample and
\begin{equation}\label{eq:hatpsi}%
      \hat\Psi
  =\frac{1}{|\tilde{S}_{r_0}|}\sum_{\tilde{S}_{r_0}}
    |y_i-\psi(\tilde{\bm\beta}_0^T\bm x_i)|h(\bm x_i),
\end{equation}
  where $|\tilde{S}_{r_0}|$ is the size of $\tilde{S}_{r_0}$, %
   and $h(\bm x)=\|\bm x\|$ for {\rm MVc} or $h(\bm x)=\|\Sigma_\psi(\tilde{\bm\beta}_0)^{-1}\bm x\|$ for {\rm MV} with $\Sigma_\psi(\tilde{\bm\beta}_0)$ calculated as
    $\Sigma_\psi(\tilde{\bm\beta}_0)={{|\tilde{S}_{r_0}|}}^{-1}\sum_{\tilde{S}_{r_0}}\dot{\psi}(\tilde{\bm\beta}_0^T\bm x_i^*)\bm x_i^*\bm x_i^{*T}.$}
Let $\tilde{p}_i^{\rm os}$ be the approximated subsampling probabilities with $\hat{\bm\beta}_{\rm QLE}$, $M$, and $\Psi$ {\blue in (\ref{eq:p-os}) } replaced by the pilot estimator {\blue $\tilde{\bm\beta}_0$, $M=\infty$, and $\hat\Psi$}.
The  {\bluee weighted} estimator with $\tilde{p}_i^{\rm os}$ inserted in \eqref{eq:reweigt} may be sensitive to data points with $y_i-\dot{\psi}(\tilde{\bm\beta}_{0}^T\bm x_i)\approx0$ if they are included in the subsample.
To make the estimator more stable and robust, we adopt the idea of shrinkage-based subsampling method proposed in \cite{Ma2015A}.
To be specific, we use the following subsampling probabilities
{\red
\begin{equation}\label{eq:opts-prob}
  \tilde{p}_i^{\rm sos}= (1-{\varrho})
  \frac{r|y_i-\psi(\tilde{\bm\beta}_0^T\bm x_i)|h(\bm x_i)}
  {N\hat\Psi}+{\varrho}{rN^{-1}}, \quad i=1, ..., N,
\end{equation}
where $\varrho\in(0,1)$.

Note that when $\tilde{\bm\beta}_0$ and $\hat\Psi$ are calculated from the pilot subsample, $\tilde{p}_i^{\rm sos}$ depends on the $i$-th observation $(\bm x_i,y_i)$ only. Thus, each $\tilde{p}_i^{\rm sos}$ can be calculated when scanning the data from hard drive line-by-line or block-by-block; there is no need to calculate $\tilde{p}_i^{\rm sos}$'s all at once. Therefore, there is not need to load the full data into memory to calculate all $\tilde{p}_i^{\rm sos}$'s and this is very computationally beneficially in terms of memory usage.
}

In \eqref{eq:opts-prob}, $\tilde{\bm{p}}^{\rm sos}=\{\tilde{p}_i^{\rm sos}\}_{i=1}^N$
is a convex combination of $\tilde{\bm{p}}^{\rm os}=\{\tilde{p}_i^{\rm os}\}_{i=1}^N$ and the uniform subsampling probability, and it shares the strengths of both.
When $\varrho$ is larger, the corresponding estimator will be more stable since the estimation equation will not be inflated by data points with extremely small values of $\tilde{p}_i^{\rm os}$.
The rankings of $\tilde{p}_i^{\rm sos}$ and $\tilde{p}_i^{\rm os}$ are the same, so the estimator still enjoys the benefits of the optimal subsampling strategy.
The shrinkage term not only increases small subsampling probabilities, but also shrinks large subsampling probabilities and thus protects the effects of potential outliers to some extent.

Since we approximate $\Psi$ and take $M=\infty$, some $\tilde{p}_i^{\rm sos}$ may be larger than one. Thus, we need to use inverses of $\tilde{p}_i^{\rm sos}\wedge1$'s as weights in the subsample QLE estimator.  %
For transparent presentation, we summarize the practical procedure with approximated quantities in Algorithm~\ref{alg:3}.

 \newpage

\begin{algorithm}[H]
\SetAlgoLined
 {\bf Pilot Subsampling:} Run Algorithm~\ref{alg:1} with average subsample size $r_0$ and {\blue $\bm{p}^{\rm UNIF}=\{p_i:=r_0/N\}_{i=1}^{N}$} to take {\bluee a subsample set $\tilde{S}_{r_0}$, and use it to obtain an estimate  $\tilde{\bm\beta}_0$ and $\hat\Psi$ as in (\ref{eq:hatpsi})}.

 {\bf Initialization:} $S_0=\tilde{S}_{r_0}$\;
 \For{$i=1,\ldots,N$}{
  Generate $\delta_i\sim\text{Bernoulli}(1,p_i)$ with $p_i=\tilde{p}_i^{\rm sos}\wedge 1$, where $\tilde{p}_i^{\rm sos}$ is defined in (\ref{eq:opts-prob})\;
  \eIf {$\delta_i=1$}{ Update $S_{i}=S_{{i}-1}\cup\{(y_{i},\bm {x}_{i},{p}_{i})\}$}{Set $S_{i}=S_{{i}-1}$}
 }
 {\bf Estimation:}\
  Solve the following weighted estimating equation  to obtain the estimate $\breve{\bm\beta}$ based on the subsample set $S_N$. %
  \begin{equation*}{\bluee
      Q^*(\bm\beta)=\sum_{S_N}\frac{1}{{p}_{i}}[y_{i}-\psi({\bm \beta}^{T}\bm x_{i})]\bm x_{i}=0.}
    \end{equation*}
 \caption{Practical Algorithm}\label{alg:3}
\end{algorithm}

For estimators obtained from Algorithm \ref{alg:3}, we derive  asymptotic properties as follows.

\begin{theorem}\label{thm:asy-2step-alg}
Under Assumptions~1 -- 4, if %
$r_0r^{-1/2}\rightarrow0$, then for the estimator $\breve{\bm\beta}$ obtained from Algorithm~\ref{alg:3},  %
  as ${r}\rightarrow\infty$ and $N\rightarrow\infty$, with probability approaching one, for any $\epsilon>0$, there exist   finite $\Delta_\epsilon$ and $r_\epsilon$ such that
  \begin{equation*}
  P(\|{\breve{\bm\beta}}-\hat{\bm\beta}_{\rm QLE}\|\ge {r}^{-1/2}\Delta_\epsilon|\mathcal{F}_N)<\epsilon
\end{equation*}
for all ${r}>r_\epsilon$.
\end{theorem}

\begin{theorem}\label{thm:CLT-2step}
 If Assumptions 1 -- 4 hold and $r_0r^{-1/2}\rightarrow0$, then as $r_0\rightarrow\infty$, $r\rightarrow\infty$ and $N\rightarrow\infty$,  conditionally on $\mathcal{F}_N$ in probability,
  \begin{equation*}
    V^{-1/2}(\breve{\bm\beta}-\hat{\bm\beta}_{\rm QLE}) \rightarrow N(0,I) \text{ in distribution, }
  \end{equation*}
 where $V=\Sigma_\psi(\hat{\bm\beta}_{\rm QLE})^{-1}V_c\Sigma_\psi(\hat{\bm\beta}_{\rm QLE})^{-1}$  and
  \begin{equation*}
    V_c=\frac{1}{N^2}
    \sum_{i=1}^N \frac{\{1-(p_i^{\mathrm{sos}}\wedge 1)\}\{y_i-{\psi}(\hat{\bm\beta}_{\rm QLE}^T\bm x_i)\}^2\bm x_i\bm x_i^T}{p_i^{\mathrm{sos}}\wedge 1},
  \end{equation*}
  { with
  $$p_i^{\mathrm{sos}}:=(1-{\varrho})\frac{r|y_i-{\psi}(\hat{\bm\beta}_{\rm QLE}^T\bm x_i)|\|\Sigma_\psi(\hat{\bm\beta}_{\rm QLE})^{-1}\bm x_i\|}{\sum_{j=1}^N|y_j-{\psi}(\hat{\bm\beta}_{\rm QLE}^T\bm x_j)|\|\Sigma_\psi(\hat{\bm\beta}_{\rm QLE})^{-1}\bm x_j\|}+{\varrho}{\frac{r}{N}},$$
  for $MV$ criterion and
  $$p_i^{\mathrm{sos}}:=(1-{\varrho})\frac{r|y_i-{\psi}(\hat{\bm\beta}_{\rm QLE}^T\bm x_i)|\|\bm x_i\|}{\sum_{j=1}^N|y_j-{\psi}(\hat{\bm\beta}_{\rm QLE}^T\bm x_j)|\|\bm x_j\|}+{\varrho}{\frac{r}{N}},$$
  for $MVc$ criterion. }
\end{theorem}

\section{Distributed Poisson Subsampling}\label{sec:dist-sample}

In this section, we discuss the distributed optimal Poisson subsampling procedure.
For large data sets, it is common to analyze them on multiple machines. %
This motivates us to develop divide-and-conquer subsampling procedures that take advantages of parallel and distributed computational architectures.
Although {\blue Poisson} subsampling can be easily implemented in parallel, pooling the subsample sets from multiple machines together may still result in a subsample set that exceeds the memory limit of a single machine. In addition, transferring data may be time consuming and subject to security issues.
Thus this method can only be used when the subsample size on each machine is not that big.
We propose to aggregate estimators derived in different machines to approximate the full data quasi-likelihood estimator.
Here we assume that the entire data set of size $N$ are stored in $K$ different machines, and let $\mathcal{F}_{Nj}$ ($j=1,...,K$) denote the data stored in the $j$-th machine. For simplicity, assume that the number of observations in different machines are all equal to $n$, and denote the observations in $\mathcal{F}_{Nj}$ as  $\{(y_{ji},\bm x_{ji})\}_{i=1}^n$.
We present the distributed optimal Poisson subsampling procedure in Algorithm \ref{alg:4}.

\begin{footnotesize}
\begin{algorithm}
\SetAlgoLined
 {\bf Step 1: Obtain the pilot estimator}

   \For{ $i=1,\ldots,N$}{
    Generate $\delta_{i}\sim\text{Bernoulli}(1,p_{i})$ with $p_{i}=r_0/N$\;
     \If{$\delta_i=1$}{ Add $(x_i,y_i,p_{i})$ to the subsample set $S_{r_0}$}{}}
For the obtained subsample $S_{r_0}$, calculate the pilot estimator $\tilde{\bm\beta}_0$, {\bluee $\hat\Psi$,} and $\dot{Q}_0$ according to (\ref{eq:reweigt}), (\ref{eq:hatpsi}) and \eqref{eq:dotq}, respectively.\

  {\bf Step 2: Subsampling and Compression}

\ForEach{ $\mathcal{F}_{Nj}, j= 1,\ldots,K$}{

{\bf Initialization:} $S_{j0}=\varnothing$\;
\For {$i=1,\ldots,n$} {
    Calculate the corresponding subsampling probabilities $\tilde{{p}}^{\rm sos}_{ji}$ {\bluee according to (\ref{eq:opts-prob})}\;
    Generate $\delta_{ji}\sim\text{Bernoulli}(1,p_{ji})$ with $p_{ji}=\tilde{p}^{\rm sos}_{ji}\wedge 1$\;
    \eIf{$\delta_{ji}=1$}{Update $S_{ji}=S_{{ji}-1}\cup\{(y_{ji},\bm {x}_{ji},p_{ji})\}$.} {Set $S_{ji}=S_{{ji}-1}$.}
    }

Obtain ${\tilde{\bm\beta}_j}$ by solving
\begin{equation}\label{eq:reweigt1}
   {\bluee    Q_j^*(\bm\beta)=\frac{1}{n}\sum_{ S_{jn}}\frac{1}{{p}_{ji}}\{y_{ji}-\psi({\bm \beta}^{T}\bm x_{ji})\}\bm x_{ji}={\bm 0},}
\end{equation}
and calculate
\begin{equation}\label{eq:dotq}
{\bluee \dot{Q}_j^*({\tilde{\bm\beta}_j})=-\frac{1}{n}\sum_{ S_{jn}}\frac{1}{p_{ji}}\dot\psi(\tilde{\bm\beta}_j^T\bm x_{ji})\bm x_{ji}{\bm x_{ji}}^{T}.}
\end{equation}
}

{\bf Step 3: Combination}

Combine the $K$ estimators and the pilot estimator by calculating
\begin{equation}\label{eq:crbeta}
  \tilde{\bm\beta}_{Kr}=\left\{\sum_{j=0}^K\dot{Q}_j^*(\tilde{\bm\beta}_j)\right\}^{-1}\sum_{j=0}^K\dot{Q}_j^*(\tilde{\bm\beta}_j){\tilde{\bm\beta}}_j.
\end{equation}
\caption{Distributed Optimal Poisson Subsampling}
  \label{alg:4}
\end{algorithm}
\end{footnotesize}
\begin{remark}
The first step in Algorithm \ref{alg:4} can be implemented by sampling the data machine-by-machine and pooling all the subsamples together. Since $r_0$ is usually small in our setting, the time of communication can be ignored.
\end{remark}

The results of consistency and asymptotic normality are presented in the following theorems.

\begin{theorem}\label{thm:asy-2step-alg4}
Under Assumptions~1 -- 4, if the estimator $\tilde{\bm\beta}_0$ based on the first step sample exists, $r_0{(Kr)}^{-1/2}\rightarrow0$ and the partition number $K$ satisfies $K=O(r^\eta)$ for some $\eta$ in $[0,1/3]$, then conditional on $\mathcal{F}_N$, for the estimator $\tilde{\bm\beta}_{Kr}$ obtained from Algorithm~\ref{alg:4}, %
  as ${r}\rightarrow\infty$ and $n\rightarrow\infty$, with probability approaching one, for any $\epsilon>0$, there exist finite $\Delta_\epsilon$ and $r_\epsilon$ such that
  \begin{equation*}
  P(\|\tilde{\bm\beta}_{Kr}-\hat{\bm\beta}_{\rm QLE}\|\ge {(Kr)}^{-1/2}\Delta_\epsilon|\mathcal{F}_N)<\epsilon
\end{equation*}
for all ${r}>r_\epsilon$.

\end{theorem}

\begin{theorem}\label{thm:CLT-2step-alg4}
Under Assumptions 1 -- 4,  if  $r_0{(Kr)}^{-1/2}\rightarrow0$ and the partition number $K$ satisfies $K=O(r^\eta)$ for some $\eta$ in $[0,1/3]$, then  for the estimator $\tilde{\bm\beta}_{Kr}$ obtained from Algorithm~\ref{alg:4}, conditionally on $\mathcal{F}_N$ in probability, as $n\rightarrow\infty$, $r\rightarrow\infty$ and $r_0\rightarrow\infty$,
  \begin{equation*}
    V_{opt}^{-1/2}(\tilde{\bm\beta}_{Kr}-\hat{\bm\beta}_{\rm QLE}) \rightarrow N(0,I) \quad\text{ in distribution, }
  \end{equation*}
 where $V_{opt}=\Sigma_\psi(\hat{\bm\beta}_{\rm QLE})^{-1}V_{c,opt}\Sigma_\psi(\hat{\bm\beta}_{\rm QLE})^{-1}$,
 \begin{equation*}
   V_{c,opt}=\frac{1}{KN^2}\sum_{i=1}^N \frac{\{1-(p_{i}^{\rm sos}\wedge 1)\}\{y_i-{\psi}(\hat{\bm\beta}_{\rm QLE}^T\bm x_i)\}^2\bm x_i\bm x_i^T}{ p_{i}^{\rm sos}\wedge 1},
  \end{equation*}
  and
  $p_i^{\rm sos}$ is defined in Theorem \ref{thm:CLT-2step}.
\end{theorem}

For statistical inference, we propose to estimate the {\bluee asymptotic} variance-covariance matrix of $\tilde{\bm\beta}_{Kr}$ using %
 \begin{equation}\label{eq:24}
   \tilde{V}=\left\{\frac{1}{N}\sum_{j=0}^K\dot{Q}_j^*(\tilde{\bm\beta}_j)\right\}^{-1}\tilde{V}_c\left\{\frac{1}{N}\sum_{j=0}^K\dot{Q}_j^*(\tilde{\bm\beta}_j)\right\}^{-1},
 \end{equation}
where
  \begin{align*}    \tilde{V}_c=&\frac{1}{N^2}\Bigg\{\sum_{S_{r_0}}\frac{\{y_{0i}^*-{\psi}(\tilde{\bm\beta}_0^T\bm x_{0i}^*)\}^2\bm x_{0i}^*{\bm x_{0i}^*}^T}{(r_0/N)^2}(1-r_0/N)\\
   &\qquad+\sum_{j=1}^K\sum_{S_{jn}}\frac{\{y_{ji}^*-{\psi}(\tilde{\bm\beta}_j^T\bm x_{ji}^*)\}^2\bm x_{ji}^*{\bm x_{ji}^*}^T}{(\tilde p_{ji}^{\rm sos *})^2}(1-\tilde p_{ji}^{\rm sos *})\Bigg\}.
  \end{align*}
  This formula enables us to know how well $\tilde{\bm\beta}_{Kr}$ approximates $\hat{\bm\beta}_{\rm QLE}$. When $Kr=o(N)$, we can also draw inference on the true parameter $\bm\beta_t$, since uncertainty  of $\hat{\bm\beta}_{\rm QLE}$ can be ignored under this assumption.
It is worth mentioning that if we want to calculate~\eqref{eq:24}, we also need to have $\sum_{S_{jn}}(1-\tilde p_{ji}^{\rm sos}){\{y_{ji}-{\psi}(\tilde{\bm\beta}_i^T\bm x_{ji})\}^2\bm x_{ji}\bm x_{ji}^T}/{(\tilde{p}_{ji}^{\rm sos})^2}$ calculated on each machine.

Since the pilot estimator $\tilde{\bm\beta}_0$ has to be calculated anyway, our method is valuable even for the case $K=1$ because this avoids iterative calculation on the Step 1 sample twice.

\section{Numerical Studies}\label{sec:experiments}
In this section, we present examples of numerical experiments using the methods developed in Sections \ref{sec:opt-sample} and \ref{sec:dist-sample}. Computations are performed using {\verb"R"} \citep{Rpackage2018}.
The performance of a sampling strategy is evaluated by the empirical MSE of the resultant estimator:
\[{\blue \text{MSE}=\frac{1}{T}\sum_{t=1}^T\|{\bm\beta}_{\bm p}^{(t)}-\hat{\bm\beta}_{\rm QLE}\|^2},\]
where ${\bm\beta}_{\bm{p}}^{(t)}$ is the estimate from the $t$-th subsample with subsampling probability $\bm p$ and $\hat{\bm\beta}_{\rm QLE}$ is the quasi-likelihood estimator calculated from the whole data set. We set  $T = 1000$ throughout this section.

\subsection{Simulation Studies}\label{sec:simulation}
We take Poisson regression as an example to evaluate the finite sample performance of the proposed methods throughout this section.
We also considered logistic regression and Gamma regression models, the results were similar and thus were omitted.
Full data of size $N=500,000$ are generated from a Poisson regression model such that given the covariate $\bm x$, the response $y$ follows a Poisson distribution with mean { ${E}(y|\bm {x} )=\exp(\bm\beta^T\bm{x})$}. Here we set the true value of $\bm\beta$ as a $7\times1$ vector
of 0.5. We consider the following four scenarios to generate the covariates $\bm x_i=(x_{i1}, ..., x_{i7})^{T}$.
\begin{enumerate}[{Case} 1]

\item  The seven covariates are i.i.d from the standard uniform distribution, namely, $x_{ij}\overset{\text{i.i.d}}{\sim} U(0,1)$ for $j=1, ..., 7$.

\item %
The second covariate is $x_{i2}=x_{i1}+\varepsilon_{i}$ with $x_{i1}\sim U(0,1)$, $\varepsilon_{i}\overset{\text{i.i.d}}{\sim} U(0,1)$, and other covariates are $x_{ij}\overset{\text{i.i.d}}{\sim} U(0,1)$ for $j=1,3,\ldots,7$. %
In this scenario, the first two covariates are  correlated ($\approx 0.5$). %
\item
This scenario is the same as Case 2 except that $\varepsilon_{i}$ $\overset{\text{i.i.d}}{\sim} U(\left[0,0.1 \right] )$. For this case, the correlation between the first two covariates is close to $0.8$.

\item   This scenario is the same as Case 2 except that  $x_{ij}\overset{\text{i.i.d}}{\sim} U(\left[-1, 1 \right] )$ for $j=6, 7$. For this case, the supports for different covariates are not all the same.
\end{enumerate}

In the following, we evaluate the performance of Algorithm~\ref{alg:4} based on MV and MVc subsampling probabilities with partition number $K=1$ and $K=5$. Note that Algorithm~\ref{alg:4} with $K=1$ and Algorithm~\ref{alg:3} differ only in the way to incorporate pilot sample information, so their performances are similar. Results of uniform subsampling are also calculated for comparisons.

  We fix $r_0=200$ and $\varrho=0.2$, and
choose $r$ to be 300, 500, 700, 1000, 1200, 1500, 1700 and 2000.
Since the uniform subsampling probability does not depend on unknown parameters and no pilot subsamples are required, it is implemented with subsample size $r +r_0$ for fair comparisons.

Figure \ref{fig:1} gives the simulation results. %
It is seen that for the four data sets, subsampling methods based on MV and MVc always result in smaller empirical MSEs compared with the uniform subsampling, which agrees with the theoretical results in Section \ref{sec:opt-sample}. {\red The MSEs for all subsampling methods decrease as $r$ increases, which confirms the theoretical result on consistency of the subsampling methods.}
\begin{figure}[H]%
  \centering
  \begin{subfigure}{0.49\textwidth}
    \includegraphics[width=\textwidth]{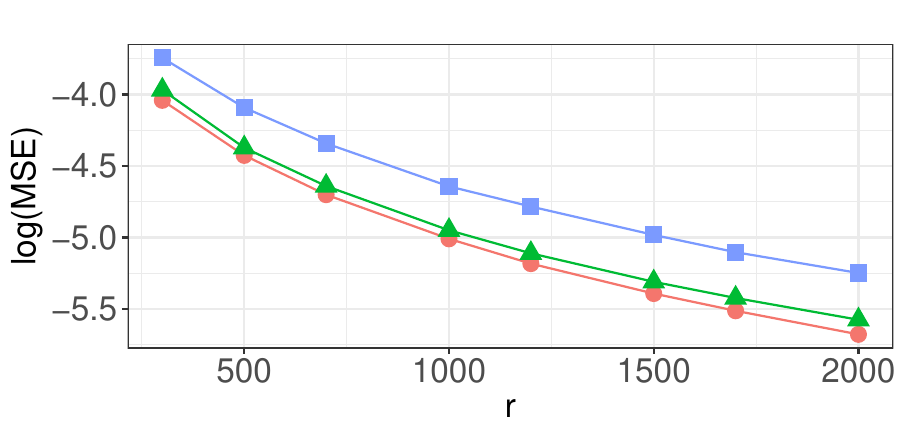}\\[-1cm]
    \caption{Case 1 (K=1)}
  \end{subfigure}
  \begin{subfigure}{0.49\textwidth}
    \includegraphics[width=\textwidth]{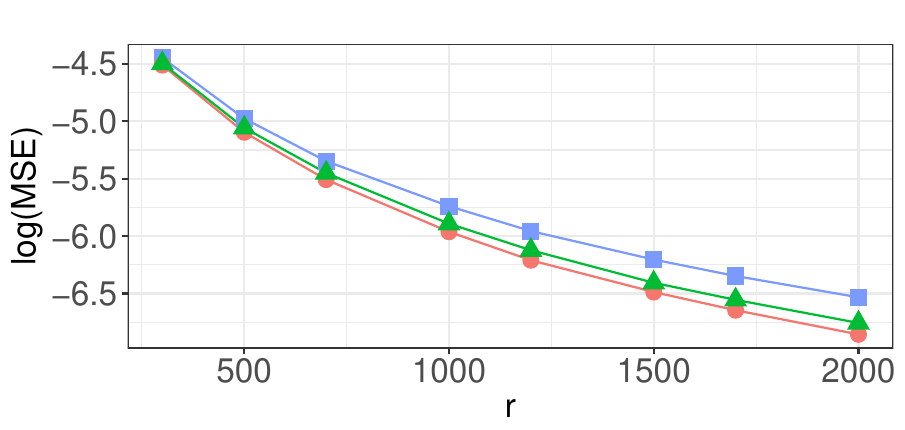}\\[-1cm]
    \caption{Case 1 (K=5)}
  \end{subfigure}\\[0.5mm]
  \begin{subfigure}{0.49\textwidth}
    \includegraphics[width=\textwidth]{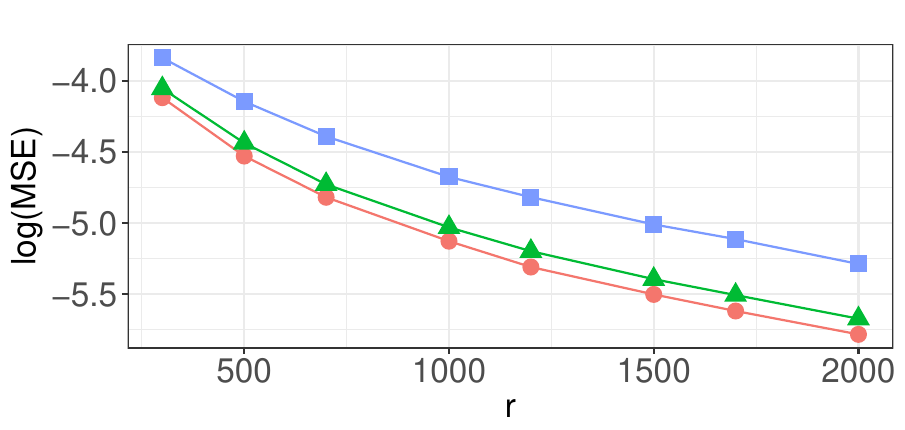}\\[-1cm]
    \caption{Case 2 (K=1)}
  \end{subfigure}
  \begin{subfigure}{0.49\textwidth}
    \includegraphics[width=\textwidth]{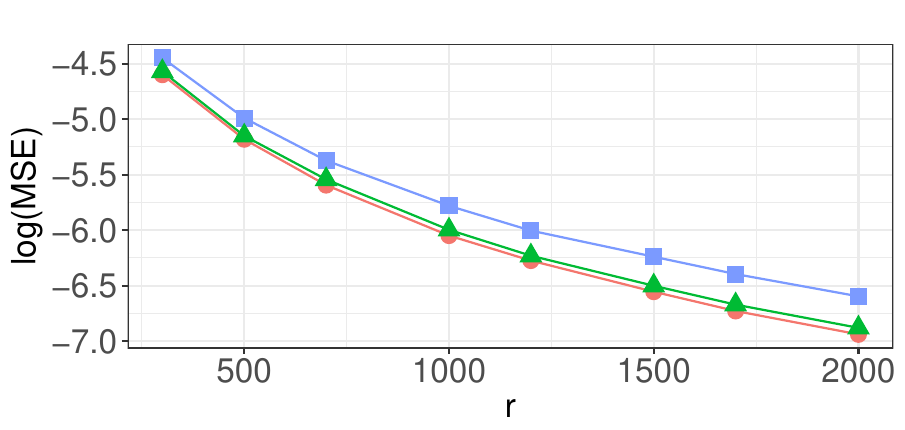}\\[-1cm]
    \caption{Case 2 (K=5)}
  \end{subfigure}\\[0.5mm]
   \begin{subfigure}{0.49\textwidth}
    \includegraphics[width=\textwidth]{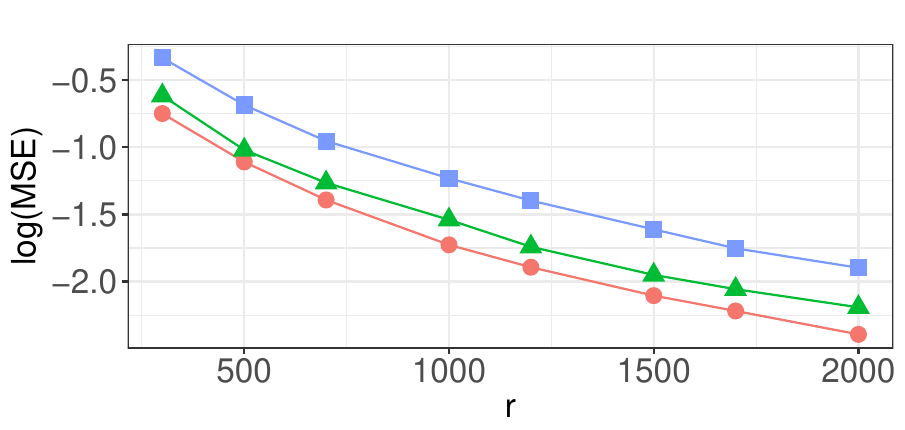}\\[-1cm]
    \caption{Case 3 (K=1)}
  \end{subfigure}
  \begin{subfigure}{0.49\textwidth}
    \includegraphics[width=\textwidth]{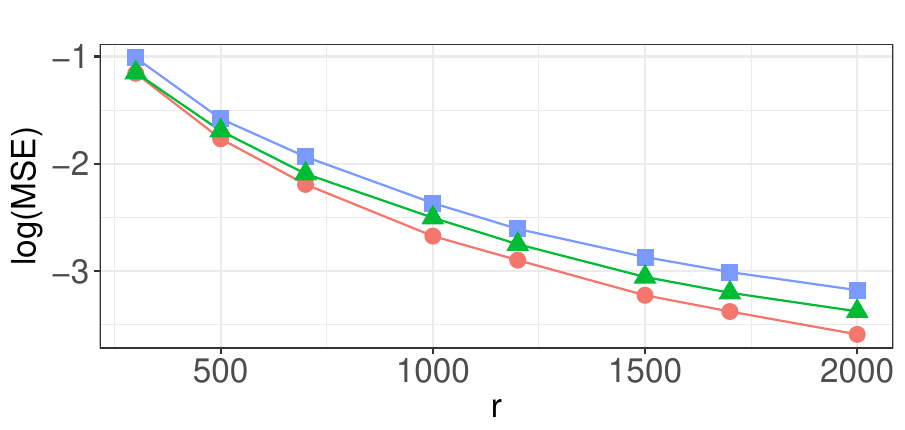}\\[-1cm]
    \caption{Case 3 (K=5)}
  \end{subfigure}\\[0.5mm]
  \begin{subfigure}{0.49\textwidth}
    \includegraphics[width=\textwidth]{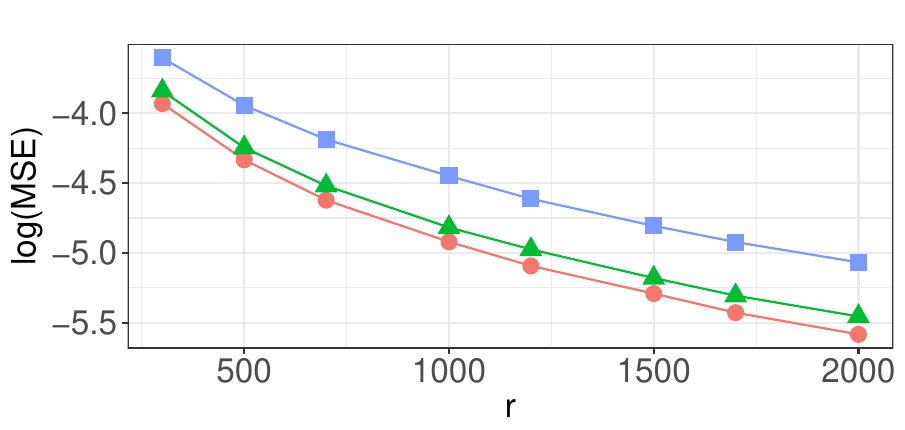}\\[-1cm]
    \caption{Case 4 (K=1)}
  \end{subfigure}
  \begin{subfigure}{0.49\textwidth}
    \includegraphics[width=\textwidth]{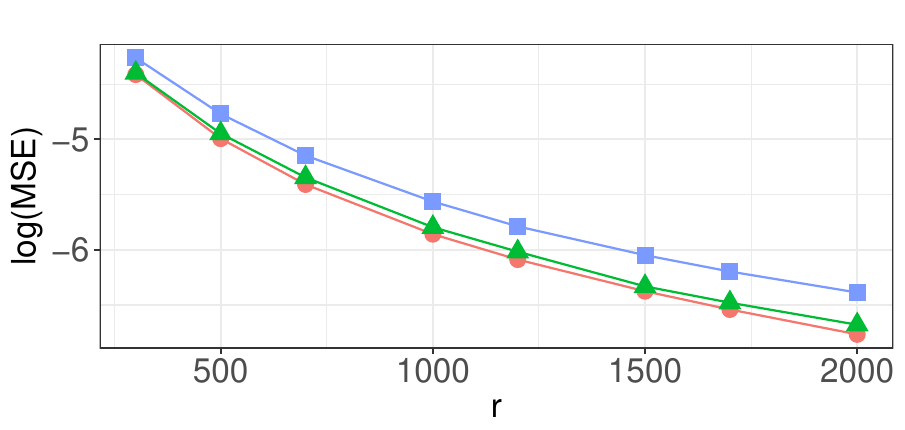}\\[-1cm]
    \caption{Case 4 (K=5)}
  \end{subfigure}\\% [-0.5mm]
  \caption{A graph showing the log of MSE with different ${r}$ and $K$ for different distributions of covariates based on MV (red circle), MVc (green triangle) and uniform subsampling (blue square) methods where $r_0=200$ and $\varrho=0.2$.}
  \label{fig:1}
\end{figure}

 Next, we will explore the effect of different $\varrho$ with fixed $r_0$ and $r$.  %
The results are given in Figure \ref{fig:2} with $r_0=200$, and $r$ = 1200 and 1500. %
It is clear to see that the subsampling method outperforms the uniform subsampling method when $\varrho\in[0.01,0.99]$.
When $\varrho$ is close to 1, the performances of $\tilde{\bm{p}}^{\rm sos}$ are similar  to that of the uniform subsampling.
The two-step approach works the best when $\varrho$ is around 0.25.
This implies that the shrinkage estimator effectively protect the weighted estimating equation from data points with $|y_i-{\psi}(\tilde{\bm\beta}_0^T\bm x_i)|$ close to zero.
We only present the performance of Case 4 here because results for all other cases are similar.

\begin{figure}[H]%
  \centering
  \begin{subfigure}{0.49\textwidth}
    \includegraphics[width=\textwidth]{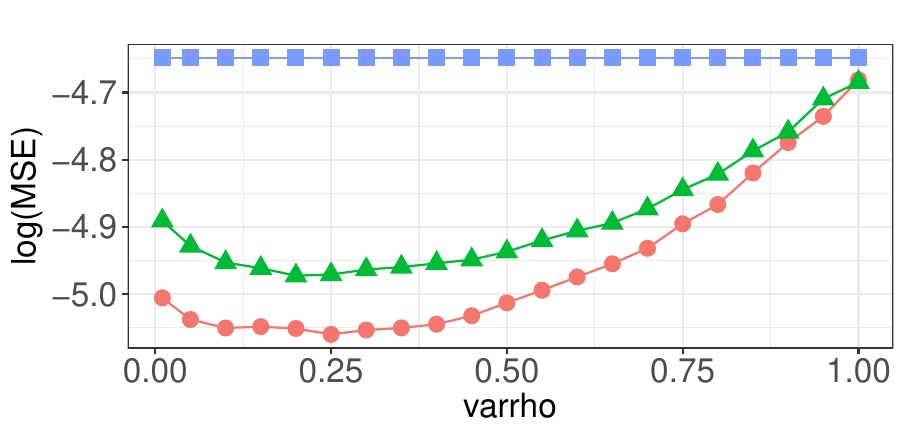}\\[-1cm]
    \caption{$r=1200$}
  \end{subfigure}
  \begin{subfigure}{0.49\textwidth}
    \includegraphics[width=\textwidth]{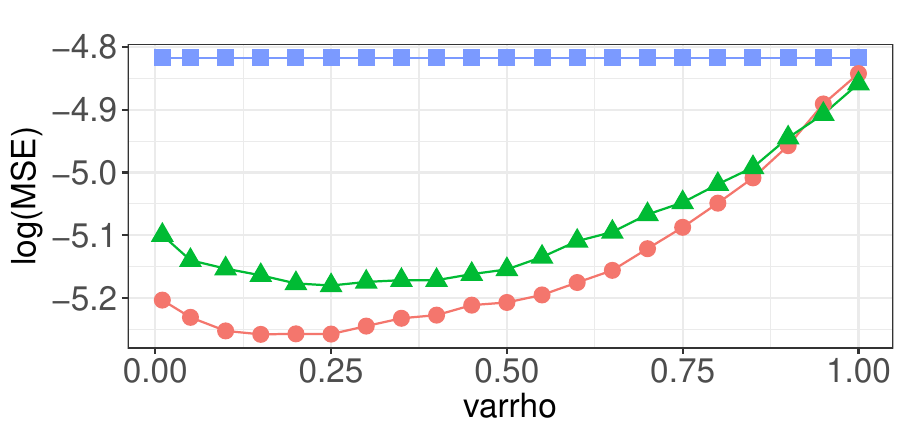}\\[-1cm]
    \caption{$r=1500$}
  \end{subfigure}\\[-0.5mm]
  \caption{Log MSEs for Case 4 with different ${\varrho}$ and a fixed $r_0=200$ based on MV (red circle), MVc (green triangle) and uniform subsampling (blue square) methods.}
  \label{fig:2}
\end{figure}

To see the effects of $M$ in $\bm{p}^{\rm sos}$, we compare the choice of $M=\infty$ with another two choices: 1) $M$ is approximated by the $(1-r/(2n))$-th quantile of $\{\hbar_i^{*\rm MV}\}_{i=1}^{r_0}$ or $\{\hbar_i^{*\rm MVc}\}_{i=1}^{r_0}$ calculated from pilot subsample set (denote this choice as $M=Q$), and 2) $M$ is calculated according to the formulas in Theorem \ref{thm:3} or \ref{thm:5} except that $\hat{\bm\beta}_{\rm QLE}$ is replaced by $\tilde{\bm\beta}_{0}$ (denote this choice as $M=E$). we  consider different values of $r/N$ with choices of 0.01, 0.1, 0.3, 0.5 and 0.7, and report results in Table \ref{tab:tab3}. %
When $r/N\le 0.3$, the choice $M=\infty$ has comparable results as the choice $M=E$ (calculating $M$ from the full). When $r/N\ge 0.5$, the choice $M=Q$ (using a quantile from the pilot subsample) still produce satisfactory results. Thus, the MSE is not very sensitive to the choice of $M$. In the big data subsampling scheme, since it is typical that $r\ll N$, we can simply use $M=\infty$.

\begin{table}[H]
  \centering
  \caption{MSE for different expected size $r$ under varying subsampling strategy with $r_0=2000$ and $\varrho=0.2$ on Case 4. Here $Q$ is the $(1-r/(2n))$-th quantile of $\{\hbar_i^{*\rm MV}\}_{i=1}^{r_0}$ or $\{\hbar_i^{*\rm MVc}\}_{i=1}^{r_0}$, and $E$ is calculated according to the formula for $M$ in Theorem \ref{thm:3} or \ref{thm:5} with $\hat{\bm\beta}_{\rm QLE}$ replaced by $\tilde{\bm\beta}_{0}$.}
  \setlength{\tabcolsep}{3.5mm}{
    \begin{tabular}{cccccc}
    \hline
    {Method} & $r/N=$0.01  & $r/N=$0.1   & $r/N=$0.3   & $r/N=$0.5   & $r/N=$0.7 \\
    \hline
        {UNIF}  & 1.75E-03 & 1.93E-04 & 5.03E-05 & 2.14E-05 & 9.21E-06 \\
        {MV with $M=\infty$}  & 1.18E-03 & 1.12E-04 & 2.35E-05 & 8.35E-06 & 3.64E-06 \\
        {MV with $M=Q$}       & 1.19E-03 & 1.11E-04 & 2.54E-05 & 8.19E-06 & 1.57E-06 \\
        {MV with $M=E$}       & 1.21E-03 & 1.16E-04 & 2.29E-05 & 7.55E-06 & 2.36E-06 \\
        {MVc with $M=\infty$}       & 1.32E-03 & 1.22E-04 & 2.82E-05 & 1.09E-05 & 5.26E-06 \\
        {MVc with $M=Q$}       & 1.28E-03 & 1.22E-04 & 2.72E-05 & 9.65E-06 & 2.12E-06 \\
        {MVc with $M=E$}       & 1.35E-03 & 1.26E-04 & 2.74E-05 & 8.76E-06 & 2.67E-06 \\
    \hline
    \end{tabular}}
\label{tab:tab3}

\end{table}

To have a closer look at the effect of $K$, we implement  Algorithm \ref{alg:4} with fixed partition number $K=5$ or $K=10$ and changing $r$ with choices of { 300, 500, 700, 1000, 1200, 1500, 1700, and 2000}. We also consider {\bluee  the cases where $r$ and $Kr$ are fixed}.%
The results for Case 4 are {\blue reported} in Figure \ref{fig:5} with $r_0= 200$ and $\varrho=0.2$. For comparisons, the uniform subsampling is also implemented through Algorithm \ref{alg:4} with $\tilde{\bm{p}}^{\rm sos}$ replaced by ${\bm{p}}^{\rm UNIF}$. Figure \ref{fig:5} shows that the subsampling method outperforms the uniform subsampling method for both $K=5$ and  $K=10$.
 If $r$ is fixed, the aggregate estimator approximates $\hat{\bm\beta}_{\rm QLE}$ better when $K$ is larger since more data are involved in each subsample set. However, when $Kr$ is fixed, as $K$ increases, the performance of the aggregate estimator deteriorates.

\begin{figure}[H]
  \centering
    \begin{subfigure}{0.49\textwidth}
    \includegraphics[width=\textwidth]{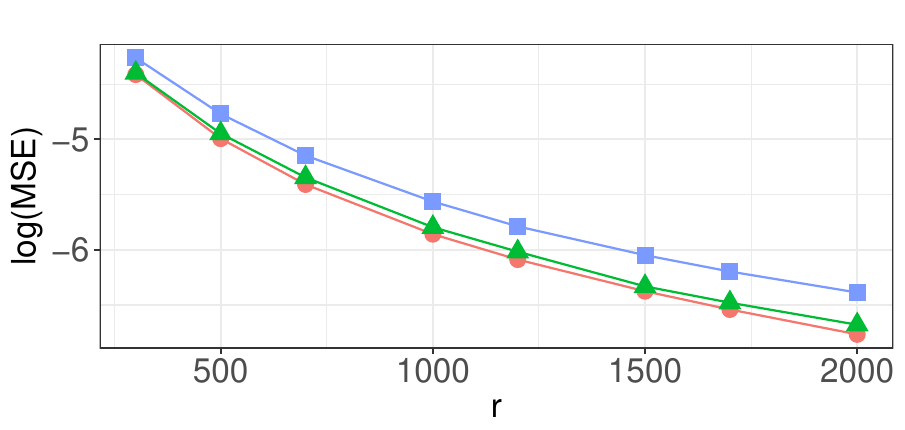}\\[-1cm]
    \caption{Case 4 (K=5) }
  \end{subfigure}
  \begin{subfigure}{0.49\textwidth}
    \includegraphics[width=\textwidth]{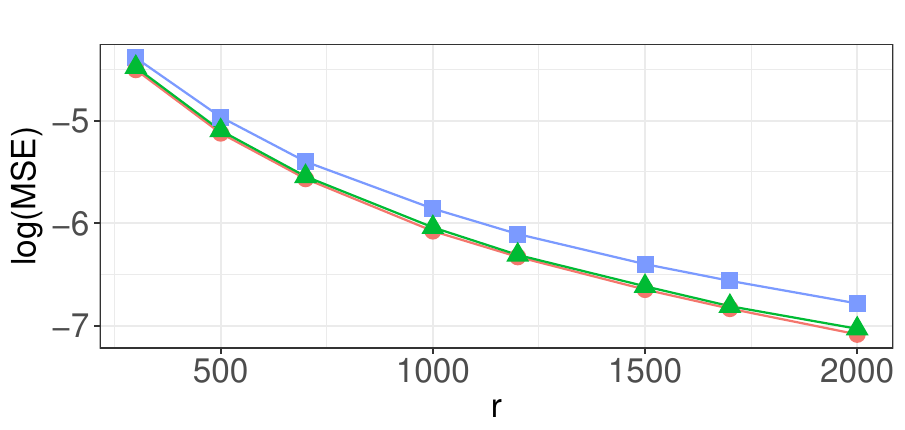}\\[-1cm]
    \caption{Case 4 (K=10)}
  \end{subfigure}\\[-0.5mm]
  \begin{subfigure}{0.49\textwidth}
    \includegraphics[width=\textwidth]{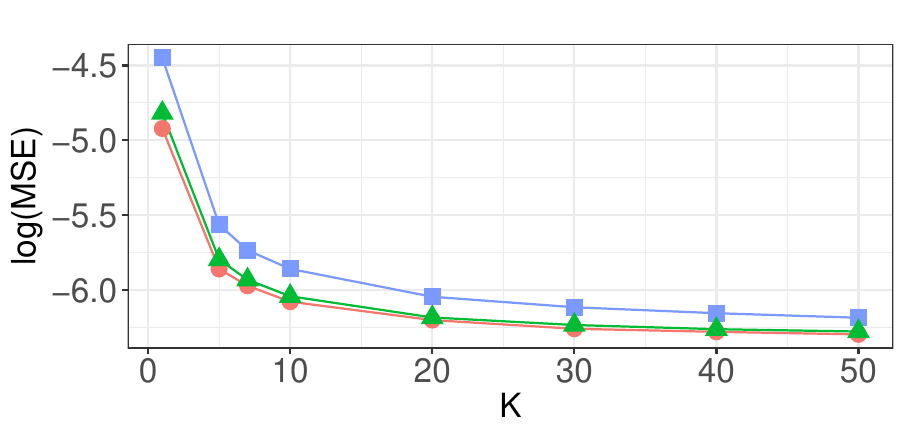}\\[-1cm]
    \caption{Case 4 (fixed $r=1000$)}
  \end{subfigure}
  \begin{subfigure}{0.49\textwidth}
    \includegraphics[width=\textwidth]{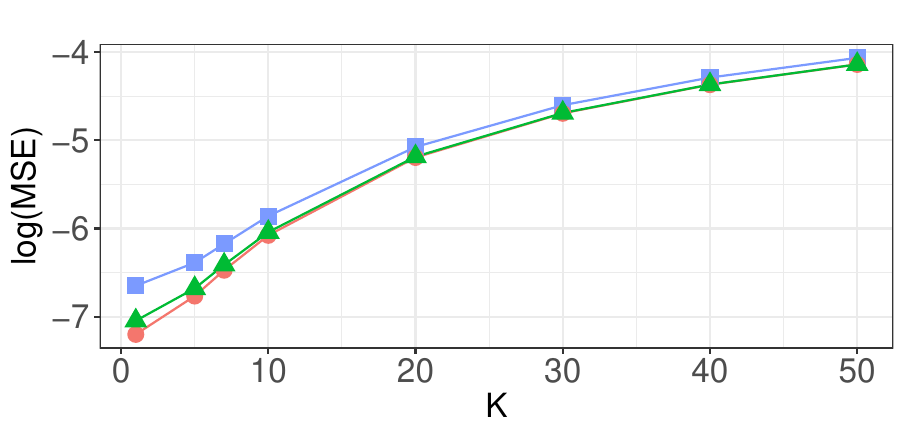}\\[-1cm]
    \caption{Case 4 (fixed $Kr=10000$)}
  \end{subfigure}\\[-0.5mm]
  \caption{Log MSEs for different combination of $r$ and $K$ with $r_0=200$ and $\varrho=0.2$ based on MV (red circle), MVc (green triangle) and uniform subsampling (blue square) methods.}
  \label{fig:5}
\end{figure}

  Now we evaluate the performance of the proposed subsampling method for statistical inference under different values of $r$ and $K$.
 As an example, we take $\beta_2$ as the parameter of interest and construct 95\% confidence intervals for it. The estimator given in \eqref{eq:24} is used to estimate the variance-covariance matrices based on selected subsamples.
Table \ref{tab:tab1} reports empirical coverage probabilities and average lengths over the four synthetic data sets with $r_0=200$ and $\varrho=0.2$.
It is clear that MV and MVc based subsampling methods have similar performances and they are uniformly better than the uniform subsampling method. As $r$ or $K$ increases, lengths of confidence intervals decrease. The 95\% confidence intervals in Case 3 are longer than those in other cases with the same subsample sizes.
 This coincides with the aforementioned results.

\begin{table}[H]
  \centering
  \caption{Empirical coverage probabilities and average lengths of 95\% confidence intervals for $\beta_2$ with $r_0=200$ and $\varrho=0.2$.}
      \setlength{\tabcolsep}{3mm}{
    \begin{tabular}{ccccccccc}
    \hline
          &       &   r    & \multicolumn{2}{c}{MV} & \multicolumn{2}{c}{MVc} & \multicolumn{2}{c}{UNIF} \\
          &       &       & coverage & length & coverage & length & coverage & length \\
          \hline
    \multirow{4}[0]{*}{case1} & \multirow{2}[0]{*}{k=1} & 1000  & 0.950  & 0.1867  & 0.949  & 0.1880  & 0.944  & 0.1932  \\
          &       & 1500  & 0.949  & 0.1829  & 0.946  & 0.1837  & 0.945  & 0.1871  \\
          & \multirow{2}[0]{*}{k=5} & 1000  & 0.947  & 0.1774  & 0.944  & 0.1776  & 0.931  & 0.1783  \\
          &       & 1500  & 0.951  & 0.1766  & 0.946  & 0.1767  & 0.935  & 0.1771  \\
          \hline
    \multirow{4}[0]{*}{case2} & \multirow{2}[0]{*}{k=1} & 1000  & 0.935  & 0.1619  & 0.944  & 0.1638  & 0.927  & 0.1680  \\
          &       & 1500  & 0.940  & 0.1587  & 0.937  & 0.1600  & 0.934  & 0.1627  \\
          & \multirow{2}[0]{*}{k=5} & 1000  & 0.945  & 0.1542  & 0.934  & 0.1546  & 0.936  & 0.1550  \\
          &       & 1500  & 0.938  & 0.1536  & 0.938  & 0.1538  & 0.932  & 0.1539  \\
          \hline
    \multirow{4}[0]{*}{case3} & \multirow{2}[0]{*}{k=1} & 1000  & 0.957  & 1.8103  & 0.956  & 1.8504  & 0.947  & 1.8961  \\
          &       & 1500  & 0.954  & 1.7788  & 0.956  & 1.8058  & 0.941  & 1.8368  \\
          & \multirow{2}[0]{*}{k=5} & 1000  & 0.956  & 1.7344  & 0.951  & 1.7418  & 0.936  & 1.7503  \\
          &       & 1500  & 0.951  & 1.7280  & 0.951  & 1.7325  & 0.945  & 1.7374  \\
          \hline
    \multirow{4}[0]{*}{case4} & \multirow{2}[0]{*}{k=1} & 1000  & 0.935  & 0.1949  & 0.928  & 0.1977  & 0.935  & 0.2034  \\
          &       & 1500  & 0.927  & 0.1913  & 0.933  & 0.1932  & 0.936  & 0.1970  \\
          & \multirow{2}[0]{*}{k=5} & 1000  & 0.928  & 0.1862  & 0.928  & 0.1865  & 0.949  & 0.1877  \\
          &       & 1500  & 0.930  & 0.1854  & 0.927  & 0.1856  & 0.946  & 0.1864  \\
          \hline
    \end{tabular}}%
  \label{tab:tab1}%
\end{table}%

Additional simulation results on both estimation efficiency and computational efficiency with larger full data sizes and higher dimensions are available in the supplementary material.

\subsection{Citation Number Data Set}
The number of citations is an important factor about the quality of a research paper, and it is of interest to most of the researchers in every field. As a result, study of paper citations itself has become an interesting research topic. In this example, we applied the proposed method to a real data set {\red about over four million} papers associated with abstract, authors, year, venue, title, type and citation numbers (\cite{Tang:08KDD}).
The data set is available at \url{https://www.aminer.cn/citation}, and our goal is to model the number of citations using features extracted from the text information about the articles.

The original data set is in text format, and we extract the following numerical features to characterize each article. First, the number of years between the year the paper was published and the year of 2018 ($x_1$). This feature describes the time effect since the citation numbers are nondecreasing in $x_1$.
We categorize the length of the abstract of each paper into  detail/brief/non-present status, and bring two indicator variables to denote them. Specifically, $x_2=1$ if the paper has an abstract with more than 100 words and $x_2=0$ otherwise; and $x_3=1$ if the paper has an abstract with less than 100 words and $x_3=0$ otherwise.
Similarly, we characterize the length of the title for each paper by defining $x_4=1$ if the title contains more than 10 words and $x_4=0$ otherwise.
We also consider the publication type, and use $x_5=1$ to denote journal papers and $x_5=0$ for the rest of papers.
To measure the influence of the journal or publisher, we use the newest SJR score ($x_6$) provided by \url{https://www.scimagojr.com}. We also consider the SJR ranking and let $x_7=1$ for journals or publishers that are marked with ``Q1'' and let $x_7=0$ otherwise.
The author information of each paper is also taken into account. We define $x_8$ as the average number of author publications for each paper, which is calculated by dividing the total number of publications from the author(s) of the paper before 2018 by the total number of author(s) in the paper.
We remove all the incomplete cases in the data set, and there are $n=2,803,027$ data points after the data cleaning.

To describe the relationship between the number of citations and the aforementioned features, a Poisson regression is used. The estimated mean model from the quasi-likelihood estimator based on the full data set is given as below:
\begin{equation*}
\blue E(Y|X)=\exp(1.32+0.39x_1+1.44x_2+1.09x_3-0.26x_4+0.03x_5+0.20x_6+0.55x_7+0.21x_8).
\end{equation*}
From the fitted model, we have the following findings. 1) A detailed abstract helps to attract more citations while a detailed title may not be popular among scholars. This may be because follow-up papers often have longer titles compared with the original paper, but they usually gain less attentions.  2) The publication type is not critical to receive high number of citations comparing with other factors.   3) SJR ranking is critical to receive higher number of citations. This is because the paper published in high quality publishers usually receive more attentions. %
4) More productive authors gain more citations since they may have more influences.

To assess the performance of the proposed method in approximating the full data estimates, we apply them on the citation data for 1000 times and report the averages of parameter estimates along with the empirical standard errors in Table~\ref{tab:citenum}. The uniform subsampling method is also implemented for comparison. In this table, $\varrho=0.2$, $r_0=800$, and $r=4400$.
It is seen that all subsampling methods produce average estimates that are close to the full data estimates. However, the proposed methods have significantly smaller empirical standard errors. %

Similar to the simulation studies, we also compare our methods with the uniform subsampling method with various sampling budget $r$ varying from 2000 to 4400 and $r_0$ being fixed at 800.  Figure~\ref{fig:citedata} shows the results on the empirical MSE.
We see that MV and MVc perform similarly and they both dominate the uniform sampling method. This pattern is similar to that in the simulation studies.

\begin{table}%
\centering
  \caption{Average estimates for the Citation number data set from the proposed methods with $\varrho=0.2$, $r_0=800$, and $r=4400$. In the table $\beta_1,\ldots,\beta_8$ are the regression coefficients for $x_1,\ldots,x_8$ respectively, and $\beta_0$ is the intercept coefficient. The numbers in the parentheses are the empirical standard errors.}
  \setlength{\tabcolsep}{2.2mm}{\blue\footnotesize
    \begin{tabular}{lrrrrrrr}
    \hline
      & \multicolumn{3}{c}{K=1}  && \multicolumn{3}{c}{K=5} \\
      \cline{2-4} \cline{6-8}
          &\multicolumn{1}{c}{UNIF}& \multicolumn{1}{c}{MV} & \multicolumn{1}{c}{MVc}   && \multicolumn{1}{c}{UNIF} & \multicolumn{1}{c}{MV} & \multicolumn{1}{c}{MVc}  \\
    \hline
    $\beta_0$ & 1.29 (0.339)  & 1.35 (0.183)  & 1.35 (0.213)  && 1.47 (0.221)       & 1.44 (0.153)       & 1.44 (0.159)  \\
    $\beta_1$ & 0.41 (0.054)  & 0.39 (0.027)  & 0.39 (0.023)  && 0.38 (0.039)       & 0.38 (0.027)       & 0.38 (0.026)  \\
    $\beta_2$ & 1.46 (0.333)  & 1.42 (0.182)  & 1.42 (0.208)  && 1.37 (0.218)       & 1.40 (0.145)       & 1.40 (0.149)  \\
    $\beta_3$ & 1.11 (0.365)  & 1.07 (0.190)  & 1.07 (0.218)  && 1.02 (0.240)       & 1.05 (0.160)       & 1.05 (0.162)  \\
    $\beta_4$ & -0.25 (0.139)  & -0.26 (0.100)  & -0.26 (0.083)  && -0.25 (0.094)       & -0.25 (0.071)       & -0.25 (0.069) \\
    $\beta_5$ & 0.03 (0.159) & 0.03 (0.103) & 0.03 (0.099) && 0.04 (0.103)      & 0.04 (0.073)      & 0.04 (0.073)  \\
    $\beta_6$ & 0.21 (0.051)  & 0.21 (0.022)  & 0.20 (0.017)  && 0.21 (0.035)       & 0.21 (0.020)       & 0.21 (0.019) \\
    $\beta_7$ & 0.54 (0.185)  & 0.55 (0.114)  & 0.55 (0.111)  && 0.54 (0.117)  & 0.55 (0.093)  & 0.55 (0.094)  \\
    $\beta_8$ & 0.22 (0.044)  & 0.21 (0.023)  & 0.21 (0.017)  && 0.21 (0.027)       & 0.21 (0.017)       & 0.21 (0.016) \\
    \hline
    \end{tabular}}%
  \label{tab:citenum}%
\end{table}%

\begin{figure}[H]%
  \centering
  \begin{subfigure}{0.495\textwidth}
    \includegraphics[width=\textwidth]{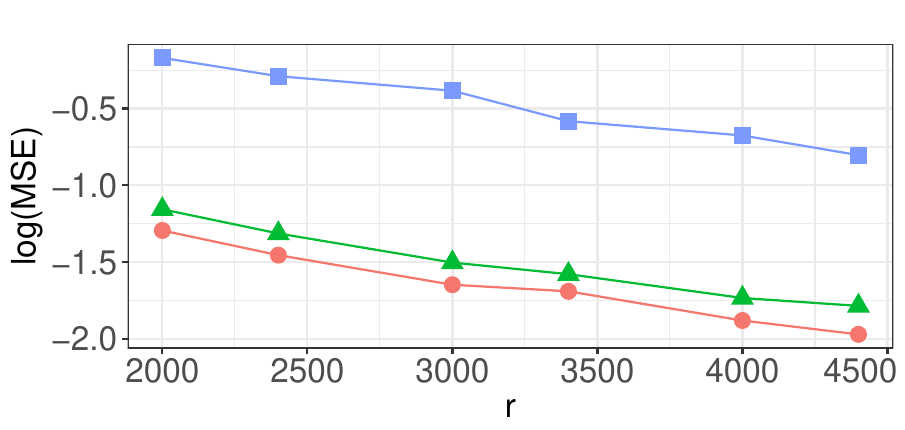}\\[-1cm]
    \caption{Log MSEs (K=1)}\label{fig:cite1}
  \end{subfigure}
  \begin{subfigure}{0.495\textwidth}
    \includegraphics[width=\textwidth]{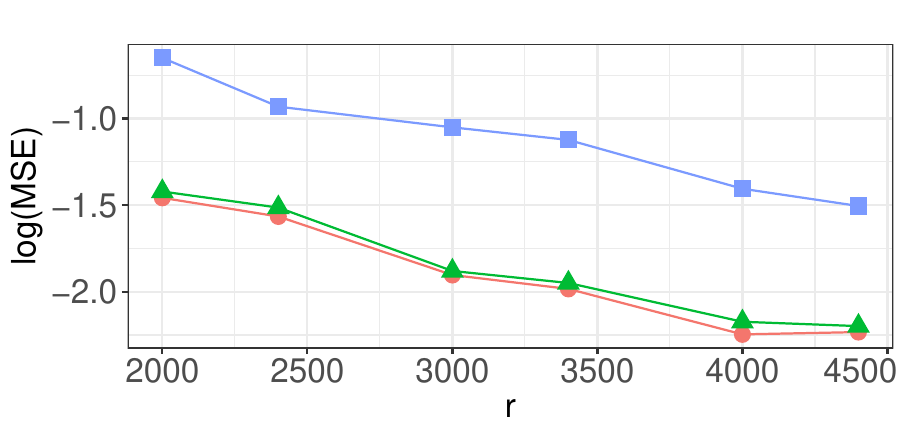}\\[-1cm]
    \caption{Log MSEs (K=5)}\label{fig:cite2}
  \end{subfigure}
  \caption{A graph showing the log of MSEs for the citation number data set with $r_0=400$ and different $r$ and partition number $K$ based on MV (red circle), MVc (green triangle) and uniform subsampling (blue square) methods. }
  \label{fig:citedata}
\end{figure}

\subsection{Airline On-time and Delay Data Set}

To track the on-time performance of domestic flights operated by large air carriers,  information about on-time, delayed, canceled, and diverted flights have been collected since October 1987. The full data set contains 123,534,969 records ($\sim$11 GB) which is available on \url{http://stat-computing.org/dataexpo/2009/the-data.html}.
One purpose for analyzing this data set is to build a model for airlines delays.
We first plot the histogram of actual arrive delays based on the pilot {\blue sample} and notice a very large discrepancy from normality. The distribution of actual delays are extremely skewed and heavy-tailed (see Figure~\ref{fig:aird1}).

In order to extract useful information about %
arrive delays, we use linear regression, log-linear regression, and Gamma regression to model the relationship between arrive delays and other covariate variables: $x_1$, the distance between airports; $x_2$, day/night status (binary; 1 if departure between 7 a.m. and 6 p.m., 0 otherwise); $x_3$, weekend/weekday status (binary; 1 if departure occurred during the weekend, 0 otherwise); and $x_4$, departure delay status (binary; 1 if the delay is 15 Minutes or More , 0 otherwise).
Note that both log linear and Gamma regression models are  defined for non-negative responses.
Thus we switch the locations of all the responses, i.e., add  1440 to all the responses.
Based on the pilot sample, the Bayesian information criterion values are 7312.672, -4407.750, and -4415.854 for linear regression, log-linear regression, and Gamma regression, respectively, which implies that the posterior probability for Gamma regression model is around 0.98 in the view of Bayesian model averaging \citep[see][]{Neath2012BIC}.
Thus we use Gamma regression for this case.
In addition, we drop the \verb"NA" values in the dataset. After data cleaning, we have $n=119,793,199$ data points.
Similar to the simulation studies, we also compare our method with the uniform subsampling method, and report the results  under various sampling budget $r$ varying from 2000 to 4400 with $r_0$ fixed at 800 in Figure~\ref{fig:aird}.
As expected, MV and MVc perform similarly and they both  outperform the uniform sampling method.

\begin{figure}[H]
  \centering
  \begin{subfigure}{0.495\textwidth}
    \includegraphics[width=\textwidth]{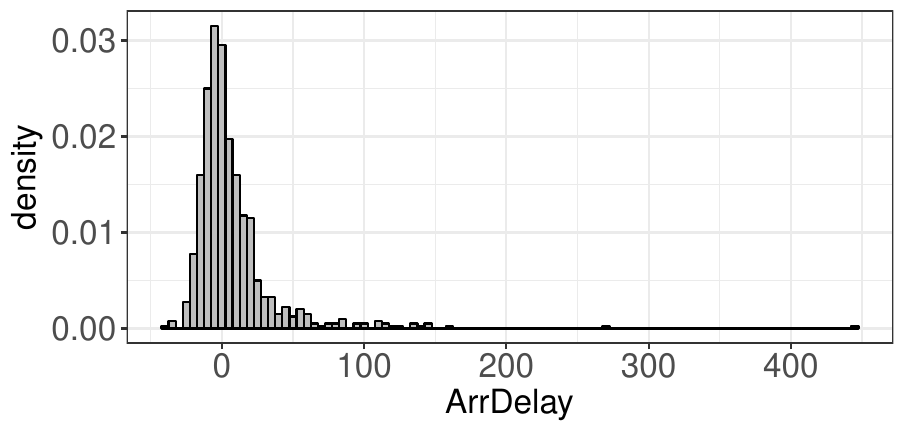}\\[-1cm]
    \caption{Actual Delays}\label{fig:aird1}
  \end{subfigure}
  \begin{subfigure}{0.495\textwidth}
    \includegraphics[width=\textwidth]{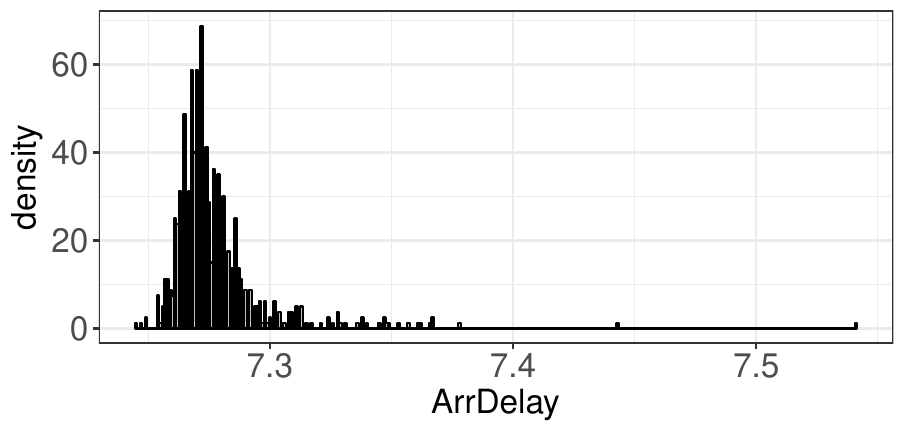}\\[-1cm]
    \caption{Log Actual Delays}\label{fig:aird2}
  \end{subfigure}
  \caption{Distribution of Actual Delays and Log-transformed Actual Delays based on the pilot samples ($r_0=800$). }
  \label{fig:airdv}
\end{figure}

\begin{figure}[H]
  \centering
  \begin{subfigure}{0.495\textwidth}
    \includegraphics[width=\textwidth]{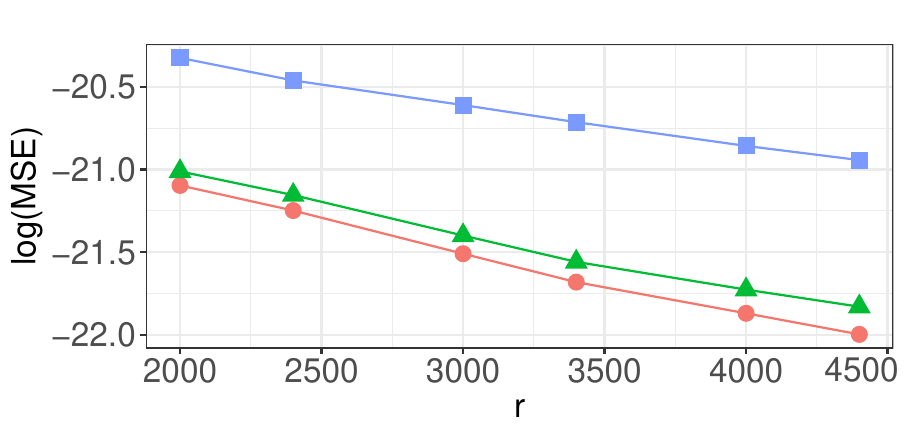}\\[-1cm]
    \caption{Log MSEs (K=1)}
  \end{subfigure}
  \begin{subfigure}{0.495\textwidth}
    \includegraphics[width=\textwidth]{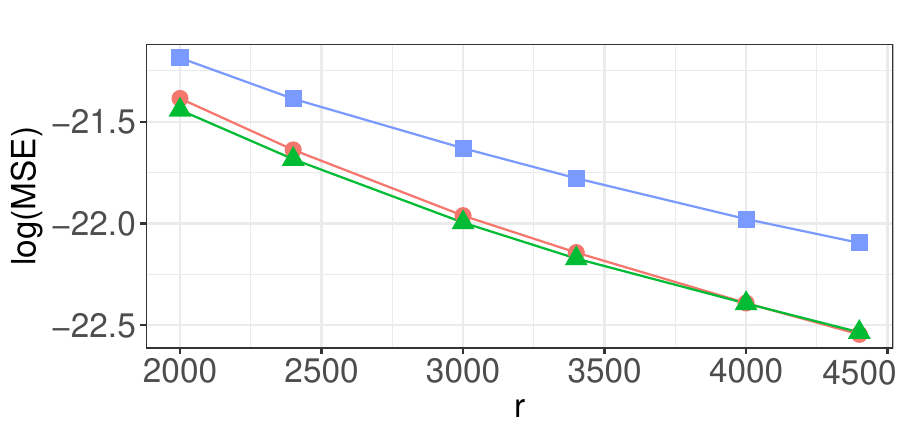}\\[-1cm]
    \caption{Log MSEs (K=5)}
  \end{subfigure}
  \caption{A graph showing the log of MSEs for the airline on-time and delay  data set with $r_0=800$, $\varrho=0.2$ and different ${r}$ and partition number $K$ based on MV (red circle), MVc (green triangle) and uniform subsampling (blue square) methods. }
  \label{fig:aird}
\end{figure}

\section{Conclusion}
In this paper, we have derived the optimal Poisson subsampling probabilities for quasi-likelihood estimation, and developed a distributed optimal subsampling method. We have investigated the theoretical properties of the proposed methods and carried out extensive numerical experiments on simulated and real data sets to evaluate their practical performance. Both theoretical results and numerical results demonstrate the great potential of the proposed method in extracting useful information from massive data sets.

\newpage
\appendix
  \begin{center}
    {\LARGE\bf Supplementary material for ``Optimal Distributed Subsampling for Maximum Quasi-Likelihood Estimators with Massive Data''}
\end{center}

\bigskip

\spacingset{1.45}
\renewcommand{\thesection}{S.\arabic{section}}
\renewcommand{\thelemma}{S.\arabic{lemma}}
\renewcommand{\theequation}{S.\arabic{equation}}
{\blue In this supplementary material we prove the theorems in the paper and present additional simulation results to evaluate the proposed methods.}

\section{Proofs}
Recall that we use  $Q(\bm\beta)$ to denote the estimating equation on full data set. For the subsamples, the weighted estimation equation (3) can be written as
\[Q^*(\bm\beta)=\sum_{i=1}^N\frac{\delta_i}{p_i}[y_i-\psi(\bm\beta^T\bm x_i)]\bm x_i,\]
where $\delta_i$ is the indicator variable that signifies whether $(\bm x_i, y_i)$ is included in the subsample.
 Denote the first derivative of $Q^*(\bm\beta)$ as $\dot{Q}^*(\bm\beta)=\partial Q^*(\bm\beta)/\partial\bm\beta$.

\subsection{Proofs of {Theorems~1 and~2}}

To prove Theorems~1 and~2, we start from proving the following lemmas.

\begin{lemma}\label{lem:lem3}
Under Assumptions 1, 2 and 5, conditional on $\mathcal{F}_N$, as $r\rightarrow\infty$ and $N\rightarrow\infty$,
\[\frac{1}{N}V_c^{-1/2}{Q}^*(\hat{\bm\beta}_{\rm QLE})\rightarrow N(\bm 0,I),\]
in distribution, where
\[V_c=\frac{1}{N^2}\sum_{i=1}^N\frac{1}{p_i}\{y_i-\psi(\hat{\bm\beta}_{\rm QLE}^T\bm x_i)\}^2\bm x_i\bm x_i^T-\frac{1}{N^2}\sum_{i=1}^N{\{y_i-\psi(\hat{\bm\beta}_{\rm QLE}^T\bm x_i)\}^2\bm x_i\bm x_i^T}.\]
\end{lemma}
\begin{proof}
Direct calculation shows that
\[\E\left\{\frac{1}{N}{Q}^*(\hat{\bm\beta}_{\rm QLE})\bigg|\mathcal{F}_N\right\}=\frac{1}{N}{Q}(\hat{\bm\beta}_{\rm QLE})=0,\]
and
\[\begin{split}
&\text{var}\left\{\frac{1}{N}Q^*(\hat{\bm\beta}_{\rm QLE})\bigg|\mathcal{F}_N\right\}\\
=&\frac{1}{N^2} \sum_{i=1}^N\frac{\text{var}(\delta_i|\mathcal{F}_N)}{p_i^2}\{y_i\bm x_i-\psi(\hat{\bm\beta}_{\rm QLE}^T\bm x_i)\bm x_i\}\{y_i\bm x_i-\psi(\hat{\bm\beta}_{\rm QLE}^T\bm x_i)\bm x_i\}^T\\
=&\frac{1}{N^2}\sum_{i=1}^N\frac{\{y_i-\psi(\hat{\bm\beta}_{\rm QLE}^T\bm x_i)\}^2\bm x_i\bm x_i^T}{p_i}-\frac{1}{N^2}\sum_{i=1}^N{\{y_i-\psi(\hat{\bm\beta}_{\rm QLE}^T\bm x_i)\}^2\bm x_i\bm x_i^T}.
\end{split}
\]

Now we check the Lindeberg-Feller condition under the conditional distribution.
Denote $\eta_i=\delta_i\{y_i-\psi(\hat{\bm\beta}_{\rm QLE}^T\bm x_i)\}\bm x_i/(Np_i)$.
For every $\varepsilon>0$,
\begin{equation}\label{eq:dbclt}
\begin{aligned}
\sum^N_{i=1} \E\{\|\eta_i\|^2
    \mathbbm{1}_{\|\eta_i\|>\varepsilon}\mid\mathcal{F}_N\} %
  &\le\frac{1}{\varepsilon}
    \sum^N_{i=1} \E(\|\eta_i\|^{3}\mid\mathcal{F}_N)\\
  &= \frac{1}{N^{3}}\frac{1}{\varepsilon}
    \sum^N_{i=1}\frac{\|y_i-{\psi}_i(\hat{\bm\beta}_{\rm QLE}^T\bm x_i)\|^{3}
    \|\bm x_i\|^{3}}{p_i^{2}}\\
  &\le \frac{1}{\varepsilon}
\left\{\max_{i=1,\ldots,N}\frac{1}{(Np_i)^{2}}\right\}\sum^N_{i=1}\frac{\|y_i-{\psi}_i(\hat{\bm\beta}_{\rm QLE}^T\bm x_i)\|^{3}\|\bm x_i\|^{3}}{N},
\end{aligned}
\end{equation}
where $\mathbbm{1}_{\cdot}$ is the indicator function.

Now we show that ${N}^{-1}\sum_{i=1}^N{\|y_i-{\psi}(\hat{\bm\beta}_{\rm QLE}^T\bm x_i)\|^3\|\bm x_i\|^3}={\blue O_{P}(1)}$. Note that
\begin{equation}\label{eq:s2}
\begin{split}
\sum\limits_{i = 1}^N {\frac{{{{| {y_i} - \psi (\hat{\bm\beta}_{\rm QLE}^T{\bm x_i})| }^3}{{\| {{\bm x_i}} \|}^3}}}{N}}  &\le \sum\limits_{i = 1}^N {\frac{{|y_i|^3{{\| {{\bm x_i}} \|}^3}}}{N}}  + 3\sum\limits_{i = 1}^N {\frac{{{y_i}^2\psi (\hat{\bm\beta}_{\rm QLE}^T{\bm x_i}){{\| {{\bm x_i}} \|}^3}}}{N}}\\
 &+3\sum\limits_{i = 1}^N {\frac{{|y_i|\psi^2(\hat{\bm\beta}_{\rm QLE}^T{\bm x_i}){{\| {{\bm x_i}} \|}^3}}}{N}}
+ \sum\limits_{i = 1}^N {\frac{{{\psi ^3}(\hat{\bm\beta}_{\rm QLE}^T{\bm x_i}){{\| {{\bm x_i}} \|}^3}}}{N}}.
 \end{split}
 \end{equation}
{\blue From (i) and (ii) in Assumption 2, we have  $N^{-1}\sum_{i=1}^N\|\bm x_i\|^6=O_P(1)$ and $N^{-1}\sum_{i=1}^N y_i^6=O_P(1)$ from the law of large numbers.
 Thus, from Holder's inequality, we have %
\begin{align}\label{eq:s3}
  &\sum\limits_{i = 1}^N {\frac{{|y_i|^3{{\| {{\bm x_i}}\|}^3}}}{N}}
\le \sqrt {\sum\limits_{i = 1}^N {\frac{{y_i^6}}{N}} } \sqrt {\sum\limits_{i = 1}^N {\frac{{{{\| {{\bm x_i}} \|}^6}}}{N}} }=O_P(1).
\end{align}}
Similarly, under Assumption 2, it can be shown that
\begin{eqnarray}
&& \sum\limits_{i = 1}^N {\frac{{{\psi ^3}(\hat{\bm\beta}_{\rm QLE}^T{\bm x_i}){{\| {{\bm x_i}} \|}^3}}}{N}} = O_P(1), \label{eq:s4}\\
&&  \sum\limits_{i = 1}^N {\frac{{{y_i}^2\psi (\hat{\bm\beta}_{\rm QLE}^T{\bm x_i}){{\| {{\bm x_i}} \|}^3}}}{N}} = O_P(1), \label{eq:s5}\\
&&  \sum\limits_{i = 1}^N {\frac{{{|y_i|}\psi^2 (\hat{\bm\beta}_{\rm QLE}^T{\bm x_i}){{\| {{\bm x_i}} \|}^3}}}{N}} = O_P(1).\label{eq:s6}
\end{eqnarray}
Here, the last two equalities come from the generalized Holder inequality \citep[See][Page 133]{schilling2017measures},
\begin{equation}\label{eq:holder}
  \frac{1}{N}\sum\limits_{i = 1}^N {| {{a_i}{b_i}{c_i}}|}  \le {\left(\frac{1}{N}\sum\limits_{i = 1}^N {| {a_i^3}|} \right)^{1/3}}{\left(\frac{1}{N}\sum\limits_{i = 1}^N {| {b_i^3} |} \right)^{1/3}}{\left(\frac{1}{N}\sum\limits_{i = 1}^N {| {c_i^3} |} \right)^{1/3}}.
\end{equation}
To be specific, the results come from the fact that
\begin{align*}
  \sum\limits_{i = 1}^N {\frac{{{y_i}^2\psi (\hat{\bm\beta}_{\rm QLE}^T{\bm x_i}){{\| {{\bm x_i}} \|}^3}}}{N}}
  &\le {\left(\sum\limits_{i = 1}^N {\frac{{y_i^6}}{N}} \right)^{1/3}}{\left(\sum\limits_{i = 1}^N {\frac{{\blue|{{\psi ^3}(\hat{\bm\beta}_{\rm QLE}^T{\bm x_i})}|}}{N}} \right)^{1/3}}{\left(\sum\limits_{i = 1}^N {\frac{{{{\| {{\bm x_i}} \|}^9}}}{N}} \right)^{1/3}}\\
  &\le {\left(\sum\limits_{i = 1}^N {\frac{{y_i^6}}{N}} \right)^{1/3}}{\left(\sum\limits_{i = 1}^N {\frac{{\blue{{\psi ^6}(\hat{\bm\beta}_{\rm QLE}^T{\bm x_i})}}}{N}} \right)^{1/6}}{\left(\sum\limits_{i = 1}^N {\frac{{{{\| {{\bm x_i}} \|}^9}}}{N}} \right)^{1/3}}
\end{align*}
and
\[\sum\limits_{i = 1}^N {\frac{{{|y_i|}\psi^2 (\hat{\bm\beta}_{\rm QLE}^T{\bm x_i}){{\| {{\bm x_i}} \|}^3}}}{N}}  \le {\left(\sum\limits_{i = 1}^N {\frac{{|y_i|^3}}{N}} \right)^{1/3}}{\left(\sum\limits_{i = 1}^N {\frac{{{\psi ^6}(\hat{\bm\beta}_{\rm QLE}^T{\bm x_i})}}{N}} \right)^{1/3}}{\left(\sum\limits_{i = 1}^N {\frac{{{{\| {{\bm x_i}} \|}^9}}}{N}} \right)^{1/3}}.\]

{\purple Combining (\ref{eq:s2}), (\ref{eq:s3}), (\ref{eq:s4}), (\ref{eq:s5}) and  (\ref{eq:s6})}, we have %
\begin{equation}\label{eq:ax-thm1}
  N^{-1}\sum\limits_{i = 1}^N {{{{{| {y_i} - \psi (\hat{\bm\beta}_{\rm QLE}^T{\bm x_i})| }^3}{{\| {{\bm x_i}} \|}^3}}}} =O_P(1).
\end{equation}

From \eqref{eq:dbclt}, \eqref{eq:ax-thm1}, and Assumption 5, we obtain
\begin{align*}
\sum^N_{i=1} E\{\|\eta_i\|^2
    \mathbbm{1}_{\|\eta_i\|>\varepsilon}|\mathcal{F}_N\}  \le \frac{1}{\varepsilon}
{\blue O_P(r^{-2})}\cdot O_P(1)=o_P(1).
\end{align*}
Thus, conditionally on $\mathcal{F}_N$, the desired result holds by the Lindeberg-Feller central limit theorem \citep[Proposition 2.27 of][]{Vaart2000Asymptotic}.
\end{proof}

\begin{lemma}\label{lem:lem4}
Under Assumptions 1 -- 5, as $N,r\rightarrow\infty$, for any $\bm s_r\rightarrow0$ in probability,
\begin{equation}\label{eq:lem4eq}
  \frac{1}{N}\sum_{i=1}^N\frac{\delta_i}{p_i}\dot{\psi}((\hat{\bm\beta}_{\rm QLE}+\bm s_r)^T\bm x_i)\bm x_i\bm x_i^T-\frac{1}{N}\sum_{i=1}^N\dot{\psi}(\hat{\bm\beta}_{\rm QLE}^T\bm x_i)\bm x_i\bm x_i^T=o_{P|\mathcal{F}_N}(1).
\end{equation}
\end{lemma}

\begin{proof}
Direct calculation shows that conditionally on $\mathcal{F}_N$,
\[\E\left\{\frac{1}{N}\sum_{i=1}^N\frac{\delta_i}{p_i}\dot{\psi}(\hat{\bm\beta}_{\rm QLE}^T\bm x_i)\bm x_i\bm x_i^T\bigg|\mathcal{F}_N\right\}=\frac{1}{N}\sum_{i=1}^N\dot{\psi}(\hat{\bm\beta}_{\rm QLE}^T\bm x_i)\bm x_i\bm x_i^T.\]
Let $\Sigma_{\psi,S}(\hat{\bm\beta}_{\rm QLE})=N^{-1}\sum_{i=1}^N{\delta_i}{p_i}^{-1}\dot{\psi}(\hat{\bm\beta}_{\rm QLE}^T\bm x_i)\bm x_i\bm x_i^T$, for any component $\Sigma_{\psi,S}(\hat{\bm\beta}_{\rm QLE})^{j_1j_2}$ of $\Sigma_{\psi,S}(\hat{\bm\beta}_{\rm QLE})$ where $1\le j_1,j_2\le p$,
\begin{align*}
  &{ E}\left\{\Sigma_{\psi,S}(\hat{\bm\beta}_{\rm QLE})^{j_1j_2}-\Sigma_\psi(\hat{\bm\beta}_{\rm QLE})^{j_1j_2}\Big|\mathcal{F}_N\right\}^2\\
  =&\sum^N_{i=1}\frac{p_i(1-p_i)}{p_i^2}
     \left\{\frac{\dot{\psi}(\hat{\bm\beta}_{\rm QLE}^T\bm x_i)x_{ij_1} x_{ij_2}}{N}\right\}^2\\
  =&\red\sum^N_{i=1}\frac{1}{p_i}
     \left\{\frac{\dot{\psi}(\hat{\bm\beta}_{\rm QLE}^T\bm x_i)x_{ij_1} x_{ij_2}}{N}\right\}^2
     -\sum_{i=1}^N\left\{\frac{\dot{\psi}(\hat{\bm\beta}_{\rm QLE}^T\bm x_i)x_{ij_1} x_{ij_2}}{N}\right\}^2\\
  \le&\red\sum^N_{i=1}\frac{1}{p_i}
     \left\{\frac{\dot{\psi}(\hat{\bm\beta}_{\rm QLE}^T\bm x_i)x_{ij_1} x_{ij_2}}{N}\right\}^2\\
\le&\red\left(\max_{i=1,\ldots,N}\frac{1}{Np_i}\right)\sum_{i=1}^N\frac{\dot{\psi}^2(\hat{\bm\beta}_{\rm QLE}^T\bm x_i)(x_{ij_1} x_{ij_2})^2}{N},
\\
  \le&\red\left(\max_{i=1,\ldots,N}\frac{1}{Np_i}\right)\sum_{i=1}^N\frac{\dot{\psi}^2(\hat{\bm\beta}_{\rm QLE}^T\bm x_i)\|\bm x_i\|^4}{N},
\end{align*}
{\red where the last equality is because {\purple $(x_{ij_1} x_{ij_2})^2\le x_{ij_1}^4+x_{ij_2}^4\le \|\bm x_i\|^4$}. } %
Utilizing (i), (iii) and (iv) in Assumption 2, it can be shown that
\[\sum_{i=1}^N\frac{\dot{\psi}^2(\hat{\bm\beta}_{\rm QLE}^T\bm x_i)\|\bm x_i\|^{\red 4}}{N}=O_P(1),\]
 by arguments similar to those used for Lemma \ref{lem:lem3}.
 Thus we have
 \[{ E}\left(\Sigma_{\psi,S}(\hat{\bm\beta}_{\rm QLE})^{j_1j_2}-\Sigma_\psi(\hat{\bm\beta}_{\rm QLE})^{j_1j_2}\Big|\mathcal{F}_N\right)^2=O_P(r^{-1})\]
 from Assumption 5.
 {\red Note that $\Sigma_\psi(\hat{\bm\beta}_{\rm QLE})^{j_1j_2}=O_P(1)$ under Assumption 3. }%
It is proved that
\[\Upsilon_{I}:=\frac{1}{N}\sum_{i=1}^N\frac{\delta_i}{{p}_i}\dot{\psi}(\hat{\bm\beta}_{\rm QLE}^T\bm x_i)\bm x_i\bm x_i^T-\frac{1}{N}\sum_{i=1}^N\dot{\psi}(\hat{\bm\beta}_{\rm QLE}^T\bm x_i)\bm x_i\bm x_i^T=o_{P|\mathcal{F}_N}(1),\]
from Chebyshev's inequality.

It remains to show that
\[\Upsilon_{II}:=\frac{1}{N}\sum_{i=1}^N\frac{\delta_i}{{p}_i}\{\dot{\psi}((\hat{\bm\beta}_{\rm QLE}+\bm s_r)^T\bm x_i)-\dot{\psi}(\hat{\bm\beta}_{\rm QLE}^T\bm x_i)\}\bm x_i\bm x_i^T=o_{P|\mathcal{F}_N}(1).\]
 {\purple According to Wely's theorem \citep[][Theorem 4.3.1]{Roger2013matrix} %
\[\begin{split}
\|\Upsilon_{II}\|_s&\le \frac{1}{N}\sum_{i=1}^N\frac{\delta_i}{{p}_i}\|\dot{\psi}((\hat{\bm\beta}_{\rm QLE}+\bm s_r)^T\bm x_i)-\dot{\psi}(\hat{\bm\beta}_{\rm QLE}^T\bm x_i)\}\bm x_i\bm x_i^T\|_s
\\
&\le\frac{1}{N}\sum_{i=1}^N\|\{\dot{\psi}((\hat{\bm\beta}_{\rm QLE}+\bm s_r)^T\bm x_i)-\dot{\psi}(\hat{\bm\beta}_{\rm QLE}^T\bm x_i)\}\bm x_i\bm x_i^T\|_s\\
&\le \frac{1}{N}\sum_{i=1}^N\frac{\delta_i}{{p}_i} m_2(\bm x_i) \|\bm s_r\| =\red o_{P}(1),
\end{split}\]
under Assumption 4. } %
Thus the result follows from the fact that $\|\Upsilon_{II}\|_s\ge0$.

\end{proof}

Now we are ready to prove Theorems 1 and 2.
\begin{proof}
The estimator $\tilde{\bm\beta}$ is the solution of
\begin{equation*}
  Q^*(\bm\beta)=\sum_{i=1}^N\frac{\delta_i}{p_i}\Big\{y_i-\psi(\bm\beta^T\bm x_i)\Big\}\bm x_i.
\end{equation*}

Note that
\begin{align}
 \E\left\{\frac{1}{N}Q^*({\bm\beta})\bigg|\mathcal{F}_N\right\}&=\frac{1}{N}Q({\bm\beta}), \label{eq:qstar1}\\
\nonumber \text{var}\left\{\frac{1}{N}Q^*({\bm\beta})\bigg|\mathcal{F}_N\right\}&=\frac{1}{N^2}\sum_{i=1}^N\frac{\{y_i-\psi({\bm\beta}^T\bm x_i)\}^2\bm x_i\bm x_i^T}{p_i}-\frac{1}{N^2}\sum_{i=1}^N\{y_i-\psi({\bm\beta}^T\bm x_i)\}^2\bm x_i\bm x_i^T\\
&\le \frac{2}{N^2}\sum_{i=1}^N\frac{\{y_i-\psi({\bm\beta}^T\bm x_i)\}^2\bm x_i\bm x_i^T}{p_i}
=O_{P|\mathcal{F}_N}(r^{-1}), \label{eq:qstar2}
\end{align}
by using the similar arguments in Lemma \ref{lem:lem3} under Assumptions 1, 2 and 5.

Therefore, as $r\rightarrow\infty$, $N^{-1}Q^*(\bm\beta)-N^{-1}Q(\bm\beta)\rightarrow0$ for all $\bm\beta\in\Lambda$ in conditional probability given $\mathcal{F}_N$.
Note that the parameter space is compact and $\hat{\bm\beta}_{\rm QLE}$ is the unique solution of $N^{-1}Q(\bm\beta)=\bm 0$ under Assumption 3  \citep[cf.][]{Tzavelas1998note}. Thus, from Theorem 5.9 and its remark of \cite{Vaart2000Asymptotic},  we have
\begin{equation}\label{eq:asy1}
  \|\tilde{\bm\beta}-\hat{\bm\beta}_{\rm QLE}\|=o_{P\mid\mathcal{F}_N}(1),
\end{equation}
as $N\rightarrow\infty, r\rightarrow\infty$, conditionally on $\mathcal{F}_N$ in probability.

By Taylor's expansion,
\begin{align}\label{eq:taylor}
  Q^*(\tilde{\bm\beta})
  &=Q^*(\hat{\bm\beta}_{\rm QLE})
    +\sum_{i=1}^N
    \frac{\delta_i}{p_i}\left(
    \begin{array}{c}
      \dot{\psi}(\acute{\bm{\beta}}_{(1)}^T\bm x_i)\red x_{i1} \\
      \vdots \\
      \dot{\psi}(\acute{\bm{\beta}}_{(d)}^T\bm x_i)\red x_{id} \\
    \end{array}
                        \right)
    \bm x_i^T(\tilde{\bm\beta}-\hat{\bm\beta}_{\rm QLE}),
\end{align} %
where all $\acute{\bm \beta}_{(1)},\ldots,\acute{\bm \beta}_{(d)}$ lie between $\hat{\bm\beta}_{\rm QLE}$ and $\tilde{\bm\beta}$.

{\red From (\ref{eq:asy1}), %
  for each $j=1,\ldots,d$, $\acute{\bm \beta}_{(j)}$ can be written as $\hat{\bm\beta}_{\rm QLE}+{\bm s_{(j)}}$ for some ${\bm s_{(j)}}=o_{P|\mathcal{F}_N}(1)$. Thus, from Lemma \ref{lem:lem4}, {\purple %
  \[{N}^{-1}\sum_{i=1}^N\frac{\delta_i}{p_i}\dot{\psi}(\acute{\bm \beta}_{(j)}^T\bm x_i)\bm x_i\bm x_i^T-{N}^{-1}\sum_{i=1}^N\dot{\psi}(\hat{\bm\beta}_{\rm QLE}^T\bm x_i)\bm x_i\bm x_i^T=o_{P|\mathcal{F}_N}(1),\]
  which implies that for every $j$,}
\[\frac{1}{N}\sum_{i=1}^N
    \frac{\delta_i}{p_i}\left(
      \begin{array}{c}
        \vdots \\
        \dot{\psi}(\acute{\bm{\beta}}_{(j)}^T\bm x_i) x_{ij} \\
        \vdots \\
      \end{array}
                        \right)
    \bm x_i^T= \frac{1}{N}\sum_{i=1}^N\left(
      \begin{array}{c}
        \vdots \\
        \dot{\psi}(\hat{\bm{\beta}}_{\rm QLE}^T\bm x_i) x_{ij} \\
        \vdots \\
      \end{array}
                        \right)
    \bm x_i^T+o_{P|\mathcal{F}_N}(1).\]}
{\red Therefore, %
\[\frac{1}{N}\sum_{i=1}^N
    \frac{\delta_i}{p_i}\left(
      \begin{array}{c}
        \dot{\psi}(\acute{\bm{\beta}}_{(1)}^T\bm x_i) x_{i1} \\
        \vdots \\
        \dot{\psi}(\acute{\bm{\beta}}_{(d)}^T\bm x_i) x_{id} \\
      \end{array}
                        \right)
    \bm x_i^T= \frac{1}{N}\sum_{i=1}^N\dot{\psi}(\hat{\bm\beta}_{\rm QLE}^T\bm x_i)\bm x_i\bm x_i^T+o_{P|\mathcal{F}_N}(1).\]}
{\blue By the definition of the quasi-likelihood estimator, the left-hand-side of (\ref{eq:taylor}) is zero. Thus,
\begin{equation}\label{eq:asy2}
\begin{split}
  \tilde{\bm\beta}-\hat{\bm\beta}_{\rm QLE}&=-\Sigma_\psi(\hat{\bm\beta}_{\rm QLE})^{-1}\frac{1}{N}Q^*(\hat{\bm\beta}_{\rm QLE})+o_{P|\mathcal{F}_N}(\|\tilde{\bm\beta}-\hat{\bm\beta}_{\rm QLE}\|)\\
  &=-\Sigma_\psi(\hat{\bm\beta}_{\rm QLE})^{-1}V_c^{1/2}V_c^{-1/2}\frac{1}{N}Q^*(\hat{\bm\beta}_{\rm QLE})+o_{P|\mathcal{F}_N}(\|\tilde{\bm\beta}-\hat{\bm\beta}_{\rm QLE}\|)\\
  &=O_{P|\mathcal{F}_N}(r^{-1/2})+o_{P|\mathcal{F}_N}(\|\tilde{\bm\beta}-\hat{\bm\beta}_{\rm QLE}\|),
\end{split}
\end{equation}
since $\Sigma_\psi(\hat{\bm\beta}_{\rm QLE})^{-1}=O_{P|\mathcal{F}_N}(1)$ under Assumption 3, $V_c^{1/2}=O_{P|\mathcal{F}_N}(r^{-1/2})$ from (\ref{eq:qstar2}) and $V_c^{-1/2}\frac{1}{N}Q^*(\hat{\bm\beta}_{\rm QLE})=O_{P|\mathcal{F}_N}(1)$ from Lemma 1.}
Therefore, $\tilde{\bm\beta}-\hat{\bm\beta}_{\rm QLE}+o_{P|\mathcal{F}_N}(\|\tilde{\bm\beta}-\hat{\bm\beta}_{\rm QLE}\|)=O_{P|\mathcal{F}_N}(r^{-1/2})$,
which implies that
\[\tilde{\bm\beta}-\hat{\bm\beta}_{\rm QLE}=O_{P|\mathcal{F}_N}(r^{-1/2}).\]%

For Theorem 2, applying (\ref{eq:asy2}),
 as $r\rightarrow\infty$, conditional on $\mathcal{F}_N$, it holds that
\[V^{-1/2}(\tilde{\bm\beta}-\hat{\bm\beta}_{\rm QLE})=-V^{-1/2}\Sigma_\psi(\hat{\bm\beta}_{\rm QLE})^{-1}V_c^{1/2}V_c^{-1/2}N^{-1}Q^*(\hat{\bm\beta}_{\rm QLE})+o_{P|\mathcal{F}_N}(1).\]
Thus, the result follows from Lemma \ref{lem:lem3} and Slutsky's theorem.
\end{proof}

\subsection{Proofs of~{Theorems~3~and~4}}

\begin{proof}
 If some elements of $\{\hbar_i\}_{i=1}^N$ are equal to zero, we can set the corresponding subsampling  probabilities as zero and then consider the subsampling probabilities among the rest. Thus, without loss of generality, we assume all $\hbar_i>0$.

In order to minimize the asymptotic mean square error, $\mathrm{tr}(V)$ in (6), it is sufficient to solve the following optimization problem:
\begin{eqnarray}\label{eq:opt-q}
  \min & & \tilde H=\sum_{i=1}^N\text{tr}\left[\frac{1}{p_i}
    \{y_i-{\psi}(\hat{\bm\beta}_{\rm QLE}^T\bm x_i)\}^2\|\Sigma_\psi(\hat{\bm\beta}_{\rm QLE})^{-1}\bm x_i\|^2\right] \\
 \nonumber  \text{s.t} & & \sum_{i=1}^N p_i=r, \qquad
0\le p_i\le 1 \quad\text{for}\quad i=1,\ldots,N.
\end{eqnarray}

 For brevity, we {denote  $\hbar_i^{\rm MV}:=|y_i-{\psi}(\hat{\bm\beta}_{\rm QLE}^T\bm x_i)|\|\Sigma_\psi(\hat{\bm\beta}_{\rm QLE})^{-1}\bm x_i\|$ as $\hbar_i$ for $i=1,\ldots,N$.
Without loss of generality, we further assume that $\hbar_1\le\hbar_2\le\cdots\le\hbar_N$.

{From the Cauchy-Schwarz inequality,
\begin{align*}
  \tilde{H}&=\sum_{i=1}^N\left[\frac{1}{p_i}\{y_i-{\psi}(\hat{\bm\beta}_{\rm QLE}^T\bm x_i)\}^2
    \|\Sigma_\psi(\hat{\bm\beta}_{\rm QLE})^{-1}\bm x_i\|^2\right]\\
  &=\frac{1}{r}\left(\sum_{j=1}^Np_j\right)\left(\sum_{i=1}^N\left[p_i^{-1}\hbar_i^2\right]\right)
  \ge\frac{1}{r}\left[\sum_{i=1}^N\hbar_i\right]^2,
\end{align*}
where  the equality in it holds if and only if  $p_i\propto \hbar_i$.}
Therefore, when $p_i= r\hbar_i/(\sum_{j=1}^N\hbar_j)$ satisfies that $p_i\le 1$ for all $i=1,\ldots,N$, $p_i$'s give the optimal solution.

Otherwise, we can  easily see that $p_N=1$ when  $r\hbar_N/(\sum_{j=1}^N\hbar_1)>1$.
Thus, the original problem (\ref{eq:opt-q}) turns into finding $p_1,\ldots,p_{N-1}$ which solve the following optimization problem:
\begin{eqnarray*}
  \min & & \sum_{i=1}^{N-1}\text{tr}\left[\frac{1}{p_i}
    \{y_i-{\psi}(\hat{\bm\beta}_{\rm QLE}^T\bm x_i)\}^2\|\Sigma_\psi(\hat{\bm\beta}_{\rm QLE})^{-1}\bm x_i\|^2\right] \\
 \nonumber  \text{s.t} & & \sum_{i=1}^{N-1} p_i=r-1,\qquad
0\le p_i\le 1 \quad\text{for}\quad i=1,\ldots,N-1.
\end{eqnarray*}
Obviously, this is a typical recursion problem, and the optimal solution is $\min\tilde{H}=\sum_{i=N-k+1}^N\hbar_i^2+(r-k)^{-1}(\sum_{i=1}^{N-k}\hbar_i)^2$, for some $k$ satisfying
\[\frac{(r-k+1)\hbar_{N-k+1}}{\sum_{i=1}^{N-k+1}\hbar_i}\ge 1\quad\text{and}\quad\frac{(r-k)\hbar_{N-k}}{\sum_{i=1}^{N-k}\hbar_i}< 1.\]

Suppose that $M$ exists such that
 \begin{equation*}
  \max_{i=1,\ldots,N}\frac{\hbar_i\wedge M}{\sum_{j=1}^N\hbar_j\wedge M}=\frac{1}{r},
   \end{equation*}
   and $\hbar_{N-k}< M\le \hbar_{N-k+1}$.
   It follows that $\sum_{i=1}^{N-k}\hbar_i=(r-k)M$.
   Therefore $\min\tilde{V}=\sum_{i=N-k+1}^N\hbar_i^2+(r-k)M^2$.

  Substituting $p_i^{\rm MV}={(\sum_{j=1}^N\hbar_j\wedge M)^{-1}}r{(\hbar_i\wedge M)}$ into \eqref{eq:opt-q}, the following equation holds.
  \[\begin{split}
  \tilde H &=\sum_{i=N-k+1}^N\hbar_i^2+\frac{1}{r}\left(\sum_{i=1}^{N-k}\hbar_i\right)^2+\frac{1}{r}\left(\sum_{i=N-k+1}^NM\right)\left(\sum_{i=1}^{N-k}\hbar_i\right)\\
  &=\sum_{i=N-k+1}^N\hbar_i^2 + \frac{1}{r}(r-k)^2M^2+\frac{1}{r}(r-k)kM^2\\
  &=\sum_{i=N-k}^N\hbar_i^2+(r-k)M^2 =\min\tilde{H}.
  \end{split}\]
  Thus, $p_i^{\rm MV}$ is the optimal solution of \eqref{eq:opt-q}.

  Now we will show that $M$ exists and satisfying $\hbar_{N-k}< M\le \hbar_{N-k+1}$.
  Note that $k$ satisfies
  \[\frac{(r-k+1)\hbar_{N-k+1}}{\sum_{i=1}^{N-k+1}\hbar_i}\ge 1\quad\text{and}\quad\frac{(r-k)\hbar_{N-k}}{\sum_{i=1}^{N-k}\hbar_i}< 1.\]
  The following inequality holds by fetching $M=\hbar_{N-k+1}$,
  \[\frac{(r-k+1)\hbar_{N-k+1}+(k-1)M}{\sum_{i=1}^{N-k+1}\hbar_i+(k-1)M}\ge 1.\]
  This implies $(\hbar_N\wedge M)/(\sum_{j=1}^N\hbar_N\wedge M)\ge 1/r$.
  Similarly, we know $(\hbar_N\wedge M)/(\sum_{j=1}^N\hbar_j\wedge M)<1/r$ by fetching $M=\hbar_{N-k}$.
  Thus the assertion $M$ exists and satisfying $\hbar_{N-k}< M\le \hbar_{N-k+1}$ follows by the facts that $\max_{i=1,\ldots,N}{\hbar_i\wedge M}/{(\sum_{j=1}^N\hbar_j\wedge M)}$ is an continuous function on $M$ conditioning on $\hbar_1,\ldots,\hbar_N$.

  Also note that for any $\hbar_N\ge M'>M$, $M'\wedge \hbar_N\ge M\wedge \hbar_N$ and $(M'/M)\sum_{i=1}^N\hbar_i\wedge M\ge \sum_{i=1}^N\hbar_i\wedge M'$.
  Thus $\hbar_N\wedge M/(\sum_{i=1}^N\hbar_i\wedge M)$ is nondecreasing on $M\in(\hbar_1,\hbar_N)$.
  Therefore
  \begin{equation*}
  \max_{i=1,\ldots,N}\frac{\hbar_i\wedge M}{\sum_{j=1}^N\hbar_j\wedge M}=\frac{1}{r},
   \end{equation*}
   indicates that $\hbar_{N-k}< M\le \hbar_{N-k+1}$.}

 The proof for Theorem~4 is similar, so we omit the details.
\end{proof}

\subsection{Proofs of {Theorems~5 and 6}}
{\blue Since $\tilde{p}_i^{\rm sos}\ge \varrho r/N$, we have $\max_{i=1,\ldots,N}(Np_i)^{-1}=O_P(r^{-1})$. %
  Theorem 1 indicates Theorem~5.
Thus it remains to show that  Theorem 6 holds.}
\begin{proof}
Note that $r_0r^{-1/2}\rightarrow0$, the contribution of the first step subsample to the {estimation equation is $o_{P\mid \mathcal{F}_N}(r^{-1/2})$. %
Thus, we can focus on the subsamples drawn in the second step only. Here we reuse the notation of $Q^*(\bm\beta)$ to represent the corresponding estimation equation. To be precise,}
\begin{equation*}
Q^*(\bm\beta)=\sum_{i=1}^{{N}} \frac{\delta_i}{\tilde p_i^{\rm sos}\wedge 1}[y_i-\psi(\bm\beta^T\bm x_i)]\bm x_i,
\end{equation*}
where %
$\tilde{p}_i^{\mathrm{sos}}$ is defined in (21) %
and $\delta_i=1$ if and only if the corresponding data point is selected in the subsample set in second step.

From Lemma \ref{lem:lem3}, {conditionally on $\tilde{\bm\beta}_0,\mathcal{F}_N$}, {it holds that}%
\[N^{-1}\tilde V_c^{-1/2}Q^*(\hat{\bm\beta}_{\rm QLE})\rightarrow N(\bm 0,I),\]
in distribution, where
\[\tilde V_c=\frac{1}{N^2}\sum_{i=1}^N\frac{1}{\tilde{p}_i^{\rm sos}\wedge 1}\{y_i-\psi(\hat{\bm\beta}_{\rm QLE}^T\bm x_i)\}^2\bm x_i\bm x_i^T-\frac{1}{N^2}\sum_{i=1}^N\{y_i-\psi(\hat{\bm\beta}_{\rm QLE}^T\bm x_i)\}^2\bm x_i\bm x_i^T,\]
since $\hat\Psi$ is a consistency estimator of $N^{-1}\sum_{i=1}^N|y_i-\psi(\tilde{\bm\beta}_0^T\bm x_i)|h(\bm x_i)$.

Let $\tilde\hbar_i^{\rm os}$ have the same expression as $\hbar^{\rm os}$ defined in (19) except that $\hat{\bm\beta}_{\rm QLE}$  is replaced by $\tilde{\bm\beta}_0$. For clarity, we consider the MVc case first. {  Recall  $\|A\|_s$ denotes the spectral norm of matrix $A$.
The distance between $\tilde V_c$ and $V_c$ can be quantified as
\begin{equation}\label{eq:vdistance}
\begin{split}
\|\tilde V_c- V_c\|_s&= \left\|\frac{1}{N^2}\sum_{i=1}^N\frac{1}{{p}_i^{\rm sos}\wedge 1}\{y_i-\psi(\hat{\bm\beta}_{\rm QLE}^T\bm x_i)\}^2\bm x_i\bm x_i^T\left(\frac{{p}_i^{\rm sos}\wedge 1}{\tilde{p}_i^{\rm sos}\wedge 1}-1\right)\right\|_s\\
&\le\frac{1}{N^2}\sum_{i=1}^N\frac{1}{{p}_i^{\rm sos}\wedge 1}\left|\frac{{p}_i^{\rm sos}\wedge 1}{\tilde{p}_i^{\rm sos}\wedge 1}-1\right|\{y_i-\psi(\hat{\bm\beta}_{\rm QLE}^T\bm x_i)\}^2\|\bm x_i\|^2\\
&\le \left(\max_{i=1,\ldots,N}\frac{1}{N{p}_i^{\rm sos}}\right)\sum_{i=1}^N\frac{1}{N}\left|\frac{{p}_i^{\rm sos}\wedge 1}{\tilde{p}_i^{\rm sos}\wedge 1}-1\right|\{y_i-\psi(\hat{\bm\beta}_{\rm QLE}^T\bm x_i)\}^2\|\bm x_i\|^2\\
&\le \left(\varrho r\right)^{-1}\sum_{i=1}^N\frac{1}{N}\left|\frac{{p}_i^{\rm sos}\wedge 1}{\tilde{p}_i^{\rm sos}\wedge 1}-1\right|\{y_i-\psi(\hat{\bm\beta}_{\rm QLE}^T\bm x_i)\}^2\|\bm x_i\|^2.
  \end{split}
\end{equation}

Simple calculation yields,
\begin{small}
\begin{eqnarray*}
\nonumber  &~& \sum_{i=1}^N\left|\frac{{p}_i^{\rm sos}\wedge 1}{\tilde{p}_i^{\rm sos}\wedge 1}-1\right|\frac{\{y_i-{\psi}(\hat{\bm\beta}_{\rm QLE}^T\bm x_i)\}^2\|\bm x_i\|^2}{N}\\
\nonumber  &\le& \sum_{i=1}^N \frac{\frac{\sum_{i=1}^N\hbar_i}{\sum_{i=1}^N\tilde\hbar_i}\left|\tilde\hbar_i-\hbar_i\right| + \left|\frac{\sum_{i=1}^N\hbar_i}{\sum_{i=1}^N\tilde\hbar_i}-1\right|\hbar_i}{\varrho  N^{-1}\sum_{i=1}^N\hbar_i}\frac{\{y_i-{\psi}(\hat{\bm\beta}_{\rm QLE}^T\bm x_i)\}^2\|\bm x_i\|^2}{N}\\
\nonumber  &=&\frac{\sum_{i=1}^N\hbar_i}{\sum_{i=1}^N\tilde\hbar_i}\sum_{i=1}^N \frac{\left|\tilde\hbar_i-\hbar_i\right|\{y_i-{\psi}(\hat{\bm\beta}_{\rm QLE}^T\bm x_i)\}^2\|\bm x_i\|^2}{N \varrho  N^{-1}\sum_{i=1}^N\hbar_i}+\left|\frac{\sum_{i=1}^N\hbar_i}{\sum_{i=1}^N\tilde\hbar_i}-1\right|\sum_{i=1}^N\frac{|y_i-{\psi}(\hat{\bm\beta}_{\rm QLE}^T\bm x_i)|^3\|\bm x_i\|^3}{N{\varrho  N^{-1}\sum_{i=1}^N\hbar_i}}\\
  &\le& \left|\frac{\sum_{i=1}^N\hbar_i}{\sum_{i=1}^N\tilde\hbar_i}-1\right|\frac{1}{\varrho  N^{-1}\sum_{i=1}^N\hbar_i}\sum_{i=1}^N\frac{|y_i-{\psi}(\hat{\bm\beta}_{\rm QLE}^T\bm x_i)|^3\|\bm x_i\|^3}{N} \\
  &+&\frac{\sum_{i=1}^N\hbar_i}{\sum_{i=1}^N\tilde\hbar_i}\frac{1}{\varrho  N^{-1}\sum_{i=1}^N\hbar_i}\sum_{i=1}^N\frac{|{\psi}(\hat{\bm\beta}_{\rm QLE}^T\bm x_i)-{\psi}(\tilde{\bm\beta}_{0}^T\bm x_i)|\{y_i-{\psi}(\hat{\bm\beta}_{\rm QLE}^T\bm x_i)\}^2\|\bm x_i\|^3}{N}, \\
\end{eqnarray*}
\end{small}
  where the first inequality comes from the facts
  \begin{align}\label{eq:ax-thm6-2}
\nonumber \left|\frac{{p}_i^{\rm sos}\wedge 1}{\tilde{p}_i^{\rm sos}\wedge 1}-1\right|&\le\left|\frac{(1-\varrho)\frac{r\hbar_i}{\sum_{i=1}^N\hbar_i}+\varrho\frac{r}{N}-(1-\varrho)\frac{r\tilde\hbar_i}{\sum_{i=1}^N \tilde\hbar_i}-\varrho\frac{r}{N}}{[(1-\varrho)\frac{r \hbar_i}{\sum\hbar_i}+\varrho\frac{r}{N}]\wedge 1}\right|\\
\nonumber&\le (1-\varrho)\left|\frac{\frac{\tilde\hbar_i}{\sum_{i=1}^N\tilde\hbar_i}-\frac{\hbar_i}{\sum_{i=1}^N\hbar_i}}{\varrho {N}^{-1}}\right|\\
\nonumber&\le \left|\frac{\left|\frac{\tilde\hbar_i}{\sum_{i=1}^N\tilde\hbar_i}-\frac{\hbar_i}{\sum_{i=1}^N\tilde\hbar_i}\right|+\left|\frac{\hbar_i}{\sum_{i=1}^N\tilde\hbar_i}-\frac{\hbar_i}{\sum_{i=1}^N\hbar_i}\right|}{\varrho {N}^{-1}}\right|\\
&=\frac{\frac{\sum_{i=1}^N\hbar_i}{\sum_{i=1}^N\tilde\hbar_i}\left|\tilde\hbar_i-\hbar_i\right| + \left|\frac{\sum_{i=1}^N\hbar_i}{\sum_{i=1}^N\tilde\hbar_i}-1\right|(\hbar_i)}{\varrho  N^{-1}\sum_{i=1}^N\hbar_i}.
\end{align}

To well exam the distance between $\tilde V_c$ and $V_c$, we will show the following equalities hold:{\purple
\begin{small}
\begin{align}
&\left|\frac{\sum_{i=1}^N\hbar_i}{\sum_{i=1}^N\tilde\hbar_i}-1\right|\frac{1}{\varrho  N^{-1}\sum_{i=1}^N\hbar_i}\sum_{i=1}^N\frac{|y_i-{\psi}(\hat{\bm\beta}_{\rm QLE}^T\bm x_i)|^3\|\bm x_i\|^3}{N}= o_P(1), \label{eq:ax-thm6-2-1}\\
&\frac{\sum_{i=1}^N\hbar_i}{\sum_{i=1}^N\tilde\hbar_i}\frac{1}{\varrho  N^{-1}\sum_{i=1}^N\hbar_i}\sum_{i=1}^N\frac{|{\psi}(\hat{\bm\beta}_{\rm QLE}^T\bm x_i)-{\psi}(\tilde{\bm\beta}_{0}^T\bm x_i)|\{y_i-{\psi}(\hat{\bm\beta}_{\rm QLE}^T\bm x_i)\}^2\|\bm x_i\|^3}{N}=o_P(1). \label{eq:ax-thm6-2-2}
\end{align}
\end{small}}

 Now we begin with showing (\ref{eq:ax-thm6-2-1}). For the sake of clarity, we first consider the case that $\tilde\hbar_i^{\rm os}$ and $\hbar_i^{\rm os}$ are selected as $\tilde\hbar_i^{\rm MVc}$ and $\hbar_i^{\rm MVc}$ respectively.
{\purple According to the triangle inequality,
\begin{align}\label{eq:ax-thm6}
\nonumber&\quad\left|\frac{1}{N}\sum_{i=1}^N|y_i-\psi(\tilde{\bm\beta}_0^T\bm x_i)|\|\bm x_i\|-\frac{1}{N}\sum_{i=1}^N|y_i-\psi(\hat{\bm\beta}_{\rm QLE}^T\bm x_i)|\|\bm x_i\|\right|\\
\nonumber&\le\frac{1}{N}\sum_{i=1}^N\left||y_i-\psi(\tilde{\bm\beta}_0^T\bm x_i)|-|y_i-\psi(\hat{\bm\beta}_{\rm QLE}^T\bm x_i)|\right|\|\bm x_i\|\\
\nonumber&\le \frac{1}{N}\sum_{i=1}^N |\psi(\tilde{\bm\beta}_0^T\bm x_i)-\psi(\hat{\bm\beta}_{\rm QLE}^T\bm x_i)|\|\bm x_i\|\\
\nonumber&\le \sqrt{\frac{1}{N}\sum_{i=1}^N |\psi(\tilde{\bm\beta}_0^T\bm x_i)-\psi(\hat{\bm\beta}_{\rm QLE}^T\bm x_i)|^2}\sqrt{\frac{1}{N}\sum_{i=1}^N\|\bm x_i\|^2 }\\
&\le \sqrt{\frac{1}{N}\sum_{i=1}^N m_1^2(\bm x_i)\|\tilde{\bm\beta}_0-\hat{\bm\beta}_{\rm QLE}\|^2}\sqrt{\frac{1}{N}\sum_{i=1}^N\|\bm x_i\|^2 }
=o_P(1),
\end{align}}
where the  last equality holds due to Assumptions 2, 4 by using the holder inequality and the fact that $\|\tilde{\bm\beta}_0-\hat{\bm\beta}_{\rm QLE}\|=o_{P\mid\mathcal{F}_N}(1)=o_P(1)$ \citep[see][Theorem 3.3]{xiong2008some}.
Therefore,
\begin{equation}\label{eq:s20}
{\sum_{i=1}^N\hbar_i}={\sum_{i=1}^N\tilde\hbar_i}+o_P(1).
\end{equation}%
Similarly it also can be shown
\[\frac{1}{N}\sum_{i=1}^N|\psi({\bm\beta}_{t}^T\bm x_i)-\psi(\tilde{\bm\beta}_{0}^T\bm x_i)|\|\bm x_i\|=o_P(1),\]
{\purple by noting that $\hat{\bm\beta}_{\rm QLE}$ is a consistent estimator of ${\bm\beta}_{t}$. }
Thus,
\[
\frac{1}{N}\sum_{i=1}^N\tilde\hbar_i:=\frac{1}{N}\sum_{i=1}^N|y_i-\psi(\tilde{\bm\beta}_{0}^T\bm x_i)|\|\bm x_i\|= \frac{1}{N}\sum_{i=1}^N|y_i-\psi({\bm\beta}_{t}^T\bm x_i)|\|\bm x_i\|+o_P(1).
\]
By the law of large number, we have
\[\frac{1}{N}\sum_{i=1}^N|y_i-\psi({\bm\beta}_{t}^T\bm x_i)|\|\bm x_i\|=E|y_1-\psi({\bm\beta}_{t}^T\bm x_1)|\|\bm x_1\| +o_P(1).\]
Obviously, $\E|y_1-\psi({\bm\beta}_{t}^T\bm x_1)|\|\bm x_1\|$ is a positive constant under the model setting.
Hence,
\begin{equation}\label{eq:s21}
\left(N^{-1}\sum_{i=1}^N\tilde\hbar_i\right)^{-1}=O_P(1).
\end{equation}
{\purple Following the facts we have dervied in (\ref{eq:s20}) and (\ref{eq:s21}), it holds that
\begin{eqnarray}
\left|\frac{\sum_{i=1}^N\hbar_i}{\sum_{i=1}^N\tilde\hbar_i}-1\right|&=&o_P(1),\label{eq:s22}\\
{\left(\varrho  N^{-1}\sum_{i=1}^N\hbar_i\right)^{-1}}&=&O_P(1).\label{eq:s23}
\end{eqnarray}

Combining the results (\ref{eq:ax-thm1}), (\ref{eq:s22})  and (\ref{eq:s23}), we have (\ref{eq:ax-thm6-2-1}).}
 Thus it remains to show (\ref{eq:ax-thm6-2-2}).
 Using the Holder's inequality (\ref{eq:holder}), it follows:
  \begin{small}
  \begin{eqnarray*}\label{eq:s24}
  \nonumber~~&&\sum_{i=1}^N\frac{|{\psi}(\hat{\bm\beta}_{\rm QLE}^T\bm x_i)-{\psi}(\tilde{\bm\beta}_{0}^T\bm x_i)|\{y_i-{\psi}(\hat{\bm\beta}_{\rm QLE}^T\bm x_i)\}^2\|\bm x_i\|^3}{N}\\
 \nonumber &\le&\sum_{i=1}^N\frac{|{\psi}(\hat{\bm\beta}_{\rm QLE}^T\bm x_i)-{\psi}(\tilde{\bm\beta}_{0}^T\bm x_i)|y_i^2\|\bm x_i\|^3}{N}+2\sum_{i=1}^N\frac{|{\psi}(\hat{\bm\beta}_{\rm QLE}^T\bm x_i)-{\psi}(\tilde{\bm\beta}_{0}^T\bm x_i)||y_i||{\psi}(\hat{\bm\beta}_{\rm QLE}^T\bm x_i)|\|\bm x_i\|^3}{N}\\
 \nonumber &&+\sum_{i=1}^N\frac{|{\psi}(\hat{\bm\beta}_{\rm QLE}^T\bm x_i)-{\psi}(\tilde{\bm\beta}_{0}^T\bm x_i)||{\psi}(\hat{\bm\beta}_{\rm QLE}^T\bm x_i)|^2\|\bm x_i\|^3}{N}\\
 \nonumber &\le&\left(\sum_{i=1}^N\frac{|{\psi}(\hat{\bm\beta}_{\rm QLE}^T\bm x_i)-{\psi}(\tilde{\bm\beta}_{0}^T\bm x_i)|^3}{N}\right)^{1/3}\left(\sum_{i=1}^N\frac{y_i^6}{N}\right)^{1/3}\left(\sum_{i=1}^N\frac{\|\bm x_i\|^9}{N}\right)^{1/3}\\
 \nonumber &&+2\left(\sum_{i=1}^N\frac{|{\psi}(\hat{\bm\beta}_{\rm QLE}^T\bm x_i)-{\psi}(\tilde{\bm\beta}_{0}^T\bm x_i)|^3}{N}\right)^{1/3}\left(\sum_{i=1}^N\frac{|y_i|^3|{\psi}(\hat{\bm\beta}_{\rm QLE}^T\bm x_i)|^3}{N}\right)^{1/3}\left(\sum_{i=1}^N\frac{\|\bm x_i\|^9}{N}\right)^{1/3}\\
  &&+\left(\sum_{i=1}^N\frac{|{\psi}(\hat{\bm\beta}_{\rm QLE}^T\bm x_i)-{\psi}(\tilde{\bm\beta}_{0}^T\bm x_i)|^3}{N}\right)^{1/3}\left(\sum_{i=1}^N\frac{{\psi}(\hat{\bm\beta}_{\rm QLE}^T\bm x_i)^6}{N}\right)^{1/3}\left(\sum_{i=1}^N\frac{\|\bm x_i\|^9}{N}\right)^{1/3}.
  \end{eqnarray*}
  \end{small}
 {\purple Thus in order to obtain (\ref{eq:ax-thm6-2-2}), %
  it is sufficient to show the following equalities hold:
  \begin{eqnarray}
    &&\sum_{i=1}^N\frac{|{\psi}(\hat{\bm\beta}_{\rm QLE}^T\bm x_i)-{\psi}(\tilde{\bm\beta}_{0}^T\bm x_i)|^3}{N}=o_P(1), \label{eq:ax-thm6-1-1}\\
    &&\sum_{i=1}^N\frac{y_i^6}{N}=O_P(1),\ \sum_{i=1}^N\frac{\|\bm x_i\|^9}{N}=O_P(1),\ \sum_{i=1}^N\frac{{\psi}(\hat{\bm\beta}_{\rm QLE}^T\bm x_i)^6}{N}=O_P(1),\label{eq:ax-thm6-1-2}
  \end{eqnarray}
  since ${N}^{-1}\sum_{i=1}^N{|y_i|^3|{\psi}(\hat{\bm\beta}_{\rm QLE}^T\bm x_i)|^3}\le{({N}^{-1}\sum_{i=1}^N{|y_i|^6})}^{1/2}{({N}^{-1}\sum_{i=1}^N|{\psi}(\hat{\bm\beta}_{\rm QLE}^T\bm x_i)|^6)}^{1/2}$.}
  From Assumption 2,  (\ref{eq:ax-thm6-1-2}) holds due to Markov's inequality. Now we check (\ref{eq:ax-thm6-1-1}). Under Assumption 4, it can be shown that
  \begin{equation*}
  \begin{split}
 \frac{1}{N}\sum_{i=1}^N |\psi(\tilde{\bm\beta}_0^T\bm x_i)-\psi(\hat{\bm\beta}_{\rm QLE}^T\bm x_i)|^3
&\le \frac{1}{N}\sum_{i=1}^N m_1^3(\bm x_i)\|\tilde{\bm\beta}_0-\hat{\bm\beta}_{\rm QLE}\|^3\\
&= \|\tilde{\bm\beta}_0-\hat{\bm\beta}_{\rm QLE}\|^3 \frac{1}{N}\sum_{i=1}^N m_1^3(\bm x_i)\\
&=o_P(1),
\end{split}
\end{equation*}
  where the  last equality holds due to  the fact that $\|\tilde{\bm\beta}_0-\hat{\bm\beta}_{\rm QLE}\|=o_{P\mid\mathcal{F}_N}(1)=o_P(1)$. %
  {\purple The results (\ref{eq:ax-thm6-1-1}) and (\ref{eq:ax-thm6-1-2}) implies that
  \begin{equation}\label{eq:s26}
  \sum_{i=1}^N\frac{|{\psi}(\hat{\bm\beta}_{\rm QLE}^T\bm x_i)-{\psi}(\tilde{\bm\beta}_{0}^T\bm x_i)|\{y_i-{\psi}(\hat{\bm\beta}_{\rm QLE}^T\bm x_i)\}^2\|\bm x_i\|^3}{N}=o_P(1)
  \end{equation}
Therefore,  (\ref{eq:ax-thm6-2-2}) holds by noting the results in (\ref{eq:s23}) and (\ref{eq:s22}).}
}

Now let us consider the case that $\tilde\hbar_i^{\rm os}$ and $\hbar_i^{\rm os}$ are selected as $\tilde\hbar_i^{\rm MV}$ and $\hbar_i^{\rm MV}$ respectively.
For brevity, let $\tilde\Sigma$ and $\hat\Sigma$ denote $\Sigma_\psi(\tilde{\bm\beta}_{0})$ and $\Sigma_\psi(\hat{\bm\beta}_{\rm QLE})$ respectively.
From Assumption 3, Lemma \ref{lem:lem4}, (\ref{eq:ax-thm6}) in this case turns into
\begin{align}\label{eq:s27}
\nonumber&\quad\left|\frac{1}{N}\sum_{i=1}^N|y_i-\psi(\tilde{\bm\beta}_0^T\bm x_i)|\|\tilde\Sigma^{-1}\bm x_i\|-\frac{1}{N}\sum_{i=1}^N|y_i-\psi(\hat{\bm\beta}_{\rm QLE}^T\bm x_i)|\|\hat\Sigma^{-1}\bm x_i\|\right|\\
\nonumber&\le \frac{1}{N}\sum_{i=1}^N |\psi(\tilde{\bm\beta}_0^T\bm x_i)-\psi(\hat{\bm\beta}_{\rm QLE}^T\bm x_i)|\|\tilde\Sigma^{-1}\bm x_i\|+ \frac{1}{N}\sum_{i=1}^N |y_i-\psi(\hat{\bm\beta}_{\rm QLE}^T\bm x_i)|\|(\tilde{\Sigma}^{-1}-\hat\Sigma^{-1})\bm x_i\|\\
\nonumber&\le \frac{\lambda_{\max}(\tilde\Sigma^{-1})}{N}\sum_{i=1}^N |\psi(\tilde{\bm\beta}_0^T\bm x_i)-\psi(\hat{\bm\beta}_{\rm QLE}^T\bm x_i)|\|\bm x_i\|
+ o_P(1)\\
&=o_P(1),
\end{align}
where the first inequality follows from the triangle inequality, and the last equality holds due to the same reason as (\ref{eq:ax-thm6}).
Using the similar arguments, it holds that
\[
\frac{1}{N}\sum_{i=1}^N\tilde\hbar_i:=\frac{1}{N}\sum_{i=1}^N|y_i-\psi(\tilde{\bm\beta}_{0}^T\bm x_i)|\|\tilde\Sigma^{-1}\bm x_i\|\ge \frac{\lambda_{\min}(\tilde\Sigma^{-1})}{N}\sum_{i=1}^N|y_i-\psi({\bm\beta}_{t}^T\bm x_i)|\|\bm x_i\|+o_P(1).
\]
Hence, $(N^{-1}\sum_{i=1}^N\tilde\hbar_i)^{-1}=O_P(1)$. {\purple  Combing this result with (\ref{eq:s27}), it is obviously that (\ref{eq:s22}) and (\ref{eq:s23}) also hold for this case.}

{\purple Combing (\ref{eq:ax-thm1}), (\ref{eq:s22}), (\ref{eq:s23}),  (\ref{eq:s26}), and Assumption 3, we have that}
\begin{eqnarray*}
  0&\le& \sum_{i=1}^N\left|\frac{{p}_i^{\rm sos}\wedge 1}{\tilde{p}_i^{\rm sos}\wedge 1}-1\right|\frac{\{y_i-{\psi}(\hat{\bm\beta}_{\rm QLE}^T\bm x_i)\}^2\|\bm x_i\|^2}{N}\\
  &\le& \sum_{i=1}^N \frac{\frac{\sum_{i=1}^N\hbar_i}{\sum_{i=1}^N\tilde\hbar_i}\left|\tilde\hbar_i-\hbar_i\right| + \left|\frac{\sum_{i=1}^N\hbar_i}{\sum_{i=1}^N\tilde\hbar_i}-1\right|\hbar_i}{\varrho  N^{-1}\sum_{i=1}^N\hbar_i}\frac{\{y_i-{\psi}(\hat{\bm\beta}_{\rm QLE}^T\bm x_i)\}^2\|\bm x_i\|^2}{N}\\
  &\le&\frac{\sum_{i=1}^N\hbar_i}{\sum_{i=1}^N\tilde\hbar_i}\sum_{i=1}^N \frac{\left|\tilde\hbar_i-\hbar_i\right|\{y_i-{\psi}(\hat{\bm\beta}_{\rm QLE}^T\bm x_i)\}^2\|\bm x_i\|^2}{N \varrho  N^{-1}\sum_{i=1}^N\hbar_i}\\
  &&~~~+\frac{\lambda_{\max}(\hat\Sigma^{-1})}{\varrho  N^{-1}\sum_{i=1}^N\hbar_i}\left|\frac{\sum_{i=1}^N\hbar_i}{\sum_{i=1}^N\tilde\hbar_i}-1\right|\sum_{i=1}^N\frac{|y_i-{\psi}(\hat{\bm\beta}_{\rm QLE}^T\bm x_i)|^3\|\bm x_i\|^3}{N}\\
  &\le&\frac{\sum_{i=1}^N\hbar_i}{\sum_{i=1}^N\tilde\hbar_i}\frac{\lambda_{\max}(\hat\Sigma^{-1})+\lambda_{\max}(\tilde\Sigma^{-1})}{\varrho  N^{-1}\sum_{i=1}^N\hbar_i}\sum_{i=1}^N\frac{|{\psi}(\hat{\bm\beta}_{\rm QLE}^T\bm x_i)-{\psi}(\tilde{\bm\beta}_{0}^T\bm x_i)|\{y_i-{\psi}(\hat{\bm\beta}_{\rm QLE}^T\bm x_i)\}^2\|\bm x_i\|^3}{N}\\
  &&~~~+\frac{\lambda_{\max}(\hat\Sigma^{-1})}{\varrho  N^{-1}\sum_{i=1}^N\hbar_i}\left|\frac{\sum_{i=1}^N\hbar_i}{\sum_{i=1}^N\tilde\hbar_i}-1\right|\sum_{i=1}^N\frac{|y_i-{\psi}(\hat{\bm\beta}_{\rm QLE}^T\bm x_i)|^3\|\bm x_i\|^3}{N}\\
  &=& o_{P}(1),
\end{eqnarray*}
where the last equality holds by noting that $\lambda_{\max}(\hat\Sigma^{-1})$ and $\lambda_{\max}(\tilde\Sigma^{-1})$ are $O_P(1)$. %

{\blue
Combing (\ref{eq:vdistance}), (\ref{eq:ax-thm6-2-1}), and (\ref{eq:ax-thm6-2-2}), it follows that  $\|\tilde V_c-V_c\|_s=o_P(r^{-1})$.
Therefore, the desired results follow by Lemma \ref{lem:lem3} and Slutsky's theorem by noting
 \begin{equation*}
 \begin{split}
   &~~~~V^{-1/2}\Sigma_\psi(\hat{\bm\beta}_{\rm QLE})^{-1}( \tilde V_c)^{1/2}\{V^{-1/2}\Sigma_\psi(\hat{\bm\beta}_{\rm QLE})^{-1}(\tilde V_c)^{1/2}\}^T\\
   &=V^{-1/2}\Sigma_\psi(\hat{\bm\beta}_{\rm QLE})^{-1}(\tilde V_c)\Sigma_\psi(\hat{\bm\beta}_{\rm QLE})^{-1}V^{-1/2}\\
   &=V^{-1/2}\Sigma_\psi(\hat{\bm\beta}_{\rm QLE})^{-1}( V_c)\Sigma_\psi(\hat{\bm\beta}_{\rm QLE})^{-1}V^{-1/2}+o_{P|\mathcal{F}_N}(r^{-1/2}) \\
   &=I+o_{P|\mathcal{F}_N}(r^{-1/2}),
   \end{split}
 \end{equation*}
 where the last equality is by the definition of $V_c$.}
\end{proof}

\subsection{Proofs of Theorems~7 and 8}
To prove  Theorem 7, we first establish some lemmas on the  estimator ${\tilde{\bm\beta}_j}$ {which is} calculated on a single machine $\mathcal{F}_{Nj}$.
For simplicity, {let $p_{ji}$ denote $\tilde{p}_{ji}^{\rm sos}$ for $j=1,\ldots,K$, and $i=1,\ldots,n$}.

\begin{lemma}\label{lem:local}
If Assumptions 1 -- 4  hold, then conditional on subset $\mathcal{F}_{Nj}$, as $n\rightarrow\infty$, with probability approaching one, the subsample QLE $\tilde{\bm\beta}_j$ based on subsamples {\purple in} $\mathcal{F}_{Nj}$ satisfies
\begin{equation}\label{eq:result11-ub}
  \pr(r^\alpha\|{\tilde{\bm\beta}_j}-\bm\beta_t\|>\Delta)=O( r^{2\alpha-1})
\end{equation}
for any $\Delta>0$ and $\alpha\in(1/4,1/2)$, where $\bm\beta_t$ is the true value of $\bm\beta$. %

\end{lemma}

\begin{proof}
Without of loss generality, we assume $j=1$.
Recall that the estimator $\tilde{\bm\beta}_1$ is the solution of the estimation equation $n^{-1}\sum_{i=1}^n{\delta_{1i}}{p_{1i}^{-1}}\{y_{1i}-\psi(\bm\beta^T\bm x_{1i})\}\bm x_{1i}=\bm 0$.
From \cite{Wedderburn1974Quasi}, it is clear to see that, under Assumption 3, $\tilde{\bm\beta}_1$ achieves the minimum of the following quasi-likelihood function:
\begin{equation}
  U_S(\bm\beta):=n^{-1}\sum_{i=1}^n\frac{\delta_{1i}}{p_{1i}}\int_{y_{1i}}^{\psi(\bm\beta^T\bm x_{1i})}\frac{t-y_{1i}}{\dot\psi(\psi^{-1}(t))}dt,
\end{equation}
with $\partial U_S/\partial{\bm\beta}=-n^{-1}\sum_{i=1}^n\frac{\delta_{1i}}{p_{1i}}[y_{1i}-\psi(\bm\beta^T\bm x_{1i})]\bm x_{1i}$, where $\psi^{-1}(t)$ denotes the inverse function of $\psi(t)$.
Note that $\dot\psi(t)>0$ in our model setting which indicates the existence of $\psi^{-1}(t)$.
By Taylor expansion,
\begin{equation}
\begin{split}
U_S(\tilde{\bm\beta}_1)=& U_S({\bm\beta}_t) - n^{-1}\sum_{i=1}^n\frac{\delta_{1i}}{p_{1i}}\{y_{1i}-\psi({\bm\beta}_t^T\bm x_{1i})\}\bm x_{1i}^T(\tilde{\bm\beta}_1-\bm\beta_t)\\
&+\frac{1}{2}(\tilde{\bm\beta}_1-\bm\beta_t)^T\left\{n^{-1}\sum_{i=1}^n\frac{\delta_{1i}}{p_{1i}}{\dot{\psi}}(\acute{\bm{\beta}}^T\bm x_{1i})\bm x_{1i}\bm x_{1i}^T\right\}(\tilde{\bm\beta}_1-\bm\beta_t),\\
\ge& U_S({\bm\beta}_t) - \left\|n^{-1}\sum_{i=1}^n\frac{\delta_{1i}}{p_{1i}}[y_{1i}-\psi({\bm\beta}_t^T\bm x_{1i})]\bm x_{1i}\right\|\left\|\tilde{\bm\beta}_1-\bm\beta_t\right\|\\
&+\frac{1}{2}\lambda_{\min}(\Sigma_{\psi,S}(\acute{\bm\beta}))\left\|\tilde{\bm\beta}_1-\bm\beta_t\right\|^2,
\end{split}
\end{equation}
where $\acute{\bm{\beta}}$ lies between ${\bm\beta}_{t}$ and $\tilde{\bm\beta}_1$ and $\Sigma_{\psi,n}(\acute{\bm\beta}):=n^{-1}\sum_{i=1}^n{\dot{\psi}}(\acute{\bm{\beta}}^T\bm x_{1i})\bm x_{1i}\bm x_{1i}^T$.

Simple calculation yields
\begin{equation}
  \begin{split}
&~~\left\|n^{-1}\sum_{i=1}^n\frac{\delta_{1i}}{p_{1i}}\{y_{1i}-\psi({\bm\beta}_t^T\bm x_{1i})\}\bm x_{1i}\right\|\left\|\tilde{\bm\beta}_1-\bm\beta_t\right\|\\
&\ge U_S({\bm\beta}_t)-U_S(\tilde{\bm\beta}_1)+2^{-1}\lambda_{\min}(\Sigma_{\psi,S}(\acute{\bm\beta}))\left\|\tilde{\bm\beta}_1-\bm\beta_t\right\|^2\\
&\ge 2^{-1}\lambda_{\min}(\Sigma_{\psi,S}(\acute{\bm\beta}))\left\|\tilde{\bm\beta}_1-\bm\beta_t\right\|^2,%
\end{split}
\end{equation}
by noting $U_S(\tilde{\bm\beta}_1)\le U_S({\bm\beta}_t).$
Thus based on the following event
\begin{equation*}\label{eq:as-lambda2-ub}
\begin{split}
 \Xi_{1,\lambda}&=\{ 0.5 C_\psi\leq \lambda_{\min}(\Sigma_{\psi,S}(\acute{\bm\beta}))\},
\end{split}
\end{equation*}
it follows that
\[\begin{split}
\|\tilde{\bm\beta}_{1}-\bm\beta_t\|\le \frac{4}{C_\psi}\left\|n^{-1}\sum_{i=1}^n\frac{\delta_{1i}}{p_{1i}}\{y_{1i}-\psi({\bm\beta}_t^T\bm x_{1i})\}\bm x_{1i}\right\|,
\end{split}\]
which implies
\[\|\tilde{\bm\beta}_{1}-\bm\beta_t\|^2\le \frac{16}{C_\psi^2}\left\|n^{-1}\sum_{i=1}^n\frac{\delta_{1i}}{p_{1i}}\{y_{1i}-\psi({\bm\beta}_t^T\bm x_{1i})\}\bm x_{1i}\right\|^2. \]
Thus,
\begin{equation}\label{eq:ax-lemma3-4}
\begin{split}
&{\rm pr}(\|\tilde{\bm\beta}_{1}-\bm\beta_t\|\ge\Delta,\Xi_{1,\lambda}\ \text{happens})=\E \mathbbm{1}_{\{\|\tilde{\bm\beta}_{1}-\bm\beta_t\|\ge\Delta,\Xi_{1,\lambda}\ \text{happens}.\}}\\
\le&\E\left\{ \frac{16}{C_\psi^2\Delta^2}\left\|n^{-1}\sum_{i=1}^n\frac{\delta_{1i}}{p_{1i}}\{y_{1i}-\psi({\bm\beta}_t^T\bm x_{1i})\}\bm x_{1i}\right\|^2 \times \mathbbm{1}_{\{\|\tilde{\bm\beta}_{1}-\bm\beta_t\|\ge\Delta,\Xi_{1,\lambda}\ \text{happens}\}}\right\}\\
\le& \E\left\{ \frac{16}{C_\psi^2\Delta^2}\left\|n^{-1}\sum_{i=1}^n\frac{\delta_{1i}}{p_{1i}}\{y_{1i}-\psi({\bm\beta}_t^T\bm x_{1i})\}\bm x_{1i}\right\|^2 \right\}.
\end{split}
\end{equation}
To deal with (\ref{eq:ax-lemma3-4}), we note that
\begin{eqnarray}\label{eq:ax-thm6-4}
\nonumber&\quad&\E\left\{\left\|n^{-1}\sum_{i=1}^n\frac{\delta_{1i}}{p_{1i}}\{y_{1i}-\psi({\bm\beta}_t^T\bm x_{1i})\}\bm x_{1i}\right\|^2\right\}\\
\nonumber&=& \E\left\{\E\left(\sum_{k=1}^d\left[\frac{1}{n}\sum_{i=1}^n\frac{\delta_{1i}}{p_{1i}}\{y_{1i}-\psi({\bm\beta}_t^T\bm x_{1i})\}x_{1ik}\right]^2\bigg|\mathcal{F}_{N1},\tilde{\bm\beta}_0\right)\right\}\\
&=& \E\bigg[\sum_{k=1}^d\frac{1}{n^2}\sum_{i=1}^n\frac{1}{p_{1i}}\{y_{1i}-\psi({\bm\beta}_t^T\bm x_{1i})\}^2x_{1ik}^2\bigg]\\
\nonumber&+&\sum_{k=1}^d\E\bigg[\sum_{i=1}^n\sum_{i\neq j}\{y_{1i}-\psi({\bm\beta}_t^T\bm x_{1i})\}\{y_{1i}-\psi({\bm\beta}_t^T\bm x_{1j})\}x_{1ik}x_{1jk}\bigg]\\
\nonumber&\le& \frac{1}{\varrho r }\E\left[\sum_{k=1}^d\frac{1}{n}\sum_{i=1}^n\{y_{1i}-\psi({\bm\beta}_t^T\bm x_{1i})\}^2x_{1ik}^2\right]
\le  \frac{d}{\varrho r }\E[\{y_1-\psi({\bm\beta}_t^T\bm x_{11})\}^2\|\bm x_{11}\|^2]\\
\nonumber&\le&   \frac{d}{\varrho r }\sqrt{\E[\{y_1-\psi({\bm\beta}_t^T\bm x_{11})\}^4]\E\{\|\bm x_{11}\|^4\}}
\le  \frac{d\mathfrak{C}_1}{\varrho r },
\end{eqnarray}
{\purple where  $\mathfrak{C}_1$ denotes for some positive constant, the second last inequality is from Holder inequality, and {the} last inequality comes from Assumption 2.}

{\purple Combining} (\ref{eq:ax-lemma3-4}) and (\ref{eq:ax-thm6-4}), it yields
\begin{equation}\label{eq:ax-thm6-3}
\begin{split}
  \pr(r^{\alpha}\|\tilde{\bm\beta}_1-\bm\beta_t\|>\Delta, \Xi_{1,\lambda} \text{happens})&\le  \frac{16r^{2\alpha-1}d\mathfrak{C}_1}{C_{\psi}^{2}\Delta^2\varrho}.
  \end{split}
\end{equation}

Now we evaluate the probability that   $\Xi_{1,\lambda}$  happens.

From (\ref{eq:asy2}), it holds that $\|\tilde{\bm\beta}_1-\hat{\bm\beta}_{\rm QLE}\|=O_P(r^{-1/2})$. And also note $\|\hat{\bm\beta}_{\rm QLE}-\bm\beta_t\|=O_P(n^{-1/2})$ under Assumptions 1 -- 3 \citep[see][Chapter 36 Theorem 3.1]{NEWEY1994Large}. Thus it is clear to see that $\|\tilde{\bm\beta}_1-{\bm\beta}_t\|=O_P(r^{-1/2})$. %
As $r\to\infty$, with probability approaching one, we have $\E\{\E_{\acute{\bm\beta}}(\Sigma_{\psi,n}(\acute{\bm\beta})|\acute{\bm\beta},\mathcal{F}_{N1})\}\ge \E\inf_{\bm\beta\in\Lambda}(\Sigma_{\psi,n}(\acute{\bm\beta})|\mathcal{F}_{N1})$ where $\E_{\acute{\bm\beta}}$ means the expectation is taken with respect to $\acute{\bm\beta}$ conditional on $\mathcal{F}_{N1}$. %
Thus from Assumption 3, it holds that $\E(\Sigma_{\psi,n}(\acute{\bm\beta}))\ge \E(\E\inf_{\bm\beta\in\Lambda}(\Sigma_{\psi,n}(\acute{\bm\beta})|\mathcal{F}_{N1}))\ge C_\psi I_d$ for some constant $C_\psi$  by taking expectation  on both sides of the inequality. {Therefore, let $\Xi_{2,\lambda}:=\{\lambda_{\max}(E\Sigma_{\psi,n}(\acute{\bm\beta})-\Sigma_{\psi,S}(\acute{\bm\beta}))\le 0.5C_\psi\}$, it holds that $\Xi_{2,\lambda}\subseteq\Xi_{1,\lambda}$ according to the facts that $\lambda_{\min}(E\Sigma_{\psi,n}(\acute{\bm\beta})+\Sigma_{\psi,S}(\acute{\bm\beta})-E\Sigma_{\psi,n}(\acute{\bm\beta}))\ge \lambda_{\min}(E\Sigma_{\psi,n}(\acute{\bm\beta}))+\lambda_{\min}(\Sigma_{\psi,S}(\acute{\bm\beta})-E\Sigma_{\psi,n}(\acute{\bm\beta}))
=\lambda_{\min}(E\Sigma_{\psi,n}(\acute{\bm\beta}))-\lambda_{\max}(E\Sigma_{\psi,n}(\acute{\bm\beta})-\Sigma_{\psi,S}(\acute{\bm\beta}))$,
where the last equality due to the fact $\lambda_{\min}(A)=-\lambda_{\max}(-A)$.
Let $\Xi_{3,\lambda}:=\{\|\Sigma_{\psi,n}(\acute{\bm\beta})-\Sigma_{\psi,S}(\acute{\bm\beta})\|_s\le 4^{-1}C_\psi\}$,  %
 and $\Xi_{4,\lambda}:=\{\|\Sigma_{\psi,n}(\acute{\bm\beta})-E\Sigma_{\psi,n}(\acute{\bm\beta})\|_s\le 4^{-1}C_\psi\}$. %
It follows easily that $\Xi_{3,\lambda}\cap\Xi_{4,\lambda}\subseteq \Xi_{2,\lambda}$ which implies $\pr(\Xi_{1,\lambda}^c)\le \pr(\Xi_{2,\lambda}^c)\le \pr(\Xi_{3,\lambda}^c)+\pr(\Xi_{4,\lambda}^c).$} {\purple Thus, in order to get the desired result, it is sufficient to evaluate the probabilities that   $\Xi_{3,\lambda}$ and  $\Xi_{4,\lambda}$  happens respectively.}

From  Assumption 2, we have
\begin{align*}
  {E}\left\{{\Sigma}_{\psi, S}^{j_1j_2}(\acute{\bm\beta})-{\Sigma}_{\psi,n}^{j_1j_2}(\acute{\bm\beta})\right\}^2
  =&{ E}\left[{ E}\left\{{\Sigma}_{\psi, S}^{j_1j_2}(\acute{\bm\beta})-{\Sigma}_{\psi,n}^{j_1j_2}(\acute{\bm\beta})\mid\mathcal{F}_{N1},\acute{\bm\beta}\right\}^2\right]\\
  =&{ E}\left({ E}\left[\sum^n_{i=1}\frac{p_i(1-p_i)}{p_i^2}
     \left\{\frac{\dot{\psi}(\acute{\bm\beta}^T\bm x_{1i})x_{1ij_1} x_{1ij_2}}{n}\right\}^2\bigg|\acute{\bm\beta}\right]\right)\\
  \le&{ E}\left({ E}\left[\sum^n_{i=1}\frac{1}{p_i}
     \left\{\frac{\dot{\psi}(\acute{\bm\beta}^T\bm x_{1i})\|\bm x_{1i}\|^2}{n}\right\}^2\bigg|\acute{\bm\beta}\right]\right)\\
  \le&{ E}\left[{ E}\left(\max_i\frac{1}{np_i}\right)\left\{\sum^n_{i=1}
     \frac{\dot{\psi}^2(\acute{\bm\beta}^T\bm x_{1i})\|\bm x_{1i}\|^4}{n}\bigg|\acute{\bm\beta}\right\}\right]\\
  \le&{ E}\left[\left(\frac{1}{\varrho r}\right){ E}\left\{\sqrt{\sum^n_{i=1}
     \frac{\dot{\psi}^4(\acute{\bm\beta}^T\bm x_{1i})}{n}}\sqrt{\sum^n_{i=1}\frac{\|\bm x_{1i}\|^8}{n}}\bigg|\acute{\bm\beta}\right\}\right]\\
  \le&\left(\frac{1}{\varrho r}\right){ E}\left[\sqrt{{ E}\left\{\sum^n_{i=1}
     \frac{\dot{\psi}^4(\acute{\bm\beta}^T\bm x_{1i})}{n}\bigg|\acute{\bm\beta}\right\}}\sqrt{{ E}\left(\sum^n_{i=1}\frac{\|\bm x_{1i}\|^8}{n}\bigg|\acute{\bm\beta}\right)}\right]\\
  \le&\left(\frac{1}{\varrho r}\right)\left(\sqrt{{ E}\left[{ E}\left\{\sum^n_{i=1}
     \frac{\dot{\psi}^4(\acute{\bm\beta}^T\bm x_{1i})}{n}\bigg|\acute{\bm\beta}\right\}\right]}\sqrt{{ E}\sum^n_{i=1}\frac{\|\bm x_{1i}\|^8}{n}}\right)\\
  =&\left(\frac{1}{\varrho r}\right)\left(\sqrt{{ E}\left[{ E}\left\{{\dot{\psi}^4(\acute{\bm\beta}^T\bm x_{11})}|\acute{\bm\beta}\right\}\right]}\sqrt{{ E}{\|\bm x_{11}\|^8}}\right)\\
 \le&\left(\frac{1}{\varrho r}\right)\left\{\sqrt{{ E}\left(\sup_{\bm\beta\in\Lambda}{\dot{\psi}^4({\bm\beta}^T\bm x_{11})}\right)}\sqrt{{ E}{\|\bm x_{11}\|^8}}\right\}
      =O(r^{-1}).
\end{align*}
Thus
\begin{equation}\label{eq:ax-lem3-4}
\begin{split}
  \pr(\Xi_{3,\lambda}^c)&\le \pr(\|\Sigma_{\psi,n}(\acute{\bm\beta})-\Sigma_{\psi,S}(\acute{\bm\beta})\|_F\ge 4^{-1}C_\psi)\\
  &\le \frac{16{ E}\|\Sigma_{\psi,n}(\acute{\bm\beta})-\Sigma_{\psi,S}(\acute{\bm\beta})\|_F^2}{C_\psi^2}%
  =O(r^{-1}),
  \end{split}
\end{equation}
where $\|A\|_F^2:={\text{tr}(A^TA)}$ denotes the Frobenius norm of matrix $A$.

Similarly, for $1\le j_1,j_2\le p$, let ${\Sigma}_{\psi, n}(\acute{\bm\beta})^{j_1j_2}$ and $E{\Sigma}_{\psi,n}(\acute{\bm\beta})^{j_1j_2}$ denote for the $(j_1,j_2)$-th component of ${\Sigma}_{\psi,n}(\acute{\bm\beta})$ and $E{\Sigma}_{\psi,n}(\acute{\bm\beta})$, respectively.
From Assumption 2,
\begin{align*}
  &{ E}\left\{{\Sigma}_{\psi, n}^{j_1j_2}(\acute{\bm\beta})-E{\Sigma}_{\psi,n}^{j_1j_2}(\acute{\bm\beta})\right\}^2
  ={ E}\left[{ E}\left\{{\Sigma}_{\psi, n}^{j_1j_2}(\acute{\bm\beta})-E{\Sigma}_{\psi,n}^{j_1j_2}(\acute{\bm\beta})\bigg|\acute{\bm\beta}\right\}^2\right]\\
  =&{ E}\left[{ E}\left\{\sum^n_{i=1}
     \left(\frac{\dot{\psi}(\acute{\bm\beta}^T\bm x_{1i})x_{1ij_1} x_{1ij_2}}{n}\right)^2\bigg|\acute{\bm\beta}\right\}\right]-{ E}[E\{{\Sigma}_{\psi,n}^{j_1j_2}(\acute{\bm\beta})\mid\acute{\bm\beta}\}^2]\\
  \le&{ E}\left[{ E}\left\{\sum^n_{i=1}
     \left(\frac{\dot{\psi}(\acute{\bm\beta}^T\bm x_{1i})\|\bm x_{1i}\|^2}{n}\right)^2\bigg|\acute{\bm\beta}\right\}\right]\\
  \le&\frac{1}{n}\sqrt{{ E}\sup_{\bm\beta\in\Lambda}\{\dot{\psi}^4({\bm\beta}^T\bm x_{11})\}{ E}\{\|\bm x_{11}\|^8\}}=O(n^{-1}).
\end{align*}
Hence, it can be shown
\begin{equation}\label{eq:ax-lem3-6}
  \pr(\Xi_{4,\lambda}^c)\le \pr(\|\Sigma_{\psi,n}(\acute{\bm\beta})-E\Sigma_{\psi,n}(\acute{\bm\beta})\|_F\ge 4^{-1}C_\Psi)=O(n^{-1}).
\end{equation}

Combining (\ref{eq:ax-lem3-4}) %
and (\ref{eq:ax-lem3-6}), it follows
\begin{equation}
  \pr(\Xi_{1,\lambda}^c)\le \pr(\Xi_{2,\lambda}^c)\le \pr(\Xi_{3,\lambda}^c)+\pr(\Xi_{4,\lambda}^c)%
  =O(r^{-1/2}).
\end{equation}

Therefore, as $n\rightarrow\infty$, the desired result holds, with probability approaching one, from the fact
\[\begin{split}
 &\pr(\text{equation (\ref{eq:ax-thm6-3}) holds})=\pr(\text{equation (\ref{eq:ax-thm6-3}) holds},\Xi_{1,\lambda}\ \text{happens})\\
 &\quad+ \pr(\text{equation (\ref{eq:ax-thm6-3}) holds},\Xi_{1,\lambda}\ \text{not happens}) \\
 &\le \pr(\text{equation (\ref{eq:ax-thm6-3}) holds},\Xi_{1,\lambda}\ \text{happens})+\pr(\Xi_{1,\lambda}\ \text{not happens})\\
 &= O(r^{2\alpha-1})+O(r^{-1/2})\\
 &= O(r^{2\alpha-1}).%
 \end{split}\]

\end{proof}

{Recall that $\breve{\bm\beta}_{Kr}$ denote the estimator obtained from Algorithm~3 %
and $\breve{\bm\beta}$ denote the QLE obtained from the pooling subsamples. {\blue The difference between $\tilde{\bm\beta}_{Kr}$ and $\breve{\bm\beta}$ is summarized in the following lemma.} }
\begin{lemma}\label{lem:comb}
If Assumptions 1 -- 4  hold,  the estimate $\tilde{\bm\beta}_0$ based on the first step sample exists, $r_0{(Kr)}^{-1}\rightarrow0$ and the partition number $K$   satisfies $K=O(r^\eta)$ for some $\eta$ in $[0,\min(1-2\alpha,4\alpha-1))$ where $\alpha\in(1/4,1/2)$, then
  as ${r}\rightarrow\infty$ and $n\rightarrow\infty$, %
  with probability approaching one, it holds that
  \begin{equation*}
  \sqrt{Kr}\|\tilde{\bm\beta}_{Kr}-\breve{\bm\beta}\| =o_P(1).
\end{equation*}

\end{lemma}

\begin{proof}
Since $r_0(Kr)^{-1/2}\rightarrow0$, the contribution of the first step subsample to the {estimation equation}  is  $o_{P\mid \mathcal{F}_N}((Kr)^{-1/2})$.
Thus we can focus on the second step only.
{Recall that
\begin{equation}
  \tilde{\bm\beta}_{Kr}:=\left(\sum_{j=1}^K\dot{Q}_j^*(\tilde{\bm\beta}_j)\right)^{-1}\sum_{j=1}^K\dot{Q}_j^*(\tilde{\bm\beta}_j){\tilde{\bm\beta}}_j,
\end{equation}
and $\breve{\bm\beta}$ is solution of
\begin{equation}
  Q^*(\bm\beta):=\sum_{j=1}^KQ_j^*(\bm\beta)=0,
\end{equation}
where $Q_j^*(\bm\beta)$'s are defined in Algorithm~3.}

Under Assumption 3, it is clear that (24) is well defined as $n\rightarrow\infty$. %
By Taylor expansion,
\begin{equation}\label{eq:ax31}
  Q_j^*(\breve{\bm\beta})=Q_j^*(\tilde{\bm\beta}_j)+\dot{Q}_j^*(\tilde{\bm\beta}_j)(\breve{\bm\beta}-\tilde{\bm\beta}_j)+R_{j}^*,
\end{equation}
where $R_j^*$ is the Remainder with the form
\begin{equation*}
  R_j^*= \left[n^{-1}\sum_{i=1}^n\frac{\delta_{ji}}{p_{ji}}\tilde{\dot{\psi}}(\acute{\bm{\beta}}^T\bm x_{ji})\bm x_{ji}\bm x_{ji}^T-n^{-1}\sum_{i=1}^n\frac{\delta_{ji}}{p_{ji}}{\dot{\psi}}(\tilde{\bm\beta}_j^T\bm x_{ji})\bm x_{ji}\bm x_{ji}^T\right](\breve{\bm\beta}-\tilde{\bm\beta}_j),
\end{equation*}
where $\tilde{\dot{\psi}}(\acute{\bm{\beta}}^T\bm x_{ji})=\text{diag}(\dot{\psi}(\acute{\bm{\beta}}_{(1)}^T\bm x_{ji}),\ldots,\dot{\psi}(\acute{\bm{\beta}}_{(d)}^T\bm x_{ji}))$ and $\acute{\bm \beta}_{(1)},\ldots,\acute{\bm \beta}_{(d)}$ lies between $\breve{\bm\beta}$ and $\tilde{\bm\beta}_j$. Here $\text{diag}(\cdot)$ represents a diagonal matrix with diagonal elements $\dot{\psi}(\acute{\bm{\beta}}_{(1)}^T\bm x_{ji}),\ldots,\dot{\psi}(\acute{\bm{\beta}}_{(d)}^T\bm x_{ji})$.

From Assumption 4, it follows that
\begin{equation*}
\begin{split}
&\left\|\frac{\partial\psi(\acute{\bm\beta}_{(k)}^T\bm x_{ji})}{\partial\beta_k}\bm x_{ji}^T-\frac{\partial\psi(\tilde{\bm\beta}_j^T\bm x_{ji})}{\partial\beta_k}\bm x_{ji}^T\right\|_s
\le\left\|\dot{\psi}(\bm\beta_{(k)}^T\bm x_{ji})\bm x_{ji}\bm x_{ji}^T-\dot{\psi}(\tilde{\bm\beta}_j^T\bm x_{ji})\bm x_{ji}\bm x_{ji}^T\right\|_s\\
\le&m_2(\bm x_{ji})\|\acute{\bm\beta}_{(k)}-\tilde{\bm\beta}_j\|\le m_2(\bm x_{ji})\|\breve{\bm\beta}-\tilde{\bm\beta}_j\|,
\end{split}
\end{equation*}
for $k=1,\ldots,d, $ {\blue where the last inequality is due to the facts that $\acute{\bm \beta}_{(1)},\ldots,\acute{\bm \beta}_{(d)}$ lies between $\breve{\bm\beta}$ and $\tilde{\bm\beta}_j$.}
It follows:
\begin{align*}
&\left\|n^{-1}\sum_{i=1}^n\frac{\delta_{ji}}{p_{ji}}\tilde{\dot{\psi}}(\acute{\bm{\beta}}^T\bm x_{ji})\bm x_{ji}\bm x_{ji}^T-n^{-1}\sum_{i=1}^n\frac{\delta_{ji}}{p_{ji}}{\dot{\psi}}(\tilde{\bm\beta}_j^T\bm x_{ji})\bm x_{ji}\bm x_{ji}^T\right\|_s\\
\le&n^{-1}\sum_{i=1}^n\frac{\delta_{ji}}{p_{ji}}\left\|\tilde{\dot{\psi}}(\acute{\bm{\beta}}^T\bm x_{ji})\bm x_{ji}\bm x_{ji}^T-{\dot{\psi}}(\tilde{\bm\beta}_j^T\bm x_{ji})\bm x_{ji}\bm x_{ji}^T\right\|_s\\
\le&n^{-1}\sum_{i=1}^n\frac{\delta_{ji}}{p_{ji}}\left\|\tilde{\dot{\psi}}(\acute{\bm{\beta}}^T\bm x_{ji})\bm x_{ji}\bm x_{ji}^T-{\dot{\psi}}(\tilde{\bm\beta}_j^T\bm x_{ji})\bm x_{ji}\bm x_{ji}^T\right\|_F\\
=&n^{-1}\sum_{i=1}^n\frac{\delta_{ji}}{p_{ji}}\left(\sum_{k=1}^d\left\|\frac{\partial\psi(\bm\beta_{(k)}^T\bm x_{ji})}{\partial\beta_k}\bm x_{ji}^T-\frac{\partial\psi(\tilde{\bm\beta}_j^T\bm x_{ji})}{\partial\beta_k}\bm x_{ji}^T\right\|_s^2\right)^{1/2}\\
\le&n^{-1}\sqrt{d}\sum_{i=1}^n\frac{\delta_{ji}}{p_{ji}}m_2(\bm x_{ji})\|\breve{\bm\beta}-\tilde{\bm\beta}_j\|.
\end{align*}
If the following events happens
\[\Xi_{6,m}=\left\{n^{-1}\sqrt{d}\sum_{i=1}^n\frac{\delta_{ji}}{p_{ji}}m_2(\bm x_{ji})\|\breve{\bm\beta}-\tilde{\bm\beta}_j\|\le 2U\right\}. \]
We have
$\|R_j\|\le 2U\|\breve{\bm\beta}-\tilde{\bm\beta}_j\|$.
Using the fact $Q^*(\breve{\bm\beta})=0$ and $Q_j^*(\tilde{\bm\beta}_j)=0$, it holds that $\breve{\bm\beta}-\tilde{\bm\beta}_{Kr}=(\sum_{j=1}^K\dot{Q}_j^*(\tilde{\bm\beta}_j))^{-1}\sum_{j=1}^KR_j^*$ by taking the summation on both side of (\ref{eq:ax31}) and definition of $\tilde{\bm\beta}_{Kr}$.
Thus for all $\omega\in\Phi_{N,K,\Delta}$, if $\Xi_{1,\lambda}$ defined in Lemma \ref{lem:local} and $\Xi_{6,m}$  happen on all machine $\mathcal{F}_{Nj}$'s $(j=1,\ldots,K)$, the following inequalities hold,
\begin{equation}\label{eq:ax-lem3-8}
\begin{split}
\|\breve{\bm\beta}-\tilde{\bm\beta}_{Kr}\|&\le \left\|\left(\frac{1}{K}\sum_{j=1}^K\dot{Q}_j^*(\tilde{\bm\beta}_j)\right)^{-1}\right\|_s\left\|\frac{1}{K}\sum_{j=1}^KR_j^*\right\|\\
&\le \frac{4U}{KC_{\psi}}\sum_{j=1}^K\|\breve{\bm\beta}-\tilde{\bm\beta}_{j}\|
\le C_1r^{-2\alpha}\Delta,
\end{split}
\end{equation}
where $C_1:=4U/C_{\psi}$ and the second last inequality due to the facts ${K}^{-1}\sum_{j=1}^K\dot{Q}_j^*(\tilde{\bm\beta}_j))\ge 0.5C_{\psi}I_d$ if $\Xi_{1,\lambda}$ defined in Lemma \ref{lem:local} happens on all machines, %
 as $n\rightarrow\infty$. %

{\purple Define $\Xi_{5,n,j,\Delta}=\{r^{\alpha}\|\tilde{\bm\beta}_j-\bm\beta_t\|\le \Delta\}$, $\Xi_{5,N,\Delta}=\{(Kr)^{\alpha}\|\breve{\bm\beta}-\bm\beta_t\|\le \Delta\}$ and $\Phi_{N,K,\Delta}=\cap_{j=1}^K \Xi_{5,n,j,\Delta}\cap \Xi_{5,N,\Delta}$.}
Then for any %
 $K=O(r^\eta)$ with $\eta<\min(1-2\alpha,4\alpha-1)$, we have $\Phi_{N,K,\Delta}\subset\{\sqrt{Kr}\|\tilde{\bm\beta}-\tilde{\bm\beta}_{Kr}\|\le\Delta_1\}$ for any $\Delta_1:=C_1\Delta$, when  $\Xi_{1,\lambda}$  happens and $r$ is large enough, since $\sqrt{Kr}\|\breve{\bm\beta}-\tilde{\bm\beta}_{Kr}\|\le \sqrt{Kr} r^{-2\alpha}\Delta_1=O(r^{(1+\eta-4\alpha)/2})\Delta_1$.

Recall that for $\omega\in \Phi_{N,K,\Delta} $, $\|\tilde{\bm\beta}_{j}-\bm\beta_t\|\le r^{-1/2}\Delta$ and $\|\breve{\bm\beta}-\bm\beta_t\|\le (Kr)^{-1/2}\Delta$.
Thus $\|\breve{\bm\beta}-\tilde{\bm\beta}_j\|\le r^{-1/2}\Delta+(Kr)^{-1/2}\Delta.$
Combine the above results, it holds that
\begin{equation}\label{eq:ax-lem3-7}
\begin{split}
&\pr(\Xi_{6,m}^c\cap\Phi_{N,K,\Delta})= E\mathbbm{1}_{\{\Xi_{6,m}^c ~\text{and}~ \Phi_{N,K,\Delta} ~\text{happen}\}}\\
\le&(2U)^{-1}{ E}\left\{\left|n^{-1}\sqrt{d}\sum_{i=1}^n\frac{\delta_{ji}}{p_{ji}}m_2(\bm x_{ji})\|\breve{\bm\beta}-\tilde{\bm\beta}_j\|\right|\mathbbm{1}_{\Phi_{N,K,\Delta}~\text{happens}}\right\}\\
\le& (2U)^{-1}\E\left\{\left(\frac{\Delta}{r^{1/2}}+\frac{\Delta}{(Kr)^{1/2}}\right)\left(\frac{\sqrt{d}}{n}\sum_{i=1}^n\frac{\delta_{ji}}{p_{ji}}m_2(\bm x_{ji})\right) \mathbbm{1}_{\{\Phi_{N,K,\Delta} \text{happens}\}}\right\}\\
\le& (2U)^{-1}\sqrt{d}C_m\left(\frac{\Delta}{r^{1/2}}+\frac{\Delta}{(Kr)^{1/2}}\right)
=O(r^{-1/2}).
\end{split}
\end{equation}

{By (\ref{eq:ax-lem3-8}), (\ref{eq:ax-lem3-7}) and the results in Lemma \ref{lem:local}}, as $r \rightarrow \infty$, for any $\Delta_1>0$, it follows that
\begin{small}
\begin{align*}
 &\pr(\sqrt{Kr}\|\breve{\bm\beta}-\tilde{\bm\beta}_{Kr}\|>\Delta_1)\\
 =&\pr(\sqrt{Kr}\|\breve{\bm\beta}-\tilde{\bm\beta}_{Kr}\|>\Delta_1,\Phi_{N,K,\Delta} ~  \text{happens},\  \Xi_{1,\lambda}\  \text{and}\   \Xi_{6,m}  \ \text{happen on all Machines})\\
+& \pr(\sqrt{Kr}\|\breve{\bm\beta}-\tilde{\bm\beta}_{Kr}\|>\Delta_1, \Phi_{N,K,\Delta}~  \text{happens, all}\ \Xi_{1,\lambda}\ \text{happen and some}\
  \Xi_{6,m}\  \text{not happen})\\
+& \pr(\sqrt{Kr}\|\breve{\bm\beta}-\tilde{\bm\beta}_{Kr}\|>\Delta_1, \Phi_{N,K,\Delta}~ \text{not happens or some}\  \Xi_{1,\lambda}\  \text{not happen}) \\
 \le & \pr(\sqrt{Kr}\|\breve{\bm\beta}-\tilde{\bm\beta}_{Kr}\|>\Delta_1, \Phi_{N,K,\Delta} ~\text{happens, }\  \Xi_{1,\lambda}\  \text{and}\   \Xi_{6,m}  \ \text{happen on all Machines})\\
+&\pr(\Phi_{N,K,\Delta} ~ \text{happens, }~  \Xi_{6,m}~ \text{not happen on some of the Machines})\\
+&\pr(\Phi_{N,K,\Delta}^c)+P(\Xi_{1,\lambda}\   \ \text{not happen on some of the Machines})\\
 \le & \pr(\sqrt{Kr}\|\breve{\bm\beta}-\tilde{\bm\beta}_{Kr}\|>\Delta_1, \Phi_{N,K,\Delta} ~\text{happens,}\  \Xi_{1,\lambda}\  \text{and}\   \Xi_{6,m}  \ \text{happen on all Machines})\\
+&\pr(\Phi_{N,K,\Delta} ~ \text{happens, }  \Xi_{6,m}~ \text{not happen on some of the Machines})+\pr(\Phi_{N,K,\Delta}^c)\\
+&\pr(\Xi_{1,\lambda}\   \text{not happen on some of the Machines})\\
=& \sqrt{K}O(r^{1/2-2\alpha})+ KO(r^{-1}) + O({(Kr)}^{2\alpha-1})+KO({r}^{2\alpha-1}) + KO(r^{-1/2})=o(1).
 \end{align*}
\end{small}

\end{proof}

Now we prove Theorem~7.
\begin{proof}

From the definition of $\breve{\bm\beta}$ in Lemma \ref{lem:comb} and Theorem 3.3 in \cite{xiong2008some}, {it can be shown that} %
\[  \breve{\bm\beta}-\hat{\bm\beta}_{\rm QLE}=O_{P|\mathcal{F}_N}((Kr)^{-1/2}).\]
Combining this result with Lemma \ref{lem:comb}, the desired result holds by fetching $\alpha=1/3$.

\end{proof}

{Based on Lemma \ref{lem:comb} and Theorem 6, the proof of Theorem~8 is in the following}.
\begin{proof}

From the definition of $\tilde{\bm\beta}_{Kr}$, we known that  $\tilde{\bm\beta}_{Kr}-\breve{\bm\beta}=o_P((Kr)^{-1})$, where $\breve{\bm\beta}$ is calculated from Algorithm 2 by using the pooling subsamples.  Combine this result with Theorem 6, the desired result follows immediately from the Slutsky's Theorem.

\end{proof}

\blue
\section{Additional Simulation Results}
To further evaluate our proposed methods, we consider the following additional cases with a lager $N=5,000,000$.

\begin{enumerate}[{Case S}1:]
\item The true value of $\bm\beta$ is a $d\times1$ vector of 0.5 with $d=7$, and  $\bm x$ follows a multivariate normal distribution $N(\bm {0.15}, I_d)$, where $I_d$ is the  $d\times d$  identity matrix.

\item The true value of $\bm\beta$ is a $d\times1$ vector of 0.5 with $d=7$, and  $\bm x$ follows a multivariate normal distribution $N(\bm {0.15}, \Sigma)$, where $\Sigma$ is the  $d\times d$  matrix with the $(i,j)$-th entry being $0.5^{|i-j|}$.

\item The true value of $\bm\beta$ is a $d\times1$ vector of 0.5 with $d=7$, and $\bm x$ follows a multivariate $t$ distribution with degrees of freedom 9, $\bm x\sim t_9(\bm {0.15}, I_d)/10$, where the $I_d$ is defined in Case S1. For this set up, only part of the Assumptions in Sections 2.2, 3.2 and 4 are satisfied.

\item The true value of $\bm\beta$ is a $d\times1$ vector with $d=35$, whose first ten elements are 0.5, the last five elements are -0.1, and the rest are 0.2. In this case, $\bm x$ follows a multivariate normal distribution with $N(\bm \mu, I_d)$, where $\bm\mu$ is a $35\times1$ vector whose first seven elements are 0.15 and rest are zeros.

\end{enumerate}

To be aligned with the the settings described in Section 5 of the main paper, we also demonstrate our methods on Poisson distribution with mean ${E}(y|\bm {x} )=\exp(\bm\beta^T\bm{x})$.
It {\bluee is} worth mentioning that in Cases S1, S2, and S4, all the Assumptions 1--4 are fully satisfied. However, in Case S3, the conditions in (iii) and (iv) of Assumption 2 are not satisfied.

Here, we fix $r_0=400$ and $\varrho=0.2$.
We choose the sample size $r$  to be 500, 700, 1000, 1200, 1500, 1700 and 2000. Since we will consider a case where the full data QLE cannot be calculated, we use the true value of $\bm\beta$ to calculate the MSE in this section, i.e., we calculate the MSE from
\[{\blue \text{MSE}=T^{-1}\sum_{t=1}^T\|{\bm\beta}_{\bm p}^{(t)}-{\bm\beta}\|^2},\]
where ${\bm\beta}_{\bm{p}}^{(t)}$ is the estimate from the $t$-th subsample with subsampling probability $\bm p$ and $\bm\beta$ is the true parameter. The results are reported  in Figure \ref{fig:additional_distribution}.

\begin{figure}[!htp]
  \centering
  \begin{subfigure}{0.45\textwidth}
    \includegraphics[width=\textwidth]{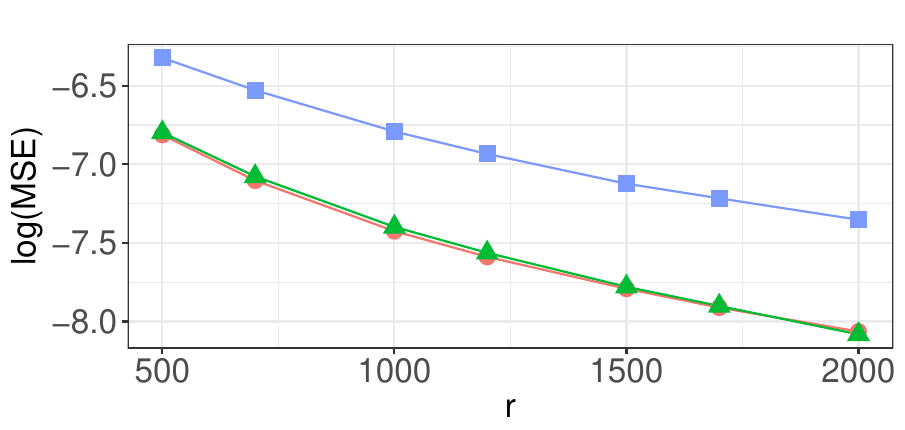}\\[-1.2cm]
    \caption{Case S1 (K=1)}
  \end{subfigure}
  \begin{subfigure}{0.45\textwidth}
    \includegraphics[width=\textwidth]{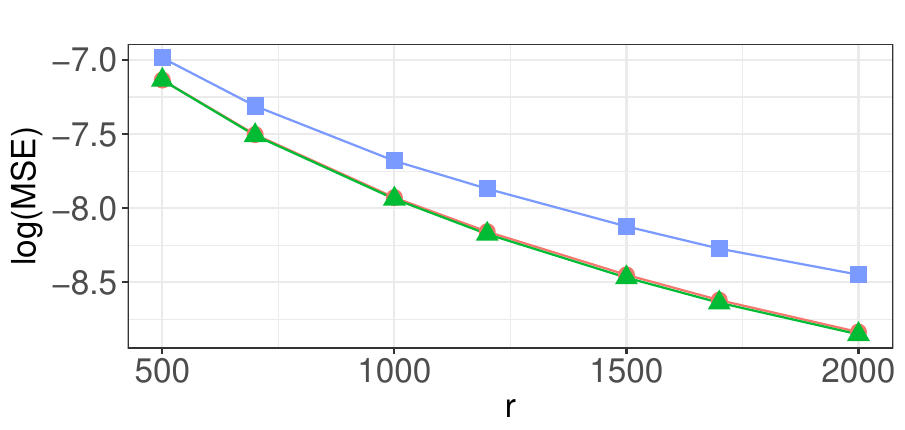}\\[-1.2cm]
    \caption{Case S1 (K=5)}
  \end{subfigure}\\[-0.5mm]
  \begin{subfigure}{0.45\textwidth}
    \includegraphics[width=\textwidth]{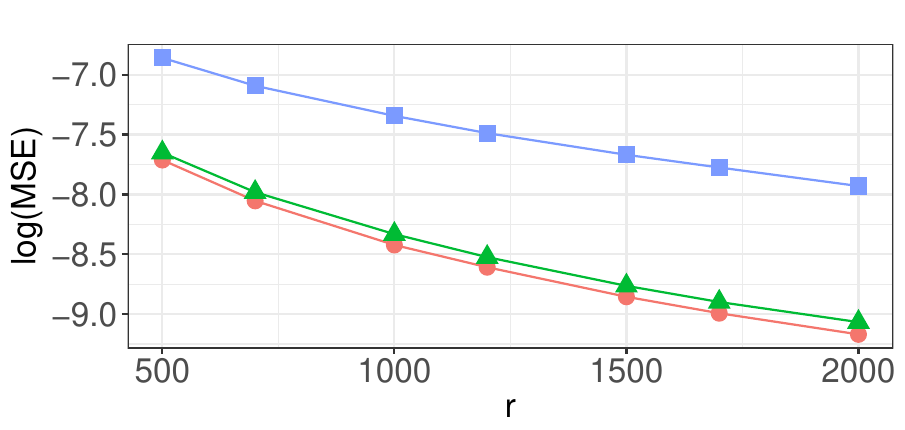}\\[-1.2cm]
    \caption{Case S2 (K=1)}
  \end{subfigure}
  \begin{subfigure}{0.45\textwidth}
    \includegraphics[width=\textwidth]{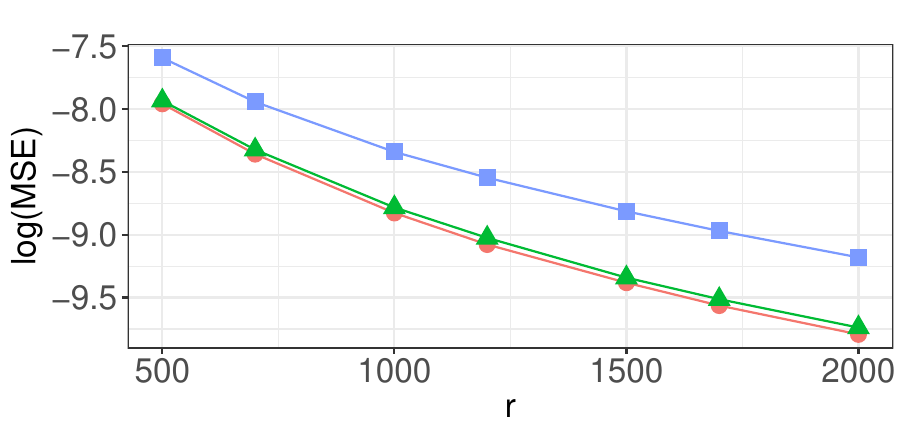}\\[-1.2cm]
    \caption{Case S2 (K=5)}
  \end{subfigure}\\[-0.5mm]
   \begin{subfigure}{0.45\textwidth}
    \includegraphics[width=\textwidth]{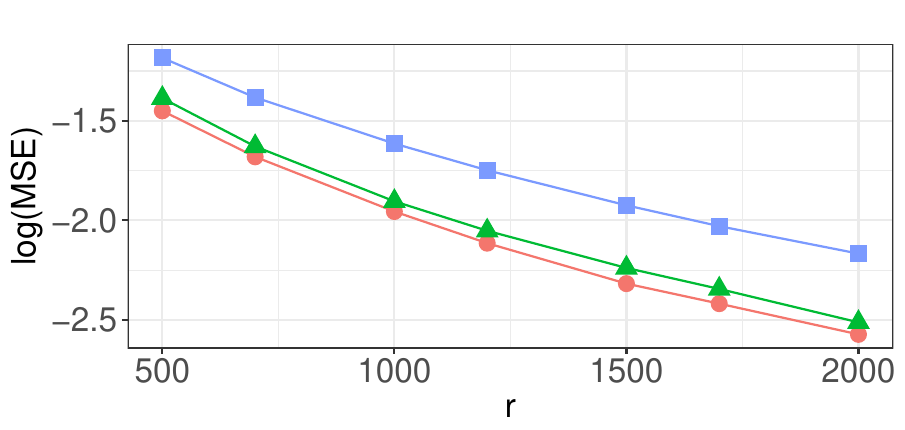}\\[-1.2cm]
    \caption{Case S3 (K=1)}
  \end{subfigure}
  \begin{subfigure}{0.45\textwidth}
    \includegraphics[width=\textwidth]{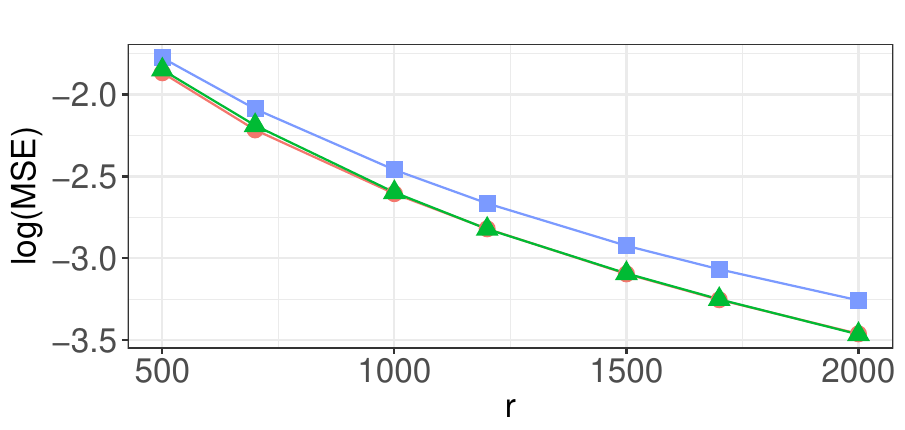}\\[-1cm]
    \caption{Case S3 (K=5)}
  \end{subfigure}\\[-0.5mm]
  \begin{subfigure}{0.45\textwidth}
    \includegraphics[width=\textwidth]{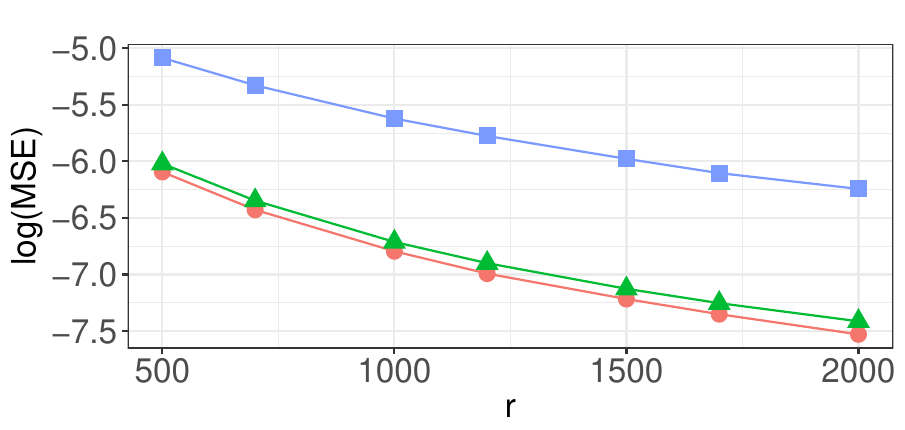}\\[-1.2cm]
    \caption{Case S4 (K=1)}
  \end{subfigure}
  \begin{subfigure}{0.45\textwidth}
    \includegraphics[width=\textwidth]{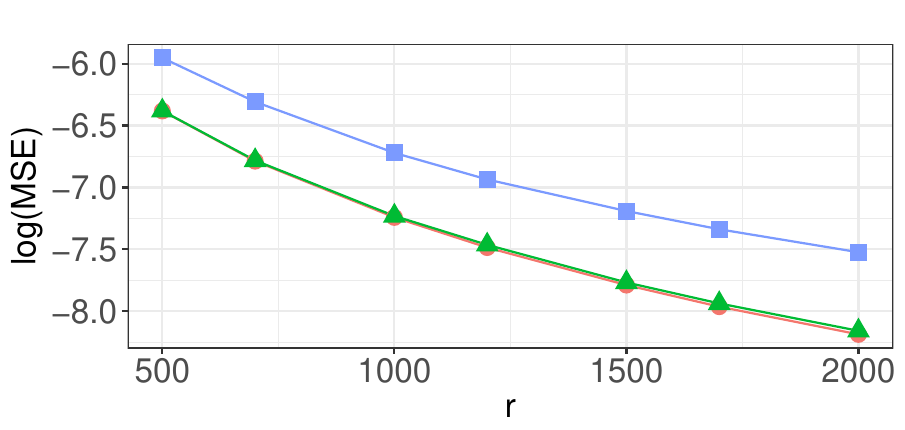}\\[-1.2cm]
    \caption{Case S4 (K=5)}
  \end{subfigure}\\ [-0.5mm]
  \caption{A graph showing the log of MSE with different ${r}$ and $K$ for different distributions of covariates based on MV (red circle), MVc (green triangle) and uniform subsampling (blue square) methods where $r_0=400$ and $\varrho=0.2$.}
  \label{fig:additional_distribution}
\end{figure}

In Figure \ref{fig:additional_distribution}, the relative patterns among different methods are similar to Cases 1--4 in the main paper for all the five additional datasets. Furthermore, the results in Case S3 indicate
that our methods have higher estimation efficiency than the uniform subsampling method even when some regularity assumptions required in the theoretical investigations are not satisfied.

Now we evaluate the computational efficiency of the subsampling
strategies. We implemented all methods using the R programming language and recorded the computing times of the three subsampling strategies (uniform, MV, MVc) using \verb"Sys.time()" function. {\red To simulate the distributed computing environments, we used the \verb"foreach" and \verb"doSNOW" packages in \verb"R".}
Computations {\bluee were} carried out on a desktop computer with Window 10 platform and an Intel I7 processor. Each subsampling strategy {\blue was} repeated for 50 times.  Figure~\ref{fig:time} shows the results for the four cases listed above with different $r$ and $k$, and a fixed $r_0=400$. The computing time for using the full data
is also given  for comparisons.

\begin{figure}[!htp]
  \centering
  \begin{subfigure}{0.45\textwidth}
    \includegraphics[width=\textwidth]{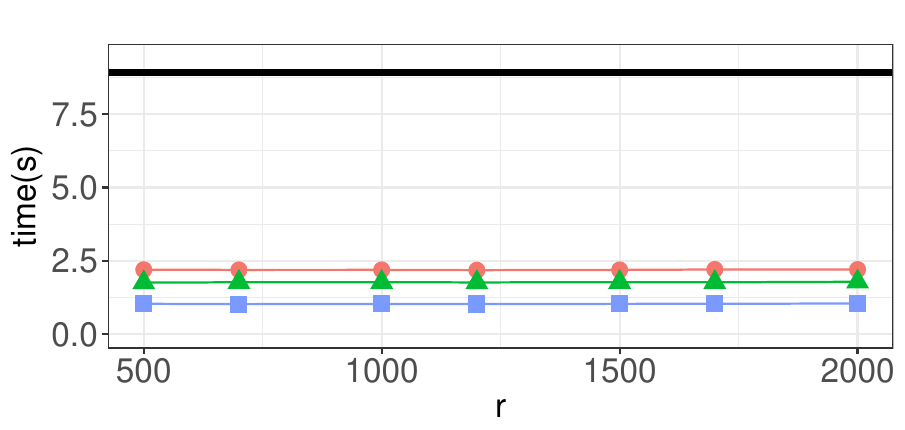}\\[-1.2cm]
    \caption{Case S1 (K=1)}
  \end{subfigure}
  \begin{subfigure}{0.45\textwidth}
    \includegraphics[width=\textwidth]{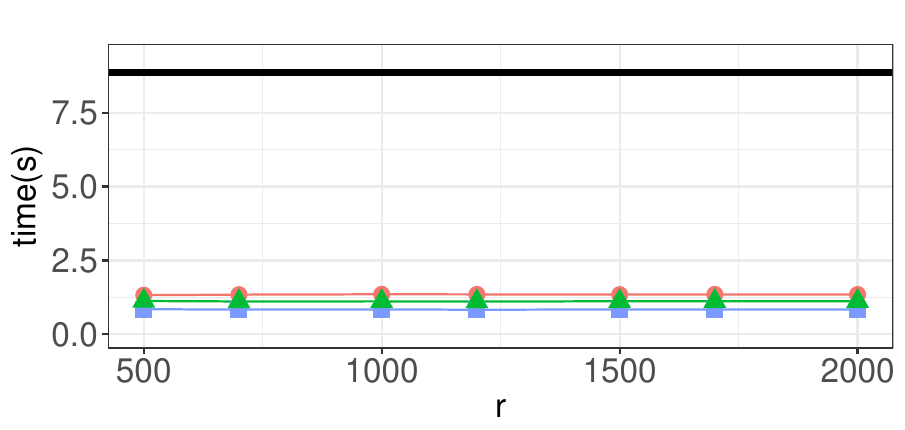}\\[-1.2cm]
    \caption{Case S1 (K=5)}
  \end{subfigure}\\[-0.5mm]
  \begin{subfigure}{0.45\textwidth}
    \includegraphics[width=\textwidth]{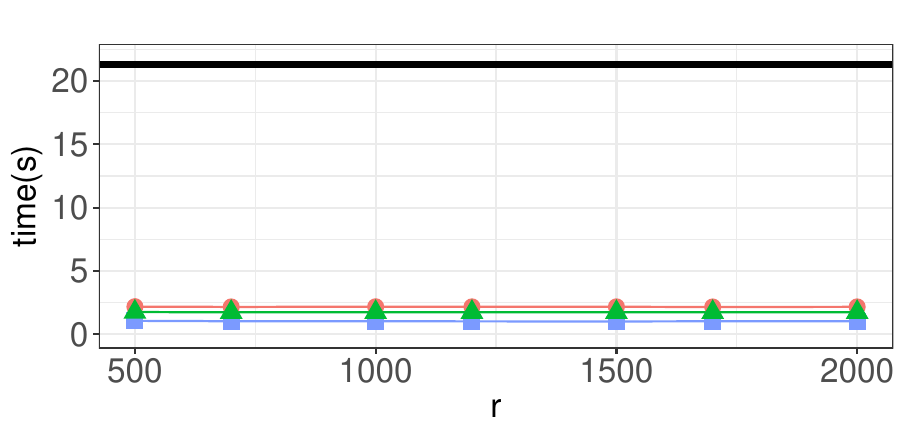}\\[-1.2cm]
    \caption{Case S2 (K=1)}
  \end{subfigure}
  \begin{subfigure}{0.45\textwidth}
    \includegraphics[width=\textwidth]{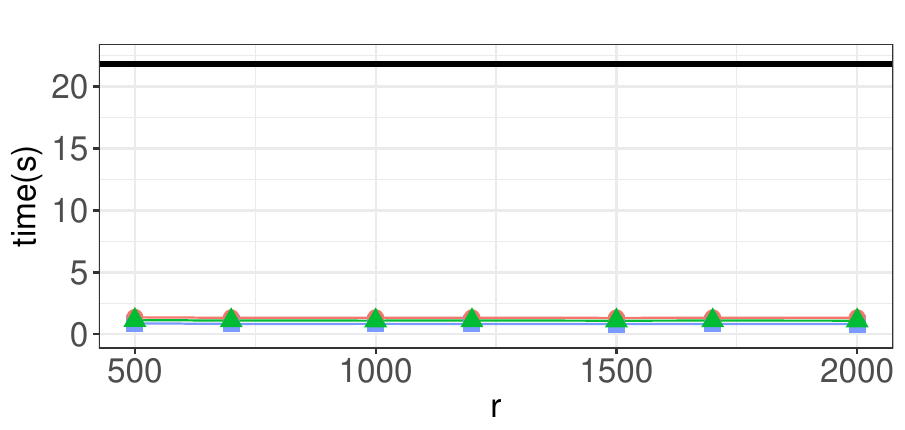}\\[-1.2cm]
    \caption{Case S2 (K=5)}
  \end{subfigure}\\[-0.5mm]
   \begin{subfigure}{0.45\textwidth}
    \includegraphics[width=\textwidth]{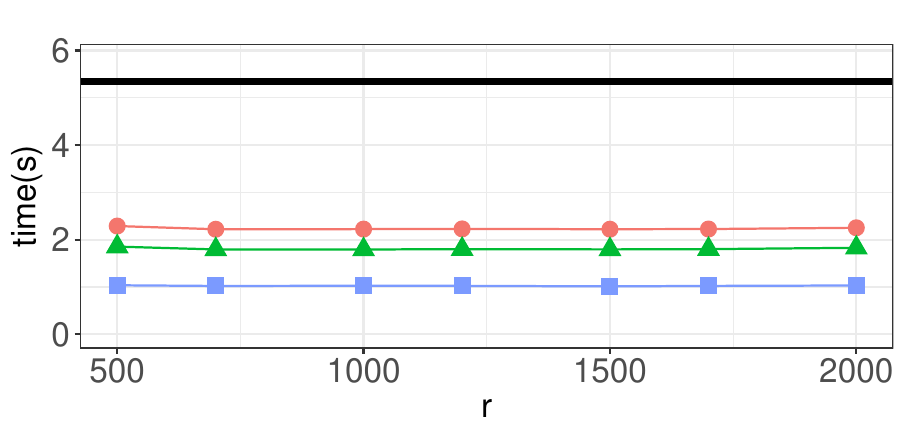}\\[-1.2cm]
    \caption{Case S3 (K=1)}
  \end{subfigure}
  \begin{subfigure}{0.45\textwidth}
    \includegraphics[width=\textwidth]{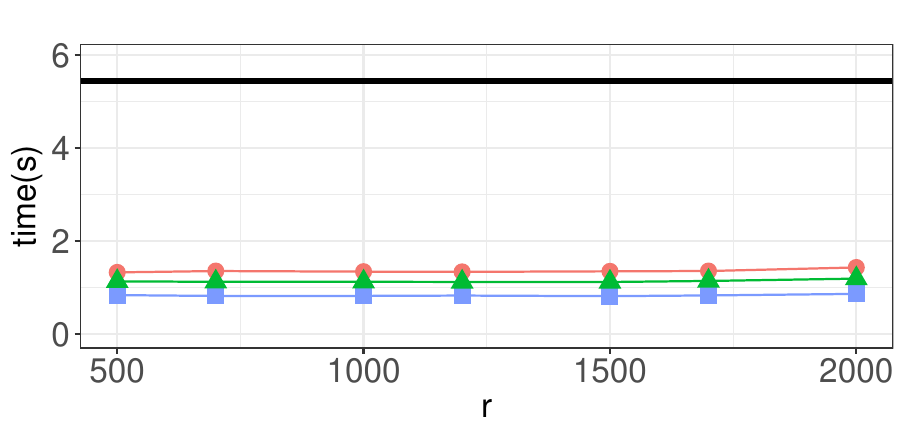}\\[-1cm]
    \caption{Case S3 (K=5)}
  \end{subfigure}\\[-0.5mm]
  \begin{subfigure}{0.45\textwidth}
    \includegraphics[width=\textwidth]{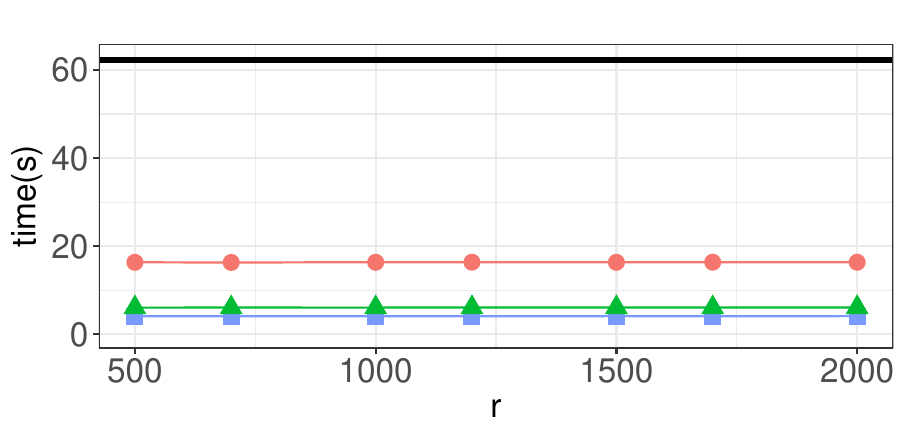}\\[-1.2cm]
    \caption{Case S4 (K=1)}
  \end{subfigure}
  \begin{subfigure}{0.45\textwidth}
    \includegraphics[width=\textwidth]{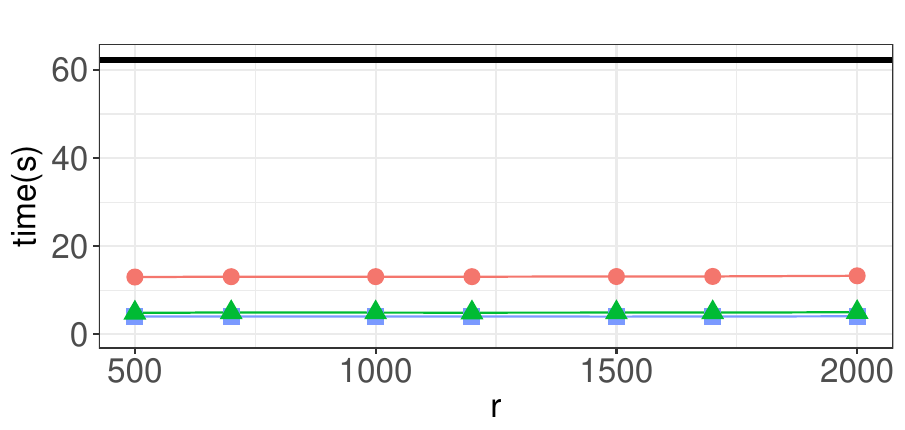}\\[-1.2cm]
    \caption{Case S4 (K=5)}
  \end{subfigure}\\ [-0.5mm]
  \caption{A graph showing the computational time (in seconds) with different ${r}$ and $K$ for different distributions of covariates based on MV (red circle), MVc (green triangle) and uniform subsampling (blue square) methods where $r_0=400$ and $\varrho=0.2$. The black solid line in the picture stands for the computing time for the full sample QLE.}
  \label{fig:time}
\end{figure}

Figure \ref{fig:time} reveals that the computation time is not very sensitive to the subsample size. All the subsampling algorithms take significantly less computing time compared with the full data QLE.
This agrees with the fact that the computational complexity for the full data QLE with Newton-Raphson method is $O(\zeta N d^2)$ while subsampling methods do not need to perform iterative calculations on the full data. Here $\zeta$ stands for the number of iterations in the Newton-Raphson method.
 In most cases, the MVc method requires significantly less
computational time compared with the MV method, especially for the last case. %
This is because the MV method requires $O(Nd^2 + \zeta rd^2)$ time and the MVc methods requires $O(Nd + \zeta rd^2)$. If $N>\zeta rd$, then these time complexities are $O(Nd^2)$ and $O(Nd)$, for the MV and MVc methods, respectively.
The uniform sampling always takes the least computation time, since it does not involve the computation of the sampling probability. Its time complexities is $O(N + \zeta rd^2)$, which reduces to $O(N)$ if $N>\zeta rd^2$.

\red To further investigate the case that full data QLE are infeasible to calculate due to limited RAM, we consider the scenario that the data is loaded into the RAM block-by-bock. For this, we consider the following Case S5 with $N=5,000,000$ and $d=140$. This is also a case with higher dimension.

\begin{enumerate}[{Case S}1:]
\setcounter{enumi}{4}
\item The true value of $\bm\beta$ is a $d\times1$ vector with $d=140$ whose first 35 elements are the same as $\bm\beta$ in Case S4 and rest are zeros. Here, $\bm x$ follows a multivariate normal distribution with $N(\bm \mu, I_d)$, where $\bm\mu$ is a $140\times1$ vector whose first seven elements are 0.15 and rest are zeros. %
\end{enumerate}

In this case, we generate and store the full data in five files, with each file about $\sim$2.4 GB containing 1,000,000 observations. In this setup, we obtain the pilot estimators for each block independently and combine the five pilot estimators with the five optimal estimators using (24). The results about MSE and computing time are in Figures S3(a) and S3(b), respectively.

\begin{figure}[!htp]
  \centering
  \begin{subfigure}{0.45\textwidth}
    \includegraphics[width=\textwidth]{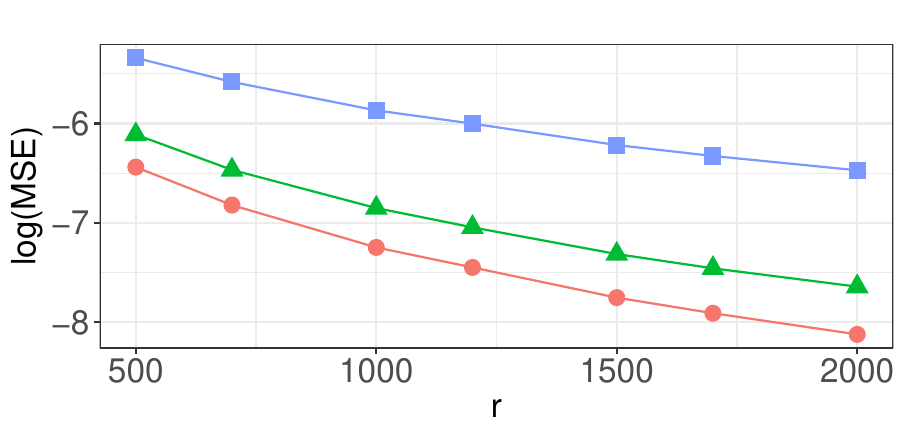}\\[-1cm]
    \caption{Log MSE}\label{fig:block1}
  \end{subfigure}
  \begin{subfigure}{0.45\textwidth}
    \includegraphics[width=\textwidth]{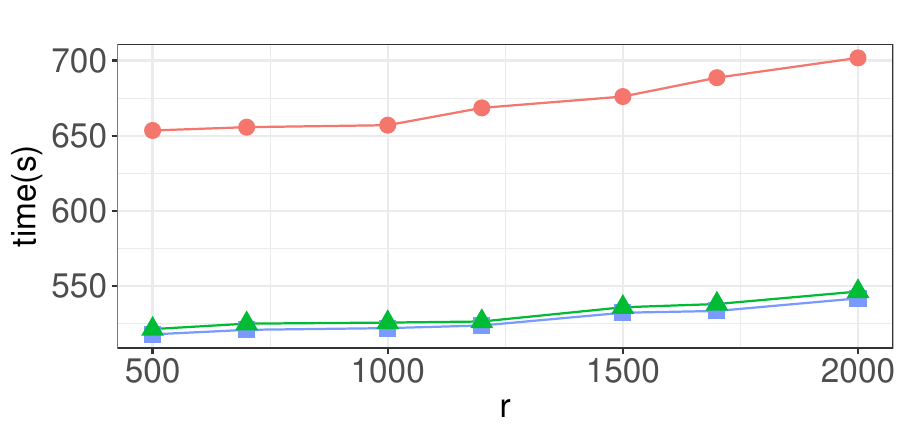}\\[-1cm]
    \caption{Times (K=5)}\label{fig:block2}
  \end{subfigure}\\[-1mm]
  \caption{{\bluee A graph showing the log MSE and computational time (in seconds) with different ${r}$ for Case S5 based on MV (red circle), MVc (green triangle) and uniform subsampling (blue square) methods where $r_0=400$ and $\varrho=0.2$.} }
  \label{fig:block}
\end{figure}

From Figure \ref{fig:block}, it is clear to see that the relative performance among the three methods are the same as in Figures \ref{fig:additional_distribution} and \ref{fig:time}. %
In Figure S3(b), the times for loading all data files are also recorded. Note that communication between CPU and hard drive is much slower than communication between CPU and RAM. Thus the computational time in Figure S3(b) are significantly longer than those observed in Figure~\ref{fig:time}. We also see that MVc is much faster than MV in this case as the dimension $d$ is larger.

In order to show the trade-off more intuitively between estimation efficiency and computation cost, we plot the log MSEs against the used CPU times for Case S5.  %
The results are presented in Figure~\ref{fig:S4}.
\begin{figure}[!htp]
  \centering
    \includegraphics[width=\textwidth]{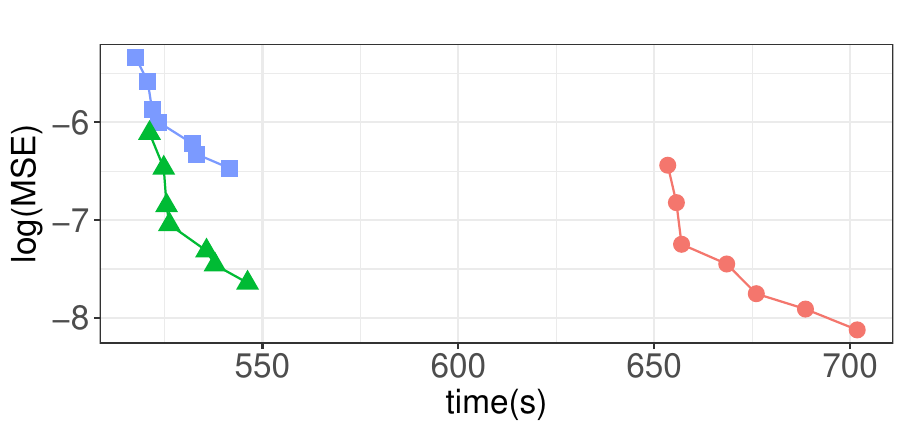}\\[-0.5cm]
  \caption{A graph showing the computational time (in seconds) and log MSE with different ${r}$ for Case S5 based on MV (red circle), MVc (green triangle) and uniform subsampling (blue square) methods where $r_0=400, K=5$ and $\varrho=0.2$. }%
  \label{fig:S4}
\end{figure}

From Figure~\ref{fig:S4}, %
the MVc approach produce a smaller MSE compared with the uniform subsampling method with the same CPU time, and this advantage becomes more evident with the increase of the used CPU time. For the MV approach, although it produces a smaller MSE compared with the uniform subsampling method and the MVc approach with the same subsample size, it requires a longer time to achieve the same level of accuracy. Note that this does not mean the MVc approach is always better than the MV approach. For example, if the available memory only allow the analysis of a subsample of size $r$ while the computational time is relatively cheap, then the MV approach may be preferable as it often results in more informative subsamples.

To evaluate the asymptotic normality visually, we create histograms for parameter estimates from the $T=1000$ repetitions of the simulation. Figure \ref{fig:R1} reports the results for parameter $\beta_1$ in Case S1 when $r=2000$. The red curves are the normal density functions with the same means and standard deviations for the $1000$ subsample estimates. From Figure \ref{fig:R1}, we see that the distribution of the subsampling estimators are very close to normal distributions. The results for other parameters and other cases and thus are omitted.

\begin{figure}[H]
  \centering
    \begin{subfigure}{0.3\textwidth}
    \includegraphics[width=\textwidth]{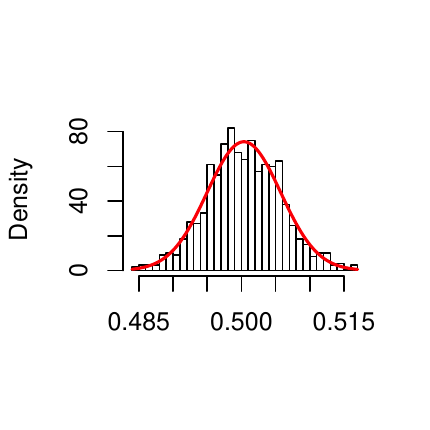}\\[-1.5cm]
    \caption{UNIF }
  \end{subfigure}
  \begin{subfigure}{0.3\textwidth}
    \includegraphics[width=\textwidth]{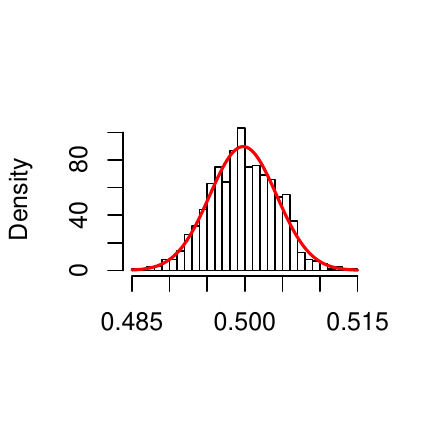}\\[-1.5cm]
    \caption{MV}
  \end{subfigure}
  \begin{subfigure}{0.3\textwidth}
    \includegraphics[width=\textwidth]{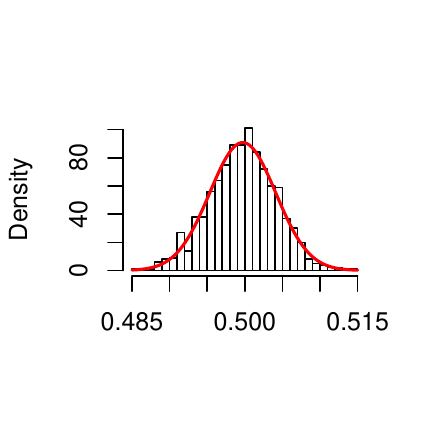}\\[-1.5cm]
    \caption{MVc}
  \end{subfigure}\\[-0.5mm]
  \caption{Empirical distribution of $\beta_1$ for Case S1 with $r=2000$, $r_0=400$ and $\varrho=0.2$ based on MV, MVc and uniform subsampling methods. The red curve is the estimated normal density through the method of moments.}
  \label{fig:R1}
\end{figure}

\color{black}

\bibliographystyle{jasa}

\bibliography{report}
\end{document}